\def\ARXIV{}
\ifdefined\ARXIV
  \def\EXTENDED{}
  \documentclass[acmsmall,nonacm,screen]{acmart}
\else
  \documentclass[acmsmall,screen]{acmart}
\fi

\usepackage{microtype}
\usepackage[utf8]{inputenc}            
\usepackage[T1]{fontenc}               
\usepackage{mathpartir}                
\usepackage{mathwidth}                 
\usepackage{stmaryrd}                  
\usepackage{xspace}                    
\usepackage{hyperref}                  
\usepackage[customcolors,shade]{hf-tikz}
\hfsetfillcolor{gray!40}
\hfsetbordercolor{gray!40}
\usepackage{float}
\usepackage{multirow}

\usepackage{booktabs}   
\usepackage{subcaption} 

\newcommand\invisiblesection[1]{%
  \refstepcounter{section}%
  \sectionmark{#1}}

\newcommand{\extendedsection}[1]{
  \ifdefined\EXTENDED
    \section{#1}
  \else
    \invisiblesection{#1}
  \fi
}
  {%
    \ifdefined\EXTENDED
    \else
      \setbox0\vbox\bgroup
    \fi
  }{%
    \ifdefined\EXTENDED
    \else
      \egroup
    \fi
  }%

\newcommand*{\extendedref}[1]{%
  {\ifdefined\EXTENDED\ref{#1}\else\protect\NoHyper\ref{#1}\protect\endNoHyper\fi}%
}

\ifx\ARXIV\undefined
  \setcopyright{rightsretained}
  \acmPrice{}
  \acmDOI{10.1145/3408982}
  \acmYear{2020}
  \copyrightyear{2020}
  \acmSubmissionID{icfp20main-p71-p}
  \acmJournal{PACMPL}
  \acmVolume{4}
  \acmNumber{ICFP}
  \acmArticle{100}
  \acmMonth{8}
\fi

\newcommand{\naive}{naïve\xspace}
\newcommand{\naively}{naïvely\xspace}
\newcommand{\Naive}{Naïve\xspace}

\newcommand{\defas}[0]{\mathrel{:=}} 

\newcommand{\N}{\ensuremath{\mathbb{N}}}
\newcommand{\B}{\ensuremath{\mathbb{B}}}

\newcommand{\Iff}{\Leftrightarrow}
\newcommand{\Implies}{\Rightarrow}

\newcommand{\BCalc}{\ensuremath{\lambda_{\textrm{\normalfont b}}\xspace}}
\newcommand{\BCalcS}{\ensuremath{\lambda_{\textrm{\normalfont s}}\xspace}}
\newcommand{\BCalcE}{\ensuremath{\lambda_{\textrm{\normalfont e}}\xspace}}
\newcommand{\BCalcSE}{\ensuremath{\lambda_{\textrm{\normalfont se}}\xspace}}
\newcommand{\HCalc}{\ensuremath{\lambda_{\textrm{\normalfont h}}\xspace}}
\newcommand{\EC}{\ensuremath{\mathcal{E}}}
\newcommand{\HC}{\ensuremath{\mathcal{H}}}

\newcommand{\slab}[1]{\textrm{#1}}
\newcommand{\semlab}[1]{\text{\scshape{S-#1}}}
\newcommand{\tylab}[1]{\text{\scshape{T-#1}}}
\newcommand{\mlab}[1]{\text{\scshape{M-#1}}}

\newcommand{\revto}{\ensuremath{\leftarrow}}

\newcommand{\dec}[1]{\mathsf{#1}}
\newcommand{\keyw}[1]{\mathbf{#1}}
\newcommand{\Handle}{\keyw{handle}}
\newcommand{\With}{\keyw{with}}
\newcommand{\Let}{\keyw{let}}
\newcommand{\Rec}{\keyw{rec}}
\newcommand{\In}{\keyw{in}}
\newcommand{\Do}{\keyw{do}}
\newcommand{\Return}{\keyw{return}}
\newcommand{\Val}{\keyw{val}}
\newcommand{\Case}{\keyw{case}}
\newcommand{\If}{\keyw{if}}
\newcommand{\Then}{\keyw{then}}
\newcommand{\Else}{\keyw{else}}

\newcommand{\Record}[1]{\ensuremath{\langle #1 \rangle}}
\newcommand{\Unit}{\Record{}}
\newcommand{\Inl}{\keyw{inl}}
\newcommand{\Inr}{\keyw{inr}}

\newcommand{\Superpoint}{\lambda\_.\Do\;\Branch~\Unit}

\newcommand{\One}{\mathsf{Unit}}

\newcommand{\Bool}{\mathsf{Bool}}
\newcommand{\List}{\mathsf{List}}
\newcommand{\Nat}{\mathsf{Nat}}

\newcommand{\Count}{\dec{count}}
\newcommand{\ECount}{\dec{effcount}}

\newcommand{\Predicate}{\dec{Predicate}}
\newcommand{\Point}{\dec{Point}}
\newcommand{\Branch}{\dec{Branch}}

\newcommand{\Countprog}{K}

\newcommand{\True}{\mathsf{true}}
\newcommand{\False}{\mathsf{false}}

\newcommand{\cekl}{\langle}
\newcommand{\cekr}{\rangle}
\newcommand{\cek}[1]{\ensuremath{\langle #1 \rangle}}

\newcommand{\Catname}[1]{\ensuremath{\mathrm{#1}}}
\newcommand{\TypeCat}{\Catname{Type}}
\newcommand{\CtxCat}{\Catname{Ctx}}
\newcommand{\CompCat}{\Catname{Comp}}
\newcommand{\ValCat}{\Catname{Val}}
\newcommand{\PureCont}{\mathsf{PureCont}}
\newcommand{\Cont}{\mathsf{Cont}}

\newcommand{\Addr}{\mathsf{Addr}}
\newcommand{\Lab}{\mathsf{Lab}}
\newcommand{\Env}{\mathsf{Env}}

\newcommand{\Ref}{\dec{Ref}}

\newcommand{\hret}{H^{\mathrm{val}}}

\newcommand{\hell}{H^{\ell}}

\newcommand{\conf}{\mathcal{C}}

\newcommand{\typ}[2]{#1 \vdash #2}
\newcommand{\typv}[2]{#1 \vdash #2}

\newcommand{\nil}{\ensuremath{[]}}
\newcommand{\cons}{\ensuremath{::}}

\newcommand{\concat}{\mathbin{+\!\!+}}
\newcommand{\snoc}[2]{\ensuremath{#1 \concat [#2]}}

\newcommand{\env}{\gamma}

\newcommand{\reducesto}[0]{\ensuremath{\leadsto}}
\newcommand{\stepsto}[0]{\ensuremath{\longrightarrow}}
\newcommand{\Stepsto}{\Longrightarrow}

\newcommand{\ba}{\begin{array}}
\newcommand{\ea}{\end{array}}

\newcommand{\bl}{\ba[t]{@{}l@{}}}
\newcommand{\el}{\ea}

\newenvironment{syntax}{\begin{displaymath}\ba{@{}l@{\quad}r@{~}c@{~}l@{}}}{\ea\end{displaymath}\ignorespacesafterend}
\newenvironment{reductions}{\begin{displaymath}\ba{@{}l@{\quad}@{}r@{~}c@{~}l@{}}}{\ea\end{displaymath}\ignorespacesafterend}
\newenvironment{equations}{\begin{displaymath}\ba{@{}r@{~}c@{~}l@{}}}{\ea\end{displaymath}\ignorespacesafterend}
\newenvironment{eqs}{\ba[t]{@{}r@{~}c@{~}l@{}}}{\ea}
\newenvironment{twoeqs}{\ba[t]{@{}r@{~}c@{~}l@{~}c@{~}r@{~}c@{~}l@{}}}{\ea}

\newenvironment{derivation}{\begin{displaymath}\ba{@{}r@{~}l@{}}}{\ea\end{displaymath}\ignorespacesafterend}
\newcommand{\reason}[1]{\quad (\text{#1})}

\newcommand{\pto}{\rightharpoonup}

\newcommand{\const}[1]{\ulcorner #1 \urcorner}

\newcommand{\val}[2]{\llbracket #1 \rrbracket #2}
\newcommand{\inv}[1]{\llparenthesis #1 \rrparenthesis}

\newcommand{\res}{\backslash}

\newcommand{\BigO}{\ensuremath{\mathcal{O}}}

\newcommand{\query}{\mathord{?}}
\newcommand{\ans}{\mathord{!}}
\newcommand{\labs}{\mathsf{labs}}
\newcommand{\steps}{\mathsf{steps}}

\newcommand{\tree}{\tau}
\newcommand{\tl}{\labs(\tree)}
\newcommand{\ts}{\steps(\tree)}

\newcommand{\tr}{\mathcal{T}}
\newcommand{\tru}{\mathcal{U}}

\usepackage{tikz}
\usetikzlibrary{arrows,matrix,shapes,shapes.geometric,decorations.pathmorphing,calc,positioning,fit,tikzmark}

\tikzset{
  port/.style = {treenode, font=\Huge, draw=white, minimum width=0.5em, minimum height=0.5em},
  blackbox/.style = {rectangle, fill=black, draw=black, minimum width=2cm, minimum height=2cm},
  treenode/.style = {align=center, inner sep=3pt, text centered},
  opnode/.style = {treenode, rectangle, draw=black},
  leaf/.style = {treenode, draw=black, ellipse, thin},
  comptree/.style = {treenode, draw=black, regular polygon, regular polygon sides=3},
  highlight/.style = {draw=red,very thick},
  pencildraw/.style={
    black!75,
    decorate,
    decoration={random steps,segment length=0.8pt,amplitude=0.1pt}
  },
  hbox/.style = {rectangle,draw=none, minimum width=6cm, minimum height=1cm},
  gbox/.style = {rectangle,draw=none,minimum width=2cm,minimum height=1cm},
  itria/.style={
    draw,dashed,shape border uses incircle,
    isosceles triangle,shape border rotate=90}
}

\newcommand{\tossTree}{
  \begin{tikzpicture}[->,>=stealth',level/.style={sibling distance = 2.5cm/##1,
      level distance = 1.0cm}]
\node (root) [opnode] {$\dec{Branch}$}
  child { node [leaf] {$\dec{Heads}$}
    edge from parent node[above left] {$\True$}
  }
  child { node [leaf] {$\dec{Tails}$}
    edge from parent node[above right] {$\False$}
  }
;
\end{tikzpicture}}

\begin{document}

\title{Effects for Efficiency}
\subtitle{Asymptotic Speedup with First-Class Control}

\author{Daniel Hillerström}
\affiliation{
  \institution{The University of Edinburgh}            
  \country{UK}
}
\email{daniel.hillerstrom@ed.ac.uk}          

\author{Sam Lindley}
\affiliation{
  \institution{The University of Edinburgh and Imperial College London and Heriot-Watt University}   \country{UK}                   
}
\email{sam.lindley@ed.ac.uk}         

\author{John Longley}
\affiliation{
  \institution{The University of Edinburgh}           
  \country{UK}                   
}
\email{jrl@staffmail.ed.ac.uk}         

\begin{abstract}
  We study the fundamental efficiency of delimited
  control. Specifically, we show that effect handlers enable an
  asymptotic improvement in runtime complexity for a certain class of
  functions. We consider the \emph{generic count} problem using a pure
  PCF-like base language $\BCalc$ and its extension with effect
  handlers $\HCalc$.
  We show that $\HCalc$ admits an asymptotically more efficient
  implementation of generic count than any $\BCalc$ implementation.
  We also show that this efficiency gap remains when $\BCalc$ is
  extended with mutable state.

  To our knowledge this result is the first of its kind for control
  operators.
\end{abstract}

\ifx\ARXIV\undefined


\fi

\maketitle

\section{Introduction}\label{sec:introduction}
\def\LLL{{\mathcal L}}
\def\N{{\mathbb N}}
In today's programming languages we find a wealth of powerful
constructs and features --- exceptions, higher-order store, dynamic
method dispatch, coroutines, explicit continuations, concurrency
features, Lisp-style `quote' and so on --- which may be present or
absent in various combinations in any given language.  There are of
course many important pragmatic and stylistic differences between
languages, but here we are concerned with whether languages may differ
more essentially in their expressive power, according to the selection
of features they contain.

One can interpret this question in various ways.  For instance,
\citet{Felleisen91} considers the question of whether a language
$\LLL$ admits a translation into a sublanguage $\LLL'$ in a way which
respects not only the behaviour of programs but also aspects of their
(global or local) syntactic structure. If the translation of some
$\LLL$-program into $\LLL'$ requires a complete global restructuring,
we may say that $\LLL'$ is in some way less expressive than $\LLL$.
In the present paper, however, we have in mind even more fundamental
expressivity differences that would not be bridged even if
whole-program translations were admitted. These fall under two
headings.
\begin{enumerate}
\item \emph{Computability}: Are there operations of a given type
  that are programmable in $\LLL$ but not expressible at all in $\LLL'$?
\item \emph{Complexity}: Are there operations programmable in $\LLL$
  with some asymptotic runtime bound (e.g.\ `$\BigO(n^2)$') that cannot be
  achieved in $\LLL'$?
\end{enumerate}
We may also ask: are there examples of \emph{natural, practically
  useful} operations that manifest such differences?  If so, this
might be considered as a significant advantage of $\LLL$ over $\LLL'$.

If the `operations' we are asking about are ordinary first-order
functions --- that is, both their inputs and outputs are of ground
type (strings, arbitrary-size integers etc.)\ --- then the situation
is easily summarised.  At such types, all reasonable languages give
rise to the same class of programmable functions, namely the
Church-Turing computable ones.  As for complexity, the runtime of a
program is typically analysed with respect to some cost model for
basic instructions (e.g.\ one unit of time per array access).
Although the realism of such cost models in the asymptotic limit can
be questioned (see, e.g., \citep[Section~2.6]{Knuth97}), it is broadly
taken as read that such models are equally applicable whatever
programming language we are working with, and moreover that all
respectable languages can represent all algorithms of interest; thus,
one does not expect the best achievable asymptotic run-time for a
typical algorithm (say in number theory or graph theory) to be
sensitive to the choice of programming language, except perhaps in
marginal cases.

The situation changes radically, however, if we consider
\emph{higher-order} operations: programmable operations whose inputs
may themselves be programmable operations.  Here it turns out that
both what is computable and the efficiency with which it can be
computed can be highly sensitive to the selection of language features
present. This is in fact true more widely for \emph{abstract data
  types}, of which higher-order types can be seen as a special case: a
higher-order value will be represented within the machine as ground
data, but a program within the language typically has no access to
this internal representation, and can interact with the value only by
applying it to an argument.

Most work in this area to date has focused on computability
differences. One of the best known examples is the \emph{parallel if}
operation which is computable in a language with parallel evaluation
but not in a typical `sequential' programming language
\citep{Plotkin77}. It is also well known that the presence of control
features or local state enables observational distinctions that cannot
be made in a purely functional setting: for instance, there are
programs involving `call/cc' that detect the order in which a
(call-by-name) `+' operation evaluates its arguments
\citep{CartwrightF92}. Such operations are `non-functional' in the
sense that their output is not determined solely by the extension of
their input (seen as a mathematical function
$\N_\bot \times \N_\bot \rightarrow \N_\bot$);
however, there are also programs with `functional' behaviour that can
be implemented with control or local state but not without them
\citep{Longley99}.  More recent results have exhibited differences
lower down in the language expressivity spectrum: for instance, in a
purely functional setting \textit{\`a la} Haskell, the expressive
power of \emph{recursion} increases strictly with its type level
\citep{Longley18a}, and there are natural operations computable by
low-order recursion but not by high-order iteration
\citep{Longley19}. Much of this territory, including the mathematical
theory of some of the natural notions of higher-order computability
that arise in this way, is mapped out by \citet{LongleyN15}.

Relatively few results of this character have so far been established
on the complexity side. \citet{Pippenger96} gives an example of an
`online' operation on infinite sequences of atomic symbols
(essentially a function from streams to streams) such that the first
$n$ output symbols can be produced within time $\BigO(n)$ if one is
working in an `impure' version of Lisp (in which mutation of `cons'
pairs is admitted), but with a worst-case runtime no better than
$\Omega(n \log n)$ for any implementation in pure Lisp (without such
mutation). This example was reconsidered by \citet{BirdJdM97} who
showed that the same speedup can be achieved in a pure language by
using lazy evaluation.  Another candidate is the familiar $\log n$
overhead involved in implementing maps (supporting lookup and
extension) in a pure functional language \cite{Okasaki99}, although to
our knowledge this situation has not yet been subjected to theoretical
scrutiny.  \citet{Jones01} explores the approach of manifesting
expressivity and efficiency differences between certain languages by
artificially restricting attention to `cons-free' programs; in this
setting, the classes of representable first-order functions for the
various languages are found to coincide with some well-known
complexity classes.

The purpose of the present paper is to give a clear example of such an
inherent complexity difference higher up in the expressivity spectrum.
Specifically, we consider the following \emph{generic count} problem,
parametric in $n$: given a boolean-valued predicate $P$ on the space
${\mathbb B}^n$ of boolean vectors of length $n$, return the number of
such vectors $q$ for which $P\,q = \True$.  We shall consider boolean
vectors of any length to be represented by the type $\Nat \to \Bool$;
thus for each $n$, we are asking for an implementation of a certain
third-order operation
\[ \Count_n : ((\Nat \to \Bool) \to \Bool) \to \Nat  \]
A \naive implementation strategy, supported by any reasonable
language, is simply to apply $P$ to each of the $2^n$ vectors in turn.
A much less obvious, but still purely `functional', approach due to
\citet{Berger90} achieves the effect of `pruned search' where the
predicate allows it (serving as a warning that counter-intuitive
phenomena can arise in this territory).  Nonetheless, under a mild
condition on $P$ (namely that it must inspect all $n$ components of
the given vector before returning), both these approaches will have a
$\Omega(n 2^n)$ runtime.  Moreover, we shall show that in a typical
call-by-value language without advanced control features, one cannot
improve on this: \emph{any} implementation of $\Count_n$ must
necessarily take time $\Omega(n2^n)$ on \emph{any} predicate $P$.  On
the other hand, if we extend our language with a feature such as
\emph{effect handlers} (see Section~\ref{sec:handlers-primer} below),
it becomes possible to bring the runtime down to $\BigO(2^n)$: an
asymptotic gain of a factor of $n$.

The \emph{generic search} problem is just like the generic count
problem, except rather than counting the vectors $q$ such that $P\,q =
\True$, it returns the list of all such vectors.
The $\Omega(n 2^n)$ runtime for purely functional implementations
transfers directly to generic search, as generic count reduces to
generic search composed with computing the length of the resulting
list.
In Section~\ref{sec:count-vs-search} we illustrate that the
$\BigO(2^n)$ runtime for generic count with effect handlers also
transfers to generic search.

The idea behind the speedup is easily explained and will already be
familiar, at least informally, to programmers who have worked with
multi-shot continuations.
Suppose for example $n=3$, and suppose that the predicate $P$ always
inspects the components of its argument in the order $0,1,2$.
A \naive implementation of $\Count_3$ might start by applying the given
$P$ to $q_0 = (\True,\True,\True)$, and then to
$q_1 = (\True,\True,\False)$.  Clearly there is some duplication here:
the computations of $P\,q_0$ and $P\,q_1$ will proceed identically up
to the point where the value of the final component is requested. What
we would like to do, then, is to record the state of the computation
of $P\,q_0$ at just this point, so that we can later resume this
computation with $\False$ supplied as the final component value in
order to obtain the value of $P\,q_1$. (Similarly for all other
internal nodes in the evident binary tree of boolean vectors.) Of
course, this `backup' approach would be standardly applied if one were
implementing a bespoke search operation for some \emph{particular}
choice of $P$ (corresponding, say, to the $n$-queens problem); but to
apply this idea of resuming previous subcomputations in the
\emph{generic} setting (that is, uniformly in $P$) requires some
special language feature such as effect handlers or multi-shot
continuations.
One could also obviate the need for such a feature by choosing to
present the predicate $P$ in some other way, but from our present
perspective this would be to move the goalposts: our intention is
precisely to show that our languages differ in an essential way
\emph{as regards their power to manipulate data of type} $(\Nat \to
\Bool) \to \Bool$.

This idea of using first-class control to achieve `backtracking' has
been exploited before and is fairly widely known (see
e.g. \citep{KiselyovSFA05}), and there is a clear programming
intuition that this yields a speedup unattainable in languages without
such control features.  Our main contribution in this paper is to
provide, for the first time, a precise mathematical theorem that pins
down this fundamental efficiency difference, thus giving formal
substance to this intuition.  Since our goal is to give a realistic
analysis of the efficiency achievable in various settings without
getting bogged down in inessential implementation details, we shall
work concretely and operationally with the languages in question,
using a CEK-style abstract machine semantics as our basic model of
execution time, and with some specific programs in these languages.
In the first instance, we formulate our results as a comparison
between a purely functional base language (a version of call-by-value
PCF) and an extension with first-class control; we then indicate how
these results can be extended to base languages with other features
such as mutable state.

In summary, our purpose is to exhibit an efficiency gap which, in our
view, manifests a fundamental feature of the programming language
landscape, challenging a common assumption that all real-world
programming languages are essentially `equivalent' from an asymptotic
point of view.  We believe that such results are important not only
for a rounded understanding of the relative merits of existing
languages, but also for informing future language design.

For their convenience as structured delimited control operators we
adopt effect handlers as our universal control abstraction of choice,
but our results adapt mutatis mutandis to other first-class control
abstractions such as `call/cc'~\citep{SperberDFvSFM09}, `control'
($\mathcal{F}$) and 'prompt' ($\textbf{\#}$)~\citep{Felleisen88}, or
`shift' and `reset'~\citep{DanvyF90}.

The rest of the paper is structured as follows.
\begin{itemize}
 \item Section~\ref{sec:handlers-primer} provides an introduction to
   effect handlers as a programming abstraction.
 \item Section~\ref{sec:calculi} presents a PCF-like language
   $\BCalc$ and its extension $\HCalc$ with effect handlers.
 \item Section~\ref{sec:abstract-machine-semantics} defines abstract
   machines for $\BCalc$ and $\HCalc$, yielding a runtime cost model.
 \item Section~\ref{sec:generic-search} introduces generic count and
   some associated machinery, and presents an implementation in
   $\HCalc$ with runtime $\BigO(2^n)$.
 \item Section~\ref{sec:pure-counting} establishes that any generic
   count implementation in $\BCalc$ must have runtime $\Omega(n2^n)$.
 \item Section~\ref{sec:robustness} shows that our results scale to
   richer settings including support for a wider class of predicates,
   the adaptation from generic count to generic search, and an
   extension of the base language with state.
 \item Section~\ref{sec:experiments} evaluates implementations of
   generic search based on $\BCalc$ and $\HCalc$ in Standard ML.
 \item Section \ref{sec:conclusions} concludes.
\end{itemize}
The languages $\BCalc$ and $\HCalc$ are rather minimal versions of
previously studied systems --- we only include the machinery needed
for illustrating the generic search efficiency phenomenon.
Auxiliary results are included in the appendices of the extended
version of the paper~\citep{hillerstromLJ20}.

\section{Effect Handlers Primer}
\label{sec:handlers-primer}
Effect handlers were originally studied as a theoretical means to
provide a semantics for exception handling in the setting of algebraic
effects~\cite{PlotkinP01, PlotkinP13}.
Subsequently they have emerged as a practical programming abstraction
for modular effectful programming~\citep{BauerP15, ConventLMM20,
  KammarLO13, KiselyovSS13, DolanWSYM15, Leijen17, HillerstromLA20}.
In this section we give a short introduction to effect handlers.  For
a thorough introduction to programming with effect handlers, we
recommend the tutorial by \citet{Pretnar15}, and as an introduction to
the mathematical foundations of handlers, we refer the reader to the
founding paper by \citet{PlotkinP13} and the excellent tutorial paper
by \citet{Bauer18}.

Viewed through the lens of universal algebra, an algebraic effect is
given by a signature $\Sigma$ of typed \emph{operation symbols} along
with an equational theory that describes the properties of the
operations~\cite{PlotkinP01}.
An example of an algebraic effect is \emph{nondeterminism}, whose
signature consists of a single nondeterministic choice operation:
$\Sigma \defas \{ \Branch : \One \to \Bool \}$.
The operation takes a single parameter of type unit and ultimately
produces a boolean value.
The pragmatic programmatic view of algebraic effects differs from the
original development as no implementation accounts for equations over
operations yet.

As a simple example, let us use the operation $\Branch$ to model a
coin toss.
Suppose we have a data type $\dec{Toss} \defas \dec{Heads} \mid
\dec{Tails}$, then
we may implement a coin toss as follows.
{\small
\[
  \bl
    \dec{toss} : \One \to \dec{Toss}\\
    \dec{toss}~\Unit =
           \If \; \Do\; \Branch\; \Unit \;
           \Then\; \dec{Heads} \;
           \Else\; \dec{Tails}
  \el
\]}%
From the type signature it is clear that the computation returns a
value of type $\dec{Toss}$. It is not clear from the signature of
$\dec{toss}$ whether it performs an effect. However, from the
definition, it evidently performs the operation $\Branch$ with
argument $\Unit$ using the $\Do$-invocation form. The result of the
operation determines whether the computation returns either
$\dec{Heads}$ or $\dec{Tails}$.
Systems such as Frank~\cite{LindleyMM17, ConventLMM20},
Helium~\cite{BiernackiPPS19, BiernackiPPS20}, Koka~\cite{Leijen17},
and Links~\cite{HillerstromL16, HillerstromLA20} include
type-and-effect systems which track the use of effectful operations,
whilst current iterations of systems such as Eff~\cite{BauerP15} and
Multicore OCaml~\cite{DolanWSYM15} elect not to track effects in the
type system.
Our language is closer to the latter two.

We may view an effectful computation as a tree, where the interior
nodes correspond to operation invocations and the leaves correspond to
return values.
The computation tree for $\dec{toss}$ is as follows.
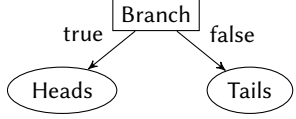
\begin{center}
  {\small
  \tossTree}%
\end{center}
It models interaction with the environment. The operation $\Branch$
can be viewed as a \emph{query} for which the \emph{response} is
either $\True$ or $\False$. The response is provided by an effect
handler. As an example, consider the following handler which enumerates
the possible outcomes of a coin toss.
{\small
\[
  \bl
    \Handle\; \dec{toss}~\Unit\;\With\\
      \quad\ba[t]{@{~}l@{~}c@{~}l}
           \Val~x &\mapsto& [x]\\
           \Branch~\Unit~r &\mapsto& r~\True \concat r~\False
           \ea
  \el
\]}%
The $\Handle$-construct generalises the exceptional syntax
of~\citet{BentonK01}.
This handler has a \emph{success} clause and an \emph{operation}
clauses.
The success clause determines how to interpret the return value of
$\dec{toss}$, or equivalently how to interpret the leaves of its
computation tree.
It lifts the return value into a singleton list.
The operation clause determines how to interpret occurrences of
$\Branch$ in $\dec{toss}$. It provides access to the argument of
$\Branch$ (which is unit) and its resumption, $r$. The resumption is a
first-class delimited continuation which captures the remainder of the
$\dec{toss}$ computation from the invocation of $\Branch$ up to its
nearest enclosing handler.

Applying $r$ to $\True$ resumes evaluation of $\dec{toss}$ via the
$\True$ branch, returning $\dec{Heads}$ and causing the success clause
of the handler to be invoked; thus the result of $r~\True$ is
$[\dec{Heads}]$. Evaluation continues in the operation clause,
meaning that $r$ is applied again, but this time to $\False$, which
causes evaluation to resume in $\dec{toss}$ via the $\False$
branch. By the same reasoning, the value of $r~\False$ is
$[\dec{Tails}]$, which is concatenated with the result of the
$\True$ branch; hence the handler ultimately returns
$[\dec{Heads}, \dec{Tails}]$.

\section{Calculi}
\label{sec:calculi}
In this section, we present our base language $\BCalc$ and its
extension with effect handlers $\HCalc$.

\subsection{Base Calculus}
The base calculus $\BCalc$ is a fine-grain
call-by-value~\cite{LevyPT03} variation of PCF~\cite{Plotkin77}.
Fine-grain call-by-value is similar to A-normal
form~\cite{FlanaganSDF93} in that every intermediate computation is
named, but unlike A-normal form is closed under reduction.

The syntax of $\BCalc$ is as follows.
{\small
\noindent
  \begin{syntax}
    \slab{Types}          &A,B,C,D\in\TypeCat  &::= & \Nat \mid \One \mid A \to B \mid A \times B \mid A + B \\
    \slab{Type Environments} &\Gamma\in\CtxCat &::= & \cdot \mid \Gamma, x:A \\
\slab{Values}        &V,W\in\ValCat  &::= & x \mid k \mid c \mid \lambda x^A .\, M \mid \Rec \; f^{A \to B}\, x.M \\
                     &               &\mid& \Unit \mid \Record{V, W} \mid (\Inl\, V)^B \mid (\Inr\, W)^A\\
\slab{Computations}  &M,N\in\CompCat
                           &::= & V\,W
                            \mid  \Let\; \Record{x,y} = V \; \In \; N \\
                     &     &\mid&\Case \; V \;\{ \Inl \; x \mapsto M; \Inr \; y \mapsto N\}\\
                     &    &\mid& \Return\; V
                           \mid \Let \; x \revto M \; \In \; N \\
\end{syntax}}%
The ground types are $\Nat$ and $\One$ which classify natural number
values and the unit value, respectively. The function type $A \to B$
classifies functions that map values of type $A$ to values of type
$B$. The binary product type $A \times B$ classifies pairs of values
whose first and second components have types $A$ and $B$
respectively. The sum type $A + B$ classifies tagged values of either
type $A$ or $B$.
Type environments $\Gamma$ map term variables to their types.

We let $k$ range over natural numbers and $c$ range over primitive
operations on natural numbers ($+, -, =$).
We let $x, y, z$ range over term variables.
For convenience, we also use $f$, $g$, and $h$ for variables of
function type, $i$ and $j$ for variables of type $\Nat$, and $r$ to
denote resumptions.
The value terms are standard.

%
We will occasionally blur the distinction between object and meta
language by writing $A$ for the meta level type of closed value terms
of type $A$.
All elimination forms are computation terms. Abstraction is eliminated
using application ($V\,W$).
The product eliminator $(\Let \; \Record{x,y} = V \; \In \; N)$ splits
a pair $V$ into its constituents and binds them to $x$ and $y$,
respectively. Sums are eliminated by a case split ($\Case\; V\;
\{\Inl\; x \mapsto M; \Inr\; y \mapsto N\}$).
A trivial computation $(\Return\;V)$ returns value $V$. The sequencing
expression $(\Let \; x \revto M \; \In \; N)$ evaluates $M$ and binds
the result value to $x$ in $N$.

\begin{figure*}
\small
\raggedright\textbf{Values}
\begin{mathpar}
  \inferrule*[Lab=\tylab{Var}]
    {x : A \in \Gamma}
    {\typv{\Gamma}{x : A}}

  \inferrule*[Lab=\tylab{Unit}]
    { }
    {\typv{\Gamma}{\Unit : \One}}

  \inferrule*[Lab=\tylab{Nat}]
    { k \in \mathbb{N} }
    {\typv{\Gamma}{k : \Nat}}

  \inferrule*[Lab=\tylab{Const}]
    {c : A \to B}
    {\typv{\Gamma}{c : A \to B}}
\\
  \inferrule*[Lab=\tylab{Lam}]
    {\typ{\Gamma, x : A}{M : B}}
    {\typv{\Gamma}{\lambda x^A .\, M : A \to B}}

  \inferrule*[Lab=\tylab{Rec}]
    {\typ{\Gamma, f : A \to B, x : A}{M : B}}
    {\typv{\Gamma}{\Rec\; f^{A \to B}\,x .\, M : A \to B}}
\\
  \inferrule*[Lab=\tylab{Prod}]
    { \typv{\Gamma}{V : A} \\
      \typv{\Gamma}{W : B}
    }
    {\typv{\Gamma}{\Record{V,W} : A \times B}}

  \inferrule*[Lab=\tylab{Inl}]
    {\typv{\Gamma}{V : A}}
    {\typv{\Gamma}{(\Inl\,V)^B : A + B}}

  \inferrule*[Lab=\tylab{Inr}]
    {\typv{\Gamma}{W : B}}
    {\typv{\Gamma}{(\Inr\,W)^A : A + B}}
\end{mathpar}

\textbf{Computations}
\begin{mathpar}
  \inferrule*[Lab=\tylab{App}]
    {\typv{\Gamma}{V : A \to B} \\
     \typv{\Gamma}{W : A}
    }
    {\typ{\Gamma}{V\,W : B}}

  \inferrule*[Lab=\tylab{Split}]
    {\typv{\Gamma}{V : A \times B} \\
     \typ{\Gamma, x : A, y : B}{N : C}
    }
    {\typ{\Gamma}{\Let \; \Record{x,y} = V\; \In \; N : C}}

  \inferrule*[Lab=\tylab{Case}]
    { \typv{\Gamma}{V : A + B}  \\
      \typ{\Gamma,x : A}{M : C} \\
      \typ{\Gamma,y : B}{N : C}
    }
    {\typ{\Gamma}{\Case \; V \;\{\Inl\; x \mapsto M; \Inr \; y \mapsto N \} : C}}
\\
  \inferrule*[Lab=\tylab{Return}]
    {\typv{\Gamma}{V : A}}
    {\typ{\Gamma}{\Return \; V : A}}

  \inferrule*[Lab=\tylab{Let}]
    {\typ{\Gamma}{M : A} \\
     \typ{\Gamma, x : A}{N : C}
    }
    {\typ{\Gamma}{\Let \; x \revto M\; \In \; N : C}}
\end{mathpar}
\caption{Typing Rules for $\BCalc$}
\label{fig:typing}
\end{figure*}

The typing rules are given in Figure~\ref{fig:typing}.
We require two typing judgements: one for values and the other for
computations.
The judgement $\typ{\Gamma}{\square : A}$ states that a $\square$-term
has type $A$ under type environment $\Gamma$, where $\square$ is
either a value term ($V$) or a computation term ($M$).
The constants have the following types.
{\small
\begin{mathpar}
\{(+), (-)\} : \Nat \times \Nat \to \Nat

(=) : \Nat \times \Nat \to \One + \One
\end{mathpar}}
\begin{figure*}
\small
\begin{reductions}
\semlab{App}     & (\lambda x^A . \, M) V   &\reducesto& M[V/x] \\
\semlab{App-Rec} & (\Rec\; f^A \,x.\, M) V  &\reducesto& M[(\Rec\;f^A\,x .\,M)/f,V/x]\\
\semlab{Const}   & c~V                      &\reducesto& \Return\;(\const{c}\,(V)) \\
\semlab{Split} & \Let \; \Record{x,y} = \Record{V,W} \; \In \; N &\reducesto& N[V/x,W/y] \\
\semlab{Case-inl} &
  \Case \; (\Inl\, V)^B \; \{\Inl \; x \mapsto M;\Inr \; y \mapsto N\} &\reducesto& M[V/x] \\
\semlab{Case-inr} &
  \Case \; (\Inr\, V)^A \; \{\Inl \; x \mapsto M; \Inr \; y \mapsto N\} &\reducesto& N[V/y]\\
\semlab{Let} &
  \Let \; x \revto \Return \; V \; \In \; N &\reducesto& N[V/x] \\
\semlab{Lift} &
  \EC[M] &\reducesto& \EC[N], \hfill \text{if }M \reducesto N \\
\end{reductions}
\begin{syntax}
\slab{Evaluation contexts} &  \mathcal{E} &::=& [\,] \mid \Let \; x \revto \mathcal{E} \; \In \; N
\end{syntax}
\caption{Contextual Small-Step Operational Semantics}
\label{fig:small-step}
\end{figure*}
We give a small-step operational semantics for \BCalc{} with
\emph{evaluation contexts} in the style of \citet{Felleisen87}. The
reduction rules are given in Figure~\ref{fig:small-step}.
We write $M[V/x]$ for $M$ with $V$ substituted for $x$ and $\const{c}$
for the usual interpretation of constant $c$ as a meta-level function
on closed values. The reduction relation $\reducesto$ is defined on
computation terms. The statement $M \reducesto N$ reads: term $M$
reduces to term $N$ in one step.
We write $R^+$ for the transitive closure of relation $R$ and $R^*$
for the reflexive, transitive closure of relation $R$.

\paragraph{Notation}
We elide type annotations when clear from context.
For convenience we often write code in direct-style assuming the
standard left-to-right call-by-value elaboration into fine-grain
call-by-value~\citep{Moggi91, FlanaganSDF93}.
For example, the expression $f\,(h\,w) + g\,\Unit$ is syntactic sugar
for:
{\small
\[
      \ba[t]{@{~}l}
      \Let\; x \revto h\,w \;\In\;
      \Let\; y \revto f\,x \;\In\;
      \Let\; z \revto g\,\Unit \;\In\;
      y + z
      \ea
\]}%
We define sequencing of computations in the standard way.
{\small
\[
  M;N \defas \Let\;x \revto M \;\In\;N, \quad \text{where $x \notin FV(N)$}
\]}%
We make use of standard syntactic sugar for pattern matching. For
instance, we write
{\small
\[
  \lambda\Unit.M \defas \lambda x^{\One}.M, \quad \text{where $x \notin FV(M)$}
\]}%
for suspended computations, and if the binder has a type other than
$\One$, we write:
{\small
\[
  \lambda\_^A.M \defas \lambda x^A.M, \quad \text{where $x \notin FV(M)$}
\]}%
We use the standard encoding of booleans as a sum:
{\small
\begin{mathpar}
\Bool \defas \One + \One

\True \defas \Inl~\Unit

\False \defas \Inr~\Unit

\If\;V\;\Then\;M\;\Else\;N \defas \Case\;V\;\{\Inl~\Unit \mapsto M; \Inr~\Unit \mapsto N\}
\end{mathpar}}%

%
%
\subsection{Handler Calculus}
\label{sec:handlers-calculus}

We now define $\HCalc$ as an extension of $\BCalc$.
{\small
\begin{syntax}
\slab{Operation symbols} &\ell \in \mathcal{L} & & \\
\slab{Signatures}        &\Sigma&::=& \cdot \mid \{\ell : A \to B\} \cup \Sigma\\
\slab{Handler types}     &F     &::=& C \Rightarrow D\\
\slab{Computations} &M, N &::=& \dots \mid \Do \; \ell \; V
                          \mid  \Handle \; M \; \With \; H \\
\slab{Handlers}     &H&::=& \{ \Val \; x \mapsto M \}
                      \mid  \{ \ell \; p \; r \mapsto N \} \uplus H\\
\end{syntax}}%
We assume a countably infinite set $\mathcal{L}$ of operation symbols
$\ell$.
An effect signature $\Sigma$ is a map from operation symbols to their
types, thus we assume that each operation symbol in a signature is
distinct. An operation type $A \to B$ classifies operations that take
an argument of type $A$ and return a result of type $B$.
We write $dom(\Sigma) \subseteq \mathcal{L}$ for the set of operation
symbols in a signature $\Sigma$.
A handler type $C \Rightarrow D$ classifies effect handlers that
transform computations of type $C$ into computations of type $D$.
Following \citet{Pretnar15}, we assume a global signature for every
program.
Computations are extended with operation invocation ($\Do\;\ell\;V$)
and effect handling ($\Handle\; M \;\With\; H$).
Handlers are constructed from one success clause $(\{\Val\; x \mapsto
M\})$ and one operation clause $(\{ \ell \; p \; r \mapsto N \})$ for
each operation $\ell$ in $\Sigma$.
Following \citet{PlotkinP13}, we adopt the convention that a handler
with missing operation clauses (with respect to $\Sigma$) is syntactic
sugar for one in which all missing clauses perform explicit
forwarding:
\[
   \{\ell \; p \; r \mapsto \Let\; x \revto \Do \; \ell \, p \;\In\; r \, x\}
\]

\begin{figure*}
\small
\raggedright
\textbf{Computations}
\begin{mathpar}
  \inferrule*[Lab=\tylab{Do}]
    {(\ell : A \to B) \in \Sigma \\ \typ{\Gamma}{V : A} }
    {\typ{\Gamma}{\Do \; \ell \; V : B}}

\inferrule*[Lab=\tylab{Handle}]
  {\typ{\Gamma}{M : C} \\
   \Gamma \vdash H : C \Rightarrow D}
  {\typ{\Gamma}{\Handle \; M \; \With \; H : D}}
\end{mathpar}
\textbf{Handlers}
\begin{mathpar}
\inferrule*[Lab=\tylab{Handler}]
    {  \hret = \{\Val \; x \mapsto M\} \\
      [\hell = \{\ell \, p \; r \mapsto N_\ell\}]_{\ell \in dom(\Sigma)} \\\\
      \typ{\Gamma, x : C}{M : D} \\
      [\typ{\Gamma, p : A_\ell, r : B_\ell \to D}{N_\ell : D}]_{(\ell : A_\ell \to B_\ell) \in \Sigma}
    }
    {{\Gamma} \vdash {H : C \Rightarrow D}}
\end{mathpar}

\caption{Additional Typing Rules for $\HCalc$}
\label{fig:typing-handlers}
\end{figure*}

The typing rules for $\HCalc$ are those of $\BCalc$
(Figure~\ref{fig:typing}) plus three additional rules for operations,
handling, and handlers given in Figure~\ref{fig:typing-handlers}.
The \tylab{Do} rule ensures that an operation invocation is only
well-typed if the operation $\ell$ appears in the effect signature
$\Sigma$ and the argument type $A$ matches the type of the provided
argument $V$. The result type $B$ determines the type of the
invocation.
The \tylab{Handle} rule types handler application.
The \tylab{Handler} rule ensures that the bodies of the success clause
and the operation clauses all have the output type $D$. The type of
$x$ in the success clause must match the input type $C$. The type of
the parameter $p$ ($A_\ell$) and resumption $r$ ($B_\ell \to D$) in
operation clause $\hell$ is determined by the type of $\ell$; the
return type of $r$ is $D$, as the body of the resumption will itself
be handled by $H$.
We write $\hell$ and $\hret$ for projecting success and operation
clauses.
{\small
\[
  \ba{@{~}r@{~}c@{~}l@{~}l}
    \hret &\defas& \{\Val\, x \mapsto M \}, &\quad \text{where } \{\Val\, x \mapsto M \} \in H\\
    \hell &\defas& \{\ell\, p\,r \mapsto M \}, &\quad \text{where } \{\ell\, p\;r \mapsto M \} \in H
  \ea
\]}%

We extend the operational semantics to $\HCalc$. Specifically, we add
two new reduction rules: one for handling return values and another
for handling operation invocations.
{\small
\begin{reductions}
\semlab{Ret} & \Handle \; (\Return \; V) \; \With \; H &\reducesto& N[V/x], \qquad
                                      \text{where } \hret = \{ \Val \; x \mapsto N \} \smallskip\\
                                      \semlab{Op}  & \Handle \; \EC[\Do \; \ell \, V] \; \With \; H &\reducesto& N[V/p,\; (\lambda y.\Handle \; \EC[\Return \; y] \; \With \; H)/r],\\
    \multicolumn{4}{@{}r@{}}{
      \hfill\text{where } \hell = \{ \ell\, p \; r \mapsto N \}
    }
\end{reductions}}%
The first rule invokes the success clause.
The second rule handles an operation via the corresponding operation
clause.
If we were \naively to extend evaluation contexts with the handle
construct then our semantics would become nondeterministic, as it may
pick an arbitrary handler in scope.
In order to ensure that the semantics is deterministic, we instead add
a distinct form of evaluation context for effectful computation, which
we call handler contexts.
{\small
\begin{syntax}
\slab{Handler contexts} &  \HC &::= & [\,] \mid \Handle \; \HC \; \With \; H
                                \mid  \Let\;x \revto \HC\; \In\; N\\
\end{syntax}}%
We replace the $\semlab{Lift}$ rule with a corresponding rule for
handler contexts.
{\small
\[
  \HC[M] ~\reducesto~ \HC[N], \qquad\hfill\text{if } M \reducesto N
\]}%
The separation between pure evaluation contexts $\EC$ and handler
contexts $\HC$ ensures that the $\semlab{Op}$ rule always selects the
innermost handler.

We now characterise normal forms and state the standard type soundness
property of $\HCalc$.
\begin{definition}[Computation normal forms]
  A computation term $N$ is normal with respect to $\Sigma$, if $N =
  \Return\;V$ for some $V$ or $N = \EC[\Do\;\ell\,W]$ for some $\ell
  \in dom(\Sigma)$, $\EC$, and $W$.
\end{definition}

\begin{theorem}[Type Soundness]
  If $\typ{}{M : C}$, then either there exists $\typ{}{N : C}$ such
  that $M \reducesto^* N$ and $N$ is normal with respect to $\Sigma$,
  or $M$ diverges.
\end{theorem}

\subsection{The Role of Types}

Readers familiar with backtracking search algorithms may wonder where
types come into the expressiveness picture.
Types will not play a direct role in our proofs but rather in the
characterisation of which programs can be meaningfully compared. In
particular, types are used to rule out global approaches such as
continuation passing style (CPS): without types one could obtain an
efficient pure generic count program by CPS transforming the entire
program.

Readers familiar with effect handlers may wonder why our handler
calculus does not include an effect type system.
As types frame the comparison of programs between languages, we
require that types be fixed across languages; hence $\HCalc$ does not
include effect types.
Future work includes reconciling effect typing with our approach to
expressiveness.

\section{Abstract Machine Semantics}
\label{sec:abstract-machine-semantics}
Thus far we have introduced the base calculus $\BCalc$ and its
extension with effect handlers $\HCalc$.
For each calculus we have given a \emph{small-step operational
  semantics} which uses a substitution model for evaluation. Whilst
this model is semantically pleasing, it falls short of providing a
realistic account of practical computation as substitution is an
expensive operation. We now develop a more practical model of
computation based on an \emph{abstract machine semantics}.

\subsection{Base Machine}
\label{sec:base-abstract-machine}

\newcommand{\Conf}{\dec{Conf}}
\newcommand{\EConf}{\dec{EConf}}
\newcommand{\MVal}{\dec{MVal}}

We choose a \emph{CEK}-style abstract machine
semantics~\citep{FelleisenF86} for \BCalc{} based on that of
\citet{HillerstromLA20}.
The CEK machine operates on configurations which are triples of the
form $\cek{M \mid \gamma \mid \sigma}$. The first component contains
the computation currently being evaluated. The second component
contains the environment $\gamma$ which binds free variables. The
third component contains the continuation which instructs the machine
how to proceed once evaluation of the current computation is complete.
The syntax of abstract machine states is as follows.
{\small
\begin{syntax}
\slab{Configurations}           & \conf \in \Conf  &::=& \cek{M \mid \env \mid \sigma} \\
\slab{Environments}       &\env \in \Env   &::=& \emptyset \mid \env[x \mapsto v] \\
\slab{Machine values}           &v, w \in \MVal  &::= & x \mid n \mid c \mid \Unit \mid \Record{v, w} \\
                                &                &\mid& (\env, \lambda x^A .\, M) \mid (\env, \Rec\, f^{A \to B}\,x . \, M)
                                                  \mid  (\Inl\, v)^B \mid (\Inr\,w)^A \\
\slab{Pure continuations}            &\sigma \in \PureCont &::=& \nil \mid (\env, x, N) \cons \sigma \\
\end{syntax}}%
Values consist of function closures, constants, pairs, and left or
right tagged values.
We refer to continuations of the base machine as \emph{pure}.
A pure continuation is a stack of pure continuation frames. A pure
continuation frame $(\env, x, N)$ closes a let-binding $\Let \;x
\revto [~] \;\In\;N$ over environment $\env$.
We write $\nil$ for an empty pure continuation and $\phi \cons \sigma$
for the result of pushing the frame $\phi$ onto $\sigma$. We use
pattern matching to deconstruct pure continuations.

\begin{figure*}
\small
\raggedright
\textbf{Transition relation}
\begin{reductions}
\mlab{App} & \cek{ V\;W \mid \env \mid \sigma}
           &\stepsto& \cek{ M \mid \env'[x \mapsto \val{W}{\env}] \mid \sigma},\\
           &&& \quad\text{ if }\val{V}{\env} = (\env', \lambda x^A . \, M)\\

\mlab{Rec} & \cek{ V\;W \mid \env \mid \sigma}
           &\stepsto& \cek{ M \mid \env'[\bl
                                         f \mapsto (\env', \Rec\,f^{A \to B}\,x. M), \\
                                         x \mapsto \val{W}{\env}] \mid \sigma},\\
                                         \el \\
           &&& \quad\text{ if }\val{V}{\env} = (\env', \Rec\, f^{A \to B}\, x. M)\\

\mlab{Const} & \cek{ V~W \mid \env \mid \sigma}
             &\stepsto& \cek{ \Return\; (\const{c}\,(\val{W}\env)) \mid \env \mid \sigma},\\
             &&& \quad\text{ if }\val{V}{\env} = c \\
\mlab{Split} & \cek{ \Let \; \Record{x,y} = V \; \In \; N \mid \env \mid \sigma}
             &\stepsto& \cek{ N \mid \env[x \mapsto v, y \mapsto w] \mid \sigma}, \\
             &&& \quad\text{ if }\val{V}{\env} = \Record{v; w} \\

\mlab{CaseL} & \ba{@{}l@{}l@{}}
               \cekl \Case\; V\, \{&\Inl\, x \mapsto M; \\
                                   &\Inr\, y \mapsto N\} \mid \env \mid \sigma \cekr \\
               \ea
             &\stepsto& \cek{ M \mid \env[x \mapsto v] \mid \sigma},\\
             &&& \quad\text{ if }\val{V}{\env} = \Inl\, v \\

\mlab{CaseR} & \ba{@{}l@{}l@{}}
               \cekl \Case\; V\, \{&\Inl\, x \mapsto M; \\
                                   &\Inr\, y \mapsto N\} \mid \env \mid \sigma \cekr \\
               \ea
             &\stepsto& \cek{ N \mid \env[y \mapsto v] \mid \sigma},\\
             &&& \quad\text{ if }\val{V}{\env} = \Inr\, v \\

\mlab{Let} & \cek{ \Let \; x \revto M \; \In \; N \mid \env \mid \sigma}
    &\stepsto& \cek{ M \mid \env \mid (\env,x,N) \cons \sigma} \\

\mlab{RetCont} &\cek{ \Return \; V \mid \env \mid (\env',x,N) \cons \sigma}
          &\stepsto& \cek{ N \mid \env'[x \mapsto \val{V}{\env}] \mid \sigma} \\

\end{reductions}

\textbf{Value interpretation}
\[
\bl
\begin{eqs}
\val{x}{\env}                    &=& \env(x) \\
\val{\Unit{}}{\env}              &=& \Unit{} \\
\end{eqs}
\qquad\qquad\qquad
\begin{eqs}
\val{n}{\env}                    &=& n \\
\val{c}\env                      &=& c \\
\end{eqs}
\qquad\qquad\qquad
\begin{eqs}
\val{\lambda x^A.M}{\env}      &=& (\env, \lambda x^A.M) \\
\val{\Rec\, f^{A \to B}\, x.M}{\env}   &=& (\env, \Rec\,f^{A \to B}\, x.M) \\
\end{eqs}
\medskip \\
\begin{eqs}
\val{\Record{V, W}}{\env} &=& \Record{\val{V}{\env}, \val{W}{\env}} \\
\end{eqs}
\qquad\qquad\qquad
\ba{@{}r@{~}c@{~}l@{}}
\val{(\Inl\, V)^B}{\env}         &=& (\Inl\; \val{V}{\env})^B \\
\val{(\Inr\, V)^A}{\env}         &=& (\Inr\; \val{V}{\env})^A \\
\ea
\el
\]

\caption{Abstract Machine Semantics for $\BCalc$}
\label{fig:abstract-machine-semantics}
\end{figure*}

The abstract machine semantics is given in
Figure~\ref{fig:abstract-machine-semantics}.
The transition relation ($\stepsto$) makes use of the value
interpretation ($\val{-}$) from value terms to machine values.
The machine is initialised by placing a term in a configuration
alongside the empty environment ($\emptyset$) and identity
pure continuation ($\nil$).
The rules (\mlab{App}), (\mlab{Rec}), (\mlab{Const}), (\mlab{Split}),
(\mlab{CaseL}), and (\mlab{CaseR}) eliminate values.
The (\mlab{Let}) rule extends the current pure continuation with let
bindings.
The (\mlab{RetCont}) rule extends the environment in the top frame of
the pure continuation with a returned value.
Given an input of a well-typed closed computation term $\typ{}{M :
  A}$, the machine will either diverge or return a value of type $A$.
A final state is given by a configuration of the form $\cek{\Return\;V
  \mid \env \mid \nil}$ in which case the final return value is given
by the denotation $\val{V}{\env}$ of $V$ under environment $\gamma$.

\paragraph{Correctness}
The base machine faithfully simulates the operational semantics for
$\BCalc$; most transitions correspond directly to $\beta$-reductions,
but $\mlab{Let}$ performs an administrative step to bring the
computation $M$ into evaluation position.
We formally state and prove the correspondence in
Appendix~\extendedref{sec:base-machine-correctness}, relying on an
inverse map $\inv{-}$ from configurations to
terms~\citep{HillerstromLA20}.
\newcommand{\contapp}[2]{#1 #2}
\newcommand{\contappp}[2]{#1(#2)}

\subsection{Handler Machine}
\newcommand{\HClosure}{\dec{HClo}}
We now enrich the $\BCalc$ machine to a $\HCalc$ machine.
We extend the syntax as follows.
{\small
\begin{syntax}
  \slab{Configurations}            &\conf \in \Conf &::=& \cek{M \mid \env \mid \kappa}\\
  \slab{Resumptions}               &\rho \in \dec{Res} &::=& (\sigma, \chi)\\
  \slab{Continuations}             &\kappa \in \Cont &::=& \nil \mid \rho \cons \kappa\\
  \slab{Handler closures}          &\chi \in \HClosure   &::=& (\env, H) \\
  \slab{Machine values}            &v, w \in \MVal  &::=& \cdots \mid \rho \\
\end{syntax}}%
The notion of configurations changes slightly in that the continuation
component is replaced by a generalised continuation
$\kappa \in \Cont$~\cite{HillerstromLA20}; a continuation is now a
list of resumptions. A resumption is a pair of a pure continuation (as
in the base machine) and a handler closure ($\chi$).
A handler closure consists of an environment and a handler definition,
where the former binds the free variables that occur in the latter.
The identity continuation is a singleton list containing the identity
resumption, which is an empty pure continuation paired with the
identity handler closure:
{\small
\[
\kappa_0 \defas [(\nil, (\emptyset, \{\Val\;x \mapsto x\}))]
\]}%
Machine values are augmented to include resumptions as an operation
invocation causes the topmost frame of the machine continuation to be
reified (and bound to the resumption parameter in the operation
clause).

The handler machine adds transition rules for handlers, and modifies
$(\mlab{Let})$ and $(\mlab{RetCont})$ from the base machine to account
for the richer continuation
structure. Figure~\ref{fig:abstract-machine-semantics-handlers}
depicts the new and modified rules.
The $(\mlab{Handle})$ rule pushes a handler closure along with an
empty pure continuation onto the continuation stack.
The $(\mlab{RetHandler})$ rule transfers control to the success clause
of the current handler once the pure continuation is empty.
The $(\mlab{Handle-Op})$ rule transfers control to the matching
operation clause on the topmost handler, and during the process it
reifies the handler closure. Finally, the $(\mlab{Resume})$ rule
applies a reified handler closure, by pushing it onto the continuation
stack.
The handler machine has two possible final states: either it yields a
value or it gets stuck on an unhandled operation.

\begin{figure*}
\small
\raggedright

\textbf{Transition relation}
\begin{reductions}
\mlab{Resume} & \cek{ V\;W \mid \env \mid \kappa}
               &\stepsto& \cek{ \Return \; W \mid \env \mid (\sigma, \chi) \cons \kappa},\\
               &&&\quad\text{ if }\val{V}{\env} = (\sigma, \chi) \\

\mlab{Let} & \cek{ \Let \; x \revto M \; \In \; N \mid \env \mid (\sigma, \chi) \cons \kappa}
    &\stepsto& \cek{ M \mid \env \mid ((\env,x,N) \cons \sigma, \chi) \cons \kappa} \\

\mlab{RetCont} &\cek{ \Return \; V \mid \env \mid ((\env',x,N) \cons \sigma, \chi) \cons \kappa}
        &\stepsto& \cek{ N \mid \env'[x \mapsto \val{V}{\env}] \mid (\sigma, \chi) \cons \kappa} \\

\mlab{Handle} & \cek{ \Handle \; M \; \With \; H \mid \env \mid \kappa}
       &\stepsto& \cek{ M \mid \env \mid (\nil, (\env, H)) \cons \kappa} \\

\mlab{RetHandler} & \cek{ \Return \; V \mid \env \mid (\nil, (\env',H)) \cons \kappa}
                  &\stepsto& \cek{ M \mid \env'[x \mapsto \val{V}{\env}] \mid \kappa},\\
                  &&&\quad\text{ if } \hret = \{\Val\; x \mapsto M\} \\

\mlab{Handle-Op} & \cek{ \Do \; \ell~V \mid \env \mid (\sigma, (\env', H)) \cons \kappa }
                 &\stepsto& \cek{ M \mid \env'[\bl
                                               p \mapsto \val{V}\env, \\
                                               r \mapsto (\sigma, (\env', H))] \mid \kappa }, \\
                                               \el \\
                 &&&\quad\bl
                   \text{ if } \ell : A \to B \in \Sigma\\
                   \text{ and } \hell = \{\ell\; p \; r \mapsto M\}
                          \el\\
\end{reductions}
\caption{Abstract Machine Semantics for $\HCalc$}
\label{fig:abstract-machine-semantics-handlers}
\end{figure*}

\paragraph{Correctness}
The handler machine faithfully simulates the operational semantics of
$\HCalc$.
Extending the result for the base machine, we formally state and prove
the correspondence in
Appendix~\extendedref{sec:handler-machine-correctness}.

\subsection{Realisability and Asymptotic Complexity}
\label{sec:realisability}
As witnessed by the work of \citet{HillerstromL18} the machine
structures are readily realisable using standard persistent functional
data structures.
Pure continuations on the base machine and generalised continuations
on the handler machine can be implemented using linked lists with a
time complexity of $\BigO(1)$ for the extension operation
$(\_\cons\_)$.
The topmost pure continuation on the handler machine may also be
extended in time $\BigO(1)$, as extending it only requires reaching
under the topmost handler closure.
Environments, $\env$, can be realised using a map, with a time
complexity of $\BigO(\log|\env|)$ for extension and
lookup~\citep{Okasaki99}.

The worst-case time complexity of a single machine transition is
exhibited by rules which involve operations on the environment, since
any other operation is constant time, hence the worst-time complexity
of a transition is $\BigO(\log|\env|)$.
The value interpretation function $\val{-}\env$ is defined
structurally on values. Its worst-time complexity is exhibited by a
nesting of pairs of variables $\val{\Record{x_1,\dots,x_n}}\env$ which
has complexity $\BigO(n\log|\env|)$.

\paragraph{Continuation copying} On the handler machine the topmost
continuation frame can be copied in constant time due to the
persistent runtime and the layout of machine continuations. An
alternative design would be to make the runtime non-persistent
in which case copying a continuation frame $((\sigma, \_) \cons
\_)$ would be a $\BigO(|\sigma|)$ time operation.

\paragraph{Primitive operations on naturals}
Our model assumes that arithmetic operations on arbitrary natural
numbers take $\BigO(1)$ time. This is common practice in the study of
algorithms when the main interest lies
elsewhere~\citep[Section~2.2]{CormenLRS09}. If desired, one could
adopt a more refined cost model that accounted for the bit-level
complexity of arithmetic operations; however, doing so would have the
same impact on both of the situations we are wishing to compare, and
thus would add nothing but noise to the overall analysis.

\section{Predicates, Decision Trees and Generic Count}
\label{sec:generic-search}

We now come to the crux of the paper. In this section and the next, we
prove that $\HCalc$ supports implementations of certain operations
with an asymptotic runtime bound that cannot be achieved in $\BCalc$
(Section~\ref{sec:pure-counting}).
While the positive half of this claim essentially consolidates a
known piece of folklore, the negative half appears to be new.
To establish our result, it will suffice to exhibit a single
`efficient' program in $\HCalc$, then show that no equivalent program
in $\BCalc$ can achieve the same asymptotic efficiency.  We take
\emph{generic search} as our example.

Generic search is a modular search procedure that takes as input
a predicate $P$ on some multi-dimensional search space,
and finds all points of the space satisfying $P$.
Generic search is agnostic to the specific instantiation of $P$,
and as a result is applicable across a wide spectrum of domains.
Classic examples such as Sudoku solving~\citep{Bird06}, the
$n$-queens problem~\citep{BellS09} and graph colouring
can be cast as instances of generic search, and similar ideas have
been explored in connection with Nash equilibria and
exact real integration~\citep{Simpson98, Daniels16}.

For simplicity, we will restrict attention to search spaces of the form $\B^n$,
the set of bit vectors of length $n$.
To exhibit our phenomenon in the simplest
possible setting, we shall actually focus on the \emph{generic count} problem:
given a predicate $P$ on some $\B^n$, return the \emph{number of} points
of $\B^n$ satisfying $P$. However, we shall explain why our results
are also applicable to generic search proper.

We shall view $\B^n$ as the set of functions $\N_n \to \B$,
where $\N_n \defas \{0,\dots,n-1\}$.
In both $\BCalc$ and $\HCalc$ we may represent such functions by terms of type $\Nat \to \Bool$.
We will often informally write $\Nat_n$ in place of $\Nat$ to indicate that
only the values $0,\dots,n-1$ are relevant, but this convention has no
formal status since our setup does not support dependent types.

To summarise, in both $\BCalc$ and $\HCalc$ we will be working with the types
{\small
\[
\begin{twoeqs}
  \Point  & \defas & \Nat \to \Bool        & \hspace*{2.0em} &
  \Point_n & \defas & \Nat_n \to \Bool \\
  \Predicate & \defas & \Point \to \Bool &  &
  \Predicate_n & \defas & \Point_n \to \Bool
\end{twoeqs}
\]
}
and will be looking for programs
{\small
\[
  \Count_n : \Predicate_n \to \Nat
\]}%
such that for suitable terms $P$ representing semantic predicates $\Pi: \B^n \to \B$,
$\Count_n~P$ finds the number of points of $\B^n$ satisfying $\Pi$.

Before formalising these ideas more closely, let us look at some examples,
which will also illustrate the machinery of \emph{decision trees} that we will be using.

\subsection{Examples of Points, Predicates and Trees}
\label{sec:predicates-points}
Consider first the following terms of type $\Point$:
{\small
\begin{mathpar}
\dec{q}_0 \defas \lambda \_. \True

\dec{q}_1 \defas \lambda i. i=0

\dec{q}_2 \defas \lambda i.\,
      \If\;i = 0\;\Then\;\True\;
      \Else\;\If\;i = 1\;\Then\;\False\;
      \Else\;\bot
\end{mathpar}}%
(Here $\bot$ is the diverging term $(\Rec\; f\,i.f\,i)\,\Unit$.)
Then $\dec{q}_0$ represents $\langle{\True,\dots,\True}\rangle \in \B^n$ for any $n$;
$\dec{q}_1$ represents $\langle{\True,\False,\dots,\False}\rangle \in \B^n$ for any $n \geq 1$;
and $\dec{q}_2$ represents $\langle{\True,\False}\rangle \in \B^2$.

Next some predicates.
First, the following terms all represent the constant true predicate $\B^2 \to \B$:
{\small
\begin{mathpar}
\dec{T}_0 \defas \lambda q. \True

\dec{T}_1 \defas \lambda q.(q\,1; q\,0; \True)

\dec{T}_2 \defas \lambda q.(q\,0; q\,0; \True)
\end{mathpar}}%
These illustrate that in the course of evaluating a predicate term $P$ at a point $\dec{q}$,
for each $i<n$ the value of $\dec{q}$ at $i$ may be inspected zero, one or many times.

Likewise, the following all represent the `identity' predicate $\B^1 \to \B$
(here $\&\&$ is shortcut `and'):
{\small
\begin{mathpar}
\dec{I}_0 \defas \lambda q. q\,0

\dec{I}_1 \defas \lambda q.\, \If\;q\,0\; \Then\; \True \; \Else\; \False

\dec{I}_2 \defas \lambda q. (q\,0) \,\&\&\, (q\,0)
\end{mathpar}}%

Slightly more interestingly, for each $n$ we have the following program which determines
whether a point contains an odd number of $\True$ components:
{\small
\[
  \dec{Odd}_n \defas \lambda q.\, \dec{fold}\otimes\False~(\dec{map}~q~[0,\dots,n-1])
\]}%
Here $\dec{fold}$ and $\dec{map}$ are the standard combinators on lists, and $\otimes$ is exclusive-or.
Applying $\dec{Odd}_2$ to $\dec{q}_0$ yields $\False$;
applying it to $\dec{q}_1$ or $\dec{q}_2$ yields $\True$.
\medskip

\newcommand{\smath}[1]{\ensuremath{{\scriptstyle #1}}}

\newcommand{\InfiniteModel}{%
\begin{tikzpicture}[->,>=stealth',level/.style={sibling distance = 3.0cm/##1,
    level distance = 1.0cm}]
\node (root) [draw=none] { }
  child { node [opnode] {$\smath{\query 0}$}
    child { node [opnode] {$\smath{\query 0}$}
      child { node [draw=none,rotate=165] {$\vdots$}
        edge from parent node { }
      }
      child { node[leaf] {$\smath{\ans\False}$}
        edge from parent node { }
      }
      edge from parent node { }
    }
    child { node [leaf] {$\smath{\ans\False}$}
      edge from parent node { }
    }
    edge from parent node { }
  }
;
\end{tikzpicture}}
\newcommand{\ShortConjModel}{%
\begin{tikzpicture}[->,>=stealth',level/.style={sibling distance = 3.5cm/##1,
    level distance = 1.0cm}]
\node (root) [draw=none] { }
  child { node [opnode] {$\smath{\query 0}$}
    child { node [opnode] {$\smath{\query 0}$}
      child { node [treenode] {$\smath{\ans\True}$}
        edge from parent node { }
      }
      child { node[treenode] {$\smath{\ans\False}$}
        edge from parent node { }
      }
      edge from parent node { }
    }
    child { node [treenode] {$\smath{\ans\False}$}
      edge from parent node { }
    }
    edge from parent node { }
  }
;
\end{tikzpicture}}

\newcommand{\TTTwoModel}{%
\begin{tikzpicture}[->,>=stealth',level/.style={sibling distance = 8cm/##1,
    level distance = 1.5cm}]
\node (root) [draw=none] { }
  child { node [opnode] {$\smath{\query 0}$}
    child { node [opnode] {$\smath{\query 1}$}
      child { node [leaf] {$\smath{\ans\True}$}
        edge from parent node { }
      }
      child { node[leaf] {$\smath{\ans\True}$}
        edge from parent node { }
      }
      edge from parent node { }
    }
    child { node [opnode] {$\smath{\query 1}$}
      child { node [leaf] {$\smath{\ans\True}$}
        edge from parent node { }
      }
      child { node[leaf] {$\smath{\ans\True}$}
        edge from parent node { }
      }
      edge from parent node { }
    }
    edge from parent node { }
  }
;
\end{tikzpicture}}
\newcommand{\XORTwoModel}{%
\begin{tikzpicture}[->,>=stealth',level/.style={sibling distance = 5.5cm/##1,
    level distance = 1cm}]
\node (root) [draw=none] { }
  child { node [opnode] {$\smath{\query 0}$}
    child { node [opnode] {$\smath{\query 1}$}
      child { node [treenode] {$\smath{\ans\False}$}
        edge from parent node { }
      }
      child { node[treenode] {$\smath{\ans\True}$}
        edge from parent node { }
      }
      edge from parent node { }
    }
    child { node [opnode] {$\smath{\query 1}$}
      child { node [treenode] {$\smath{\ans\True}$}
        edge from parent node { }
      }
      child { node[treenode] {$\smath{\ans\False}$}
        edge from parent node { }
      }
      edge from parent node { }
    }
    edge from parent node { }
  }
;
\end{tikzpicture}}
\newcommand{\TTZeroModel}{%
  \begin{tikzpicture}[->,>=stealth',level/.style={sibling distance = 1cm/##1,
    level distance = 1cm}]
  \node (root) [draw=none] { }
  child { node [treenode] {$\smath{\ans\True}$}
    edge from parent node { }
  }
;
\end{tikzpicture}}%
\begin{figure}
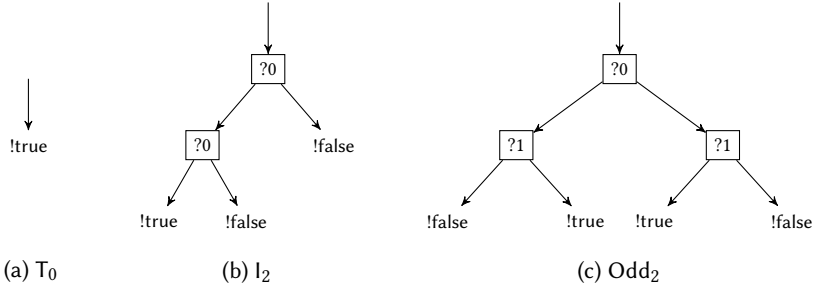

  \centering
  \begin{subfigure}{0.1\textwidth}
    \begin{center}
    \vspace*{6.5ex}
    \scalebox{1.0}{\TTZeroModel}
    \vspace*{6.5ex}
    \end{center}
    \caption{$\dec{T}_0$}
    \label{fig:tt0-tree}
  \end{subfigure}
  \begin{subfigure}{0.3\textwidth}
    \begin{center}
    \scalebox{1.0}{\ShortConjModel}
    \end{center}
    \caption{$\dec{I}_2$}
    \label{fig:div1-tree}
  \end{subfigure}
  \begin{subfigure}{0.4\textwidth}
    \begin{center}
    \scalebox{1.0}{\XORTwoModel}
    \end{center}
    \caption{$\dec{Odd}_2$}
    \label{fig:xor2-tree}
  \end{subfigure}
  \caption{Examples of Decision Trees}
  \label{fig:example-models}
\end{figure}

We can think of a predicate term $P$ as participating in a `dialogue'
with a given point $Q : \Point_n$.
The predicate may \emph{query} $Q$ at some coordinate $k$;
$Q$ may \emph{respond} with $\True$ or $\False$ and this returned value
may influence the future course of the dialogue.
After zero or more such query/response pairs, the predicate may return a
final \emph{answer} ($\True$ or $\False$).

The set of possible dialogues with a given term $P$ may be organised
in an obvious way into an unrooted binary \emph{decision tree}, in
which each internal node is labelled with a query $\query k$ (with
$k<n$), and with left and right branches corresponding to the
responses $\True$, $\False$ respectively.  Any point will thus
determine a path through the tree, and each leaf is labelled with an
answer $\ans \True$ or $\ans \False$ according to whether the
corresponding point or points satisfy the predicate.

Decision trees for a sample of the above predicate terms are depicted
in Figure~\ref{fig:example-models}; the relevant formal definitions
are given in the next subsection.  In the case of $\dec{I}_2$, one of
the $\ans \False$ leaves will be `unreachable' if we are working in
$\BCalc$ (but reachable in a language supporting mutable state).

We think of the edges in the tree as corresponding to portions of
computation undertaken by $P$ between queries, or before delivering
the final answer.  The tree is unrooted (i.e.\ starts with an edge
rather than a node) because in the evaluation of $P\,Q$ there is
potentially some `thinking' done by $P$ even before the first query or
answer is reached.  For the purpose of our runtime analysis, we will
also consider \emph{timed} variants of these decision trees, in which
each edge is labelled with the number of computation steps involved.

It is possible that for a given $P$ the construction of a decision
tree may hit trouble, because at some stage $P$ either goes undefined
or gets stuck at an unhandled operation.  It is also possible that the
decision tree is infinite because $P$ can keep asking queries forever.
However, we shall be restricting our attention to terms representing
\emph{total} predicates: those with finite decision trees in which
every path leads to a leaf.

In order to present our complexity results in a simple and clear form,
we will give special prominence to certain well-behaved decision
trees.  For $n \in \N$, we shall say a tree is \emph{$n$-standard} if
it is total (i.e.\ every maximal path leads to a leaf labelled with an
answer) and along any path to a leaf, each coordinate $k<n$ is queried
once and only once. Thus, an $n$-standard decision tree is a complete
binary tree of depth $n+1$, with $2^n - 1$ internal nodes and $2^n$
leaves.  However, there is no constraint on the order of the queries,
which indeed may vary from one path to another.  One pleasing property
of this notion is that for a predicate term with an $n$-standard
decision tree, the number of points in $\B^n$ satisfying the predicate
is precisely the number of $\ans \True$ leaves in the tree.

Of the examples we have given, the tree for $\dec{T}_0$ is 0-standard;
those for $\dec{I}_0$ and $\dec{I}_1$ are 1-standard; that for
$\dec{Odd}_n$ is $n$-standard; and the rest are not $n$-standard for
any $n$.

\subsection{Formal Definitions}
\label{sec:predicate-models}
We now formalise the above notions.  We will present our definitions
in the setting of $\HCalc$, but everything can clearly be relativised
to $\BCalc$ with no change to the meaning in the case of $\BCalc$
terms.  For the purpose of this subsection we fix $n \in \N$, set
$\N_n \defas \{0,\ldots,n-1\}$, and use $k$ to range over $\N_n$. We
write $\B$ for the set of booleans, which we shall identify with the
(encoded) boolean values of $\HCalc$, and use $b$ to range over $\B$.

As suggested by the foregoing discussion, we will need to work with
both syntax and semantics.  For points, the relevant definitions are
as follows.

\begin{definition}[$n$-points]\label{def:semantic-n-point}
  A closed value $Q : \Point$ is said to be a \emph{syntactic $n$-point} if:
{\small
  \[
    \forall k \in \N_n.\,\exists b \in \B.~ Q~k \reducesto^\ast \Return\;b
  \]}%
A \emph{semantic $n$-point} $\pi$ is simply a mathematical function
$\pi: \N_n \to \B$.  (We shall also write $\pi \in \B^n$.)  Any
syntactic $n$-point $Q$ is said to \emph{denote} the semantic
$n$-point $\val{Q}$ given by:
{\small
  \[
  \forall k \in \N_n,\, b \in \B.~ \val{Q}(k) = b \,\Iff\, Q~k \reducesto^\ast \Return\;b
  \]}%
Any two syntactic $n$-points $Q$ and $Q'$ are said to be
\emph{distinct} if $\val{Q} \neq \val{Q'}$.
\end{definition}

By default, the unqualified term \emph{$n$-point} will from now on
refer to syntactic $n$-points.

Likewise, we wish to work with predicates both syntactically and
semantically.  By a \emph{semantic $n$-predicate} we shall mean simply
a mathematical function $\Pi: \B^n \to \B$.  One slick way to define
syntactic $n$-predicates would be as closed terms $P:\Predicate$ such
that for every $n$-point $Q$, $P\,Q$ evaluates to either
$\Return\;\True$ or $\Return\;\False$.  For our purposes, however, we
shall favour an approach to $n$-predicates via \emph{decision trees},
which will yield more information on their behaviour.

We will model decision trees as certain partial functions from
\emph{addresses} to \emph{labels}.  An address will specify the
position of a node in the tree via the path that leads to it, while a
label will represent the information present at a node. Formally:

\begin{definition}[untimed decision tree]\label{def:decision-tree}
  (i) The address set $\Addr$ is simply the set $\B^\ast$ of finite lists of booleans.
      If $bs,bs' \in \Addr$, we write $bs \sqsubseteq bs'$ (resp.\ $bs \sqsubset bs'$)
      to mean that $bs$ is a prefix (resp.\ proper prefix) of $bs'$.

  (ii) The label set $\Lab$ consists of \emph{queries} parameterised by a natural
       number and \emph{answers} parameterised by a boolean:
{\small
  \[
  \Lab \defas \{\query k \mid k \in \N \} \cup \{\ans b \mid b \in \B \}
  \]
}%

 (iii) An (untimed) decision tree is a partial function $\tree : \Addr
\pto \Lab$ such that:
\begin{itemize}
  \item The domain of $\tree$ (written $dom(\tree)$) is prefix closed.
  \item Answer nodes are always leaves:
    if $\tree(bs) = \ans b$ then $\tree(bs')$ is undefined whenever $bs \sqsubset bs'$.
\end{itemize}
\end{definition}

As our goal is to reason about the time complexity of generic count
programs and their predicates, it is also helpful to decorate decision
trees with timing data that records the number of machine steps taken
for each piece of computation performed by a predicate:

\begin{definition}[timed decision tree]\label{def:timed-decision-tree}
A timed decision tree is a partial function $\tree : \Addr \pto
\Lab \times \N$ such that its first projection $bs \mapsto \tree(bs).1$
is a decision tree.
We write $\tl$ for the first projection ($bs \mapsto \tree(bs).1$) and
$\ts$ for the second projection ($bs \mapsto \tree(bs).2$) of a timed
decision tree.
\end{definition}

Here we think of $\steps(\tree)(bs)$ as the computation time
associated with the edge whose \emph{target} is the node addressed by
$bs$.

We now come to the method for associating a specific tree with a given
term $P$. One may think of this as a kind of denotational semantics,
but here we shall extract a tree from a term by purely operational
means using our abstract machine model. The key idea is to try
applying $P$ to a distinguished free variable $q: \Point$, which we
think of as an `abstract point'. Whenever $P$ wants to interrogate its
argument at some index $i$, the computation will get stuck at some
term $q\,i$: this both flags up the presence of a query node in the
decision tree, and allows us to explore the subsequent behaviour under
both possible responses to this query.

The core of our definition is couched in terms of abstract machine configurations.
We write $\Conf_q$ for the set of $\lambda_h$ configurations possibly involving $q$
(but no other free variables).
We write $a \simeq b$ for Kleene equality: either both $a$ and $b$ are
undefined or both are defined and $a = b$.

It is convenient to define the timed tree and then extract the untimed one from it:

\begin{definition}\label{def:model-construction}
 (i) Define $\tr: \Conf_q \to \Addr \pto (\Lab \times \N)$ to be the
  minimal family of partial functions satisfying the following
  equations:
{\small
\begin{mathpar}
\ba{@{}r@{~}c@{~}l@{\qquad}l@{}}
  \tr(\cek{\Return\;W \mid \env \mid \nil})\, \nil  &~=~& (!b, 0),
                                                    &\text{if }\val{W}\env = b \smallskip\\
%
  \tr(\cek{z\,V \mid \env \mid \kappa})\, \nil  &~=~& (?\val{V}{\env}, 0),
                                                    &\text{if } \gamma(z) = q \smallskip\\
  \tr(\cek{z\,V \mid \env \mid \kappa})\, (b \cons bs) &~\simeq~& \tr(\cek{\Return\;b \mid \env \mid \kappa})\,bs,
                                                                & \text{if } \gamma(z) = q \smallskip\\
  \tr(\cek{M \mid \env \mid \kappa})\, bs &~\simeq~& \mathsf{inc}\,(\tr(\cek{M' \mid \env' \mid \kappa'})\, bs),
  &\text{if } \cek{M \mid \env \mid \kappa} \stepsto \cek{M' \mid \env' \mid \kappa'}
\ea
\end{mathpar}}%
Here $\mathsf{inc}(\ell, s) = (\ell, s + 1)$, and in all of the above equations
$\gamma(q) = \gamma'(q) = q$.
Clearly $\tr(\conf)$ is a timed decision tree for any $\conf \in \Conf_q$.

(ii) The timed decision tree of a computation term is obtained by placing it in
the initial configuration:
$\tr(M) \defas \tr(\cek{M, \emptyset[q \mapsto q], \kappa_0})$.

(iii) The timed decision tree of a closed value $P:\Predicate$ is $\tr(P\,q)$.
Since $q$ plays the role of a dummy argument, we will usually omit it and write $\tr(P)$ for $\tr(P\,q)$.

(iv) The untimed decision tree $\tru(P)$ is obtained from $\tr(P)$ via
first projection: $\tru(P) = \labs(\tr(P))$.
\end{definition}

If the execution of a configuration $\conf$ runs forever or gets stuck at an unhandled operation,
then $\tr(\conf)(bs)$ will be undefined for all $bs$.
Although this is admitted by our definition of decision tree, we wish to exclude such behaviours
for the terms we accept as valid predicates. Specifically, we frame the following definition:

\begin{definition}  \label{def:n-predicate}
A decision tree $\tree$ is an \emph{$n$-predicate tree} if it satisfies the following:
\begin{itemize}
  \item For every query $\query k$ appearing in $\tree$, we have $k \in \N_n$.
  \item Every query node has both children present:
  \[ \forall bs \in \Addr,\, k \in \N_n,\, b \in \B.~ \tree(bs) = \query k \Implies \snoc{bs}{b} \in dom(\tree) \]
  \item All paths in $\tree$ are finite (so every maximal path terminates in an answer node).
\end{itemize}
A closed term $P: \Predicate$ is a \emph{(syntactic) $n$-predicate} if $\tru(P)$ is an $n$-predicate tree.
\end{definition}

If $\tree$ is an $n$-predicate tree, clearly any semantic $n$-point $\pi$ gives rise to a path $b_0 b_1 \dots $
through $\tree$, given inductively by:
{\small
\[  \forall j.~ \mbox{if~} \tau(b_0\dots b_{j-1}) = \query k_j \mbox{~then~} b_j = \pi(k_j)  \]
}%
This path will terminate at some answer node $b_0 b_1 \dots b_{r-1}$ of $\tree$,
and we may write $\tree \bullet \pi \in \B$ for the answer at this leaf.

\begin{proposition}  \label{prop:pred-tree}
If $P$ is an $n$-predicate and $Q$ is an $n$-point, then
$P\,Q \reducesto^\ast \Return\;b$ where $b = \tru(P) \bullet \val{Q}$.
\end{proposition}

\begin{proof}
By interleaving the computation for the relevant path through $\tru(P)$
with computations for queries to $Q$, and appealing to the correspondence between
the small-step reduction and abstract machine semantics.
We omit the routine details.
\end{proof}

It is thus natural to define the \emph{denotation} of an $n$-predicate
$P$ to be the semantic $n$-predicate $\val{P}$ given by
$\val{P}(\pi) = \tru(P) \bullet \pi$.

As mentioned earlier, we shall also be interested in a more constrained
class of trees and predicates:

\begin{definition}[$n$-standard trees and predicates]
An $n$-predicate tree $\tree$ is said to be $n$-standard if the following hold:
\begin{itemize}
\item The domain of $\tree$ is precisely $\Addr_n$, the set of bit vectors of length $\leq n$.
\item There are no repeated queries along any path in $\tree$:
  \[ \forall bs, bs' \in dom(\tree),\, k \in \N_n.~ bs \sqsubseteq bs' \wedge \tree(bs)=\tau(bs')=\query k \Implies bs=bs' \]
\end{itemize}
A timed decision tree $\tree$ is $n$-standard if its underlying untimed
decision tree ($bs \mapsto \tree(bs).1$) is so.
An $n$-predicate $P$ is $n$-standard if $\tr(P)$ is $n$-standard.
\end{definition}

Clearly, in an $n$-standard tree, each of the $n$ queries $\query 0,\dots, \query(n-1)$
appears exactly once on the path to any leaf, and there are $2^n$ leaves, all of them answer nodes.

\subsection{Specification of Counting Programs}
\label{sec:counting}

We can now specify what it means for a program
$\Countprog : \Predicate \to \Nat$ to implement counting.

\begin{definition} \label{def:counting-function}
(i) The \emph{count} of a semantic $n$-predicate $\Pi$, written $\sharp \Pi$,
is simply the number of semantic $n$-points $\pi \in \B^n$ for which $\Pi(\pi)=\True$.

(ii) If $P$ is any $n$-predicate, we say that $\Countprog$ \emph{correctly counts} $P$ if
$\Countprog\,P \reducesto^\ast \Return\;m$, where $m = \sharp \val{P}$.
\end{definition}

This definition gives us the flexibility to talk about counting
programs that operate on various classes of predicates, allowing us to
state our results in their strongest natural form. On the positive
side, we shall shortly see that there is a single `efficient' program
in $\HCalc$ that correctly counts all $n$-standard $\lambda_h$
predicates for every $n$; in Section~\ref{sec:beyond} we improve this
to one that correctly counts \emph{all} $n$-predicates of $\lambda_h$.
On the negative side, we shall show that an $n$-indexed family of
counting programs written in $\BCalc$, even if only required to work
correctly on $n$-standard $\lambda_b$ predicates, can never compete
with our $\HCalc$ program for asymptotic efficiency even in the most
favourable cases.

\subsection{Efficient Generic Count with Effects}
\label{sec:effectful-counting}

We now present the simplest version of our effectful implementation of
counting: one that works on $n$-standard predicates.

Our program uses a variation of the handler for
nondeterministic computation that we gave in
Section~\ref{sec:handlers-primer}.
The main idea is to implement points as `nondeterministic computations'
using the $\Branch$ operation such that the handler may respond to every query twice,
by invoking the provided resumption with $\True$ and subsequently $\False$.
The key insight is that the resumption restarts computation at the invocation
site of $\Branch$, which means that prior computation need not be repeated.
In other words, the resumption ensures that common portions of computations
prior to any query are shared between both branches.

We assert that $\Branch : \One \to \Bool \in \Sigma$ is a
distinguished operation that may not be handled in the definition of
any input predicate (it has to be forwarded according to the default
convention).
The algorithm is then as follows.
{\small
\[
  \bl
    \ECount : ((\Nat \to \Bool) \to \Bool) \to \Nat\\
    \ECount\,pred \defas
      \bl
      \Handle\; pred\,(\lambda\_. \Do\; \Branch\; \Unit)\; \With\\
      \quad\ba[t]{@{}l@{\hspace{1.5ex}}c@{\hspace{1.5ex}}l@{}}
           \Val\, x         &\mapsto& \If\; x\; \Then\;\Return\; 1 \;\Else\;\Return\; 0 \\
           \Branch\,\Unit\,\,r &\mapsto&
              \ba[t]{@{}l}
                \Let\;x_\True \revto  r\,\True\; \In\\
                \Let\;x_\False \revto r\,\False\;\In\;
                x_\True + x_\False \\
              \ea
           \ea \\
      \el
  \el
\]}%
The handler applies predicate $pred$ to a single `generic point'
defined using $\Branch$. The boolean return value is interpreted as a
single solution, whilst $\Branch$ is interpreted by alternately
supplying $\True$ and $\False$ to the resumption and summing the
results.  The sharing enabled by the use of the resumption is exactly
the `magic' we need to make it possible to implement generic count
more efficiently in $\HCalc$ than in $\BCalc$.
A curious feature of $\ECount$ is that it works for all $n$-standard
predicates without having to know the value of $n$. This is because
the generic point $(\Superpoint)$ informally serves as a
`superposition' of all possible points.

We may now articulate the crucial correctness and efficiency
properties of $\ECount$.

\begin{theorem}\label{thm:complexity-effectful-counting}
  The following hold for any $n \in \N$ and any $n$-standard predicate $P$ of $\HCalc$:
  \begin{enumerate}
  \item $\ECount$ correctly counts $P$.
  \item The number of machine steps required to evaluate $\ECount~P$ is
{\small
  \[
     \left( \displaystyle\sum_{bs \in \Addr_n} \steps(\tr(P))(bs) \right) ~+~ \BigO(2^n)
  \]}%
\end{enumerate}
\end{theorem}
\begin{proof}[Proof Outline]
  Suppose $bs \in \Addr_n$, with $|bs|=j$.  From the construction of
  $\tr(P)$, one may easily read off a configuration $\conf_{bs}$ whose
  execution is expected to compute the count for the subtree below
  node $bs$, and we can explicitly describe the form $\conf_{bs}$ will
  have.  We write $\dec{Hyp}(bs)$ for the claim that $\conf_{bs}$
  correctly counts this subtree, and does so within the following
  number of steps: {\small
  \[
     \left( \displaystyle\sum_{bs' \in \Addr_n,\; bs' \sqsupset bs} \steps(\tr(P))(bs') \right) ~+~ 9 * (2^{n-j} - 1) + 2*2^{n-j}
  \]
}%
The $9*(2^{n-j}-1)$ expression is the number of machine steps
contributed by the $\Branch$-case inside the handler, whilst the
$2*2^{n-j}$ expression is the number of machine steps contributed by
the $\Val$-case.
We prove $\dec{Hyp}(bs)$ by a laborious but routine downwards
induction on the length of $bs$. The proof combines counting of
explicit machine steps with `oracular' appeals to the assumed
behaviour of $P$ as modelled by $\tr(P)$. Once
$\dec{Hyp}(\nil)$ is established, both halves of the theorem
follow easily.
Full details are given in Appendix~\extendedref{sec:positive-theorem}.
\end{proof}

The above formula can clearly be simplified for certain reasonable
classes of predicates. For instance, suppose we fix some constant
$c \in \N$, and let $\mathcal{P}_{n,c}$ be the class of all
$n$-standard predicates $P$ for which all the edge times
$\steps(\tr(P))(bs)$ are bounded by $c$. (Clearly, many reasonable
predicates will belong to $\mathcal{P}_{n,c}$ for some modest value of
$c$.) Since the number of sequences $bs$ in question is less than
$2^{n+1}$, we may read off from the above formula that for predicates
in $\mathcal{P}_{n,c}$, the runtime of $\ECount$ is $\BigO(c2^n)$.

Alternatively, should we wish to use the finer-grained cost model that
assigns an $O(\log |\gamma|)$ runtime to each abstract machine step
(see Section~\ref{sec:realisability}), we may note that any
environment $\gamma$ arising in the computation contains at most $n$
entries introduced by the let-bindings in $\ECount$, and (if $P \in
\mathcal{P}_{n,c}$) at most $\BigO(cn)$ entries introduced by $P$.
Thus, the time for each step in the computation remains $\BigO(\log c
+ \log n)$, and the total runtime for $\ECount$ is $\BigO(c 2^n (\log
c + \log n))$.

One might also ask about the execution time for an implementation of
$\HCalc$ that performs genuine copying of continuations, as in systems
such as \citet{mlton}.
As MLton copies the entire continuation (stack), whose size is
$\BigO(n)$, at each of the $2^n$ branches, continuation copying alone
takes time $\BigO(n2^n)$ and the effectful implementation offers no
performance benefit (Table~\ref{tbl:results-mlton}).
More refined implementations \citep{FarvardinR20, FlattD20} that are
able to take advantage of delimited control operators or sharing in
copies of the stack can bring the complexity of continuation copying
back down to $\BigO(2^n)$.

Finally, one might consider another dimension of cost, namely the
space used by $\ECount$.
Consider a class $\mathcal{Q}_{n,c,d}$ of $n$-standard predicates $P$
for which the edge times in $\tr(P)$ never exceed $c$ and the sizes of
pure continuations never exceed $d$.
If we consider any $P \in \mathcal{Q}_{n,c,d}$ then the total number
of environment entries is bounded by $cn$, taking up space
$\BigO(cn(\log cn))$.
We must also account for the pure continuations. There are $n$ of
these, each taking at most $d$ space.
Thus the total space is $\BigO(n(d + c(\log c + \log n)))$.

\section{Pure Generic Count: A Lower Bound}
\label{sec:pure-counting}

\newcommand{\naivecount}{\dec{naivecount}}
\newcommand{\lazycount}{\dec{lazycount}}
\newcommand{\BergerCount}{\dec{BergerCount}}
\newcommand{\bestshot}{\dec{bestshot}}
\newcommand{\FF}{\mathcal{F}}
\newcommand{\GG}{\mathcal{G}}

We have shown that there is an implementation of generic count in
$\HCalc$ with a runtime bound of $\BigO(2^n)$ for certain well-behaved
predicates. We now prove that no implementation in $\BCalc$ can match
this: in fact, we establish a lower bound of $\Omega(n2^n)$ for the
runtime of any counting program on \emph{any} $n$-standard predicate.
This mathematically rigorous characterisation of the efficiency gap
between languages with and without first-class control constructs is
the central contribution of the paper.

One might ask at this point whether the claimed lower bound could not
be obviated by means of some known continuation passing style (CPS) or
monadic transform of effect handlers
\cite{HillerstromLAS17,Leijen17}. This can indeed be done, but only by
dint of changing the type of our predicates $P$ --- which, as noted in
the introduction, would defeat the purpose of our enquiry.
Our intention is precisely to investigate the relative power of various
languages for manipulating predicates that are given to us in a
certain way which we do not have the luxury of choosing.

To get a feel for the issues that our proof must address, let us
consider how one might construct a counting program in
$\BCalc$.  The \naive approach, of course, would be simply to apply the
given predicate $P$ to all $2^n$ possible $n$-points in turn, keeping
a count of those on which $P$ yields true.  It is a routine exercise to
implement this approach in $\BCalc$, yielding (parametrically in $n$)
a program
{\small
\[
  \naivecount_n ~: ((\Nat_n \to \Bool) \to \Bool) \to \Nat
\]}%
Since the evaluation of an $n$-standard predicate on an individual
$n$-point must clearly take time $\Omega(n)$, we have that the
evaluation of $\naivecount_n$ on any $n$-standard predicate $P$ must
take time $\Omega(n2^n)$. If $P$ is not $n$-standard, the $\Omega(n)$
lower bound need not apply, but we may still say that the evaluation
of $\naivecount_n$ on \emph{any} predicate $P$ (at level $n$) must
take time $\Omega(2^n)$.

One might at first suppose that these properties are inevitable for
any implementation of generic count within $\BCalc$, or indeed any
purely functional language: surely, the only way to learn something
about the behaviour of $P$ on every possible $n$-point is to apply $P$
to each of these points in turn?  It turns out, however, that the
$\Omega(2^n)$ lower bound can sometimes be circumvented by
implementations that cleverly exploit \emph{nesting} of calls to $P$.
The germ of the idea may be illustrated within $\BCalc$ itself.
Suppose that we first construct some program
{\small
\[
  \bestshot_n ~: ((\Nat_n \to \Bool) \to \Bool) \to (\Nat_n \to \Bool)
\]}%
which, given a predicate $P$, returns some $n$-point $Q$ such that
$P~Q$ evaluates to true, if such a point exists, and any point at all
if no such point exists.
(In other words, $\bestshot_n$ embodies Hilbert's choice operator
$\varepsilon$ on predicates.)
It is once again routine to construct such a program by \naive means;
and we may moreover assume that for any $P$, the evaluation of
$\bestshot_n\;P$ takes only constant time, all the real work being
deferred until the argument of type $\Nat_n$ is supplied.

Now consider the following program:
{\small
\[
  \lazycount_n \defas \lambda pred.\; \If \; pred~(\bestshot_n\;pred)\; \Then\; \naivecount_n\;pred\; \Else\; \Return\;0
\]}%
Here the term $pred~(\bestshot_n~pred)$ serves to test whether there
exists an $n$-point satisfying $pred$: if there is not, our count
program may return $0$ straightaway.  It is thus clear that
$\lazycount_n$ is a correct implementation of generic count, and also
that if $pred$ is the predicate $\lambda q.\False$ then
$\lazycount_n\;pred$ returns $0$ within $O(1)$ time, thus violating
the $\Omega(2^n)$ lower bound suggested above.

This might seem like a footling point, as $\lazycount_n$ offers this
efficiency gain \emph{only} on (certain implementations of) the
constantly false predicate.  However, it turns out that by a recursive
application of this nesting trick, we may arrive at a generic count
program that spectacularly defies the $\Omega(2^n)$ lower bound for an
interesting class of (non-$n$-standard) predicates, and indeed proves
quite viable for counting solutions to `$n$-queens' and similar
problems.  We shall refer to this program as $\BergerCount$, as it is
modelled largely on Berger's PCF implementation of the so-called
\emph{fan functional}~\citep{Berger90, LongleyN15}. This program is of
interest in its own right and is briefly presented in
Appendix~\extendedref{sec:berger-count}. It actually requires a mild
extension of $\BCalc$ with a `memoisation' primitive to achieve the
effect of call-by-need evaluation; but such a language can still be
seen as purely `functional' in the same sense as Haskell.

In the meantime, however, the moral is that the use of \emph{nesting}
can lead to surprising phenomena which sometimes defy intuition
(\citet{Escardo07} gives some striking further examples). What we now
wish to show is that for \emph{$n$-standard} predicates, the \naive
lower bound of $\Omega(n2^n)$ cannot in fact be circumvented. The
example of $\BergerCount$ both highlights the need for a rigorous
proof of this and tells us that such a proof will need to pay
particular attention to the possibility of nesting.

We now proceed to the proof itself. We here present the argument in
the basic setting of $\BCalc$; later we will see how a more delicate
argument applies to languages with mutable state
(Section~\ref{sec:mutable-state}).

As a first step, we note that where lower bounds are concerned, it
will suffice to work with the small-step operational semantics of
$\BCalc$ rather than the more elaborate abstract machine model
employed in Section~\ref{sec:base-abstract-machine}. This is because,
as observed in Section~\ref{sec:base-abstract-machine}, there is a
tight correspondence between these two execution models such that for
the evaluation of any closed term, the number of abstract machine
steps is always at least the number of small-step reductions.  Thus,
if we are able to show that the number of small-step reductions for
any generic program program in $\BCalc$ on any $n$-standard predicate
is $\Omega(n2^n)$, this will establish the desired lower bound on the
runtime.

Let us suppose, then, that $\Countprog$ is a program of $\BCalc$ that correctly counts
all $n$-standard predicates of $\BCalc$ for some specific $n$.
We now establish a key lemma, which vindicates the \naive intuition
that if $P$ is $n$-standard, the only way for $\Countprog$ to discover the correct
value for $\sharp \val{P}$ is to perform $2^n$ separate applications $P\;Q$
(allowing for the possibility that these applications need not
be performed `in turn' but might be nested in some complex way).

\begin{lemma}[No shortcuts]\label{lem:no-shortcuts}
  Suppose $\Countprog$ correctly counts all $n$-standard predicates of $\BCalc$.
  If $P$ is an $n$-standard predicate,
  then $\Countprog$ applies $P$ to at least $2^n$ distinct $n$-points.
  More formally, for any of the $2^n$ possible semantic $n$-points
  $\pi : \N_n \to \B$, there is a term $\EC[P~Q]$ appearing in the
  small-step reduction of $\Countprog~P$ such that $Q$ is an $n$-point and $\val{Q} = \pi$.
\end{lemma}

\begin{proof}
  Suppose for a contradiction that $\pi$ is some semantic $n$-point
  such that no application $P~Q$ with $\val{Q}=\pi$ ever arises in the
  course of computing $\Countprog~P$. Let $\tree$ be the untimed
  decision tree for $P$. Let $l$ be the maximal path through $\tau$
  associated with $\pi$: that is, the one we construct by responding
  to each query $\query k$ with $\pi(k)$. Then $l$ is a leaf node such
  that $\tree(l) = \ans (\tree \bullet \pi)$. We now let $\tree'$ be
  the tree obtained from $\tree$ by simply negating this answer value
  at $l$.

  It is a simple matter to construct a $\BCalc$ $n$-standard predicate
  $P'$ whose decision tree is $\tree'$. This may be done just by
  mirroring the structure of $\tree'$ by nested $\If$ statements; we
  omit the easy details.

  Since the numbers of true-leaves in $\tree$ and $\tree'$ differ by
  1, it is clear that if $\Countprog$ indeed correctly counts all
  $n$-standard predicates, then the values returned by $\Countprog~P$
  and $\Countprog~P'$ will have an absolute difference of 1.  On the
  other hand, we shall argue that if the computation of $\Countprog~P$
  never actually `visits' the leaf $l$ in question, then $\Countprog$
  will be unable to detect any difference between $P$ and $P'$.

  The situation is reminiscent of Milner's \emph{context
    lemma}~\citep{Milner77}, which (loosely) says that essentially the
  only way to observe a difference between two programs is to apply
  them to some argument on which they differ.  Traditional proofs of
  the context lemma reason by induction on length of reduction
  sequences, and our present proof is closely modelled on these.

  We shall make frequent use of term contexts $M[-]$ with a hole of
  type $\Predicate$ (which may appear zero, one or more times in $M$)
  in order to highlight particular occurrences of $P$ within a term.
  The following definition enables us to talk about computations that
  avoid the critical point $\pi$:

  \begin{definition}[Safe terms]\label{def:safe}
  If $M[-]$ is such a context of ground type, let us say $M[-]$ is \emph{safe} if
  \begin{itemize}
    \item $M[P]$ is closed, and $M[P] \reducesto^\ast \Return\;W$ for some closed
    ground type value $W$;
    \item For any term $\EC[P~Q]$ appearing in the reduction of $M[P]$, where the
    applicand $P$ in $P~Q$ is a residual of one of the abstracted occurrences in $M[P]$,
    we have that $\val{Q} \neq \pi$.
  \end{itemize}
  We may express this as `$M[P]$ is safe' when it is clear which occurrences of $P$
  we intend to abstract.
  \end{definition}

  For example, our current hypotheses imply that $\Countprog~P$ is safe
  (formally, $\Countprog'[-] \defas \Countprog\;-$ is safe).

  We may now prove the following:

  \begin{lemma}  \label{lem:replacement}
    (i) Suppose $Q[-] : \Point$ and $k : \Nat$ are values such that
    $Q[P]~k$ is safe, and suppose $Q[P]~k \reducesto^m \Return\;b$
    where $m \in \N$.  Then also $Q[P']~k \reducesto^\ast \Return\;b$.

    (ii) Suppose $P~Q[P]$ is safe and $P~Q[P] \reducesto^m
    \Return\;b$. Then also
    $P'~Q[P'] \reducesto^\ast \Return\;b$.
  \end{lemma}

  We prove these claims by simultaneous induction on the computation
  length $m$.  Both claims are vacuous when $m=0$ as neither $Q[P]~k$
  nor $P~Q[P]$ is a $\Return$ term.  We therefore assume $m>0$ where
  both claims hold for all $m'<m$.

  (i) Let $p:\Predicate$ be a distinguished free variable, and
  consider the behaviour of $Q[p]~k$. If this reduces to a value
  $\Return\,W$, then also $Q[P]~k \reducesto^\ast\Return\,W$, whence
  $W = b$ and also $Q[P']~k \reducesto \Return\;b$ as required.
  Otherwise, the reduction of $Q[p]~k$ will get stuck at some term
  $M_0 = \EC_0[p~Q_0[p], p]$.
  Here the first hole in $\EC_0[-,-]$ is in the evaluation position,
  and the second hole abstracts all remaining occurrences of $p$
  within $M_0$. We may also assume that $Q_0[-]$ abstracts all
  occurrences of $p$ in $Q_0[p]$.

  Correspondingly, the reduction of $Q[P]~k$ will reach
  $\EC_0[P~Q_0[P], P]$ and then proceed with the embedded reduction of
  $P~Q_0[P]$.  Note that $P~Q_0[P]$ will be safe because $Q[P]~k$ is.
  So let us suppose that $P~Q_0[P] \reducesto^\ast \Return\;b_0$,
  whence $Q[P]~k \reducesto^\ast \EC_0[\Return\;b_0, P]$.

  We may now investigate the subsequent reduction behaviour of
  $Q[P]~k$ by considering the reduction of $\EC_0[\Return\;b_0, p]$.
  Once again, this may reduce to a value $\Return\;W$, in which case
  $W = b$ and our computation is complete.  Otherwise, the reduction
  of $\EC_0[\Return\;b_0, p]$ will get stuck at some $M_1 =
  \EC_1[p~Q_1[p], p]$, and we may again proceed as above.

  By continuing in this way, we may analyse the reduction of $Q[P]~k$
  as follows.
  {\small
  \begin{mathpar}
  \begin{eqs}
     Q[P]~k & \reducesto^\ast & \EC_0[P~Q_0[P], P] ~\reducesto^\ast~ \EC_0[\Return\;b_0, P]
            ~\reducesto^\ast~ \EC_1[P~Q_1[P], P] ~\reducesto^\ast~ \EC_1[\Return\;b_1, P] \\
            & \reducesto^\ast & \dots
            ~\reducesto^\ast~ \EC_{r-1}[P~Q_{r-1}[P], P] ~\reducesto^\ast~ \EC_{r-1}[\Return\;b_{r-1}, P]
            ~\reducesto~ \Return\;b
  \end{eqs}
  \end{mathpar}
  }%

  Here the terms $P~Q_j[P]$ will be safe, and the reductions $P~Q_j[P]
  \reducesto^\ast \Return\;b_j$ each have length $<m$. We may
  therefore apply part~(ii) of the induction hypothesis and conclude
  that also $P'~Q_j[P'] \reducesto^\ast \Return\;b_j$.
  Furthermore, the remaining segments of the above computation are all
  obtained as instantiations of `generic' reduction sequences
  involving $p$, so these segments will remain valid if $p$ is
  instantiated to $P'$. Reassembling everything, we have a valid
  reduction sequence:
  {\small
  \begin{mathpar}
  \begin{eqs}
     Q[P']~k & \reducesto^\ast & \EC_0[P'~Q_0[P'], P'] ~\reducesto^\ast~ \EC_0[\Return\;b_0, P']
            ~\reducesto^\ast~ \EC_1[P'~Q_1[P'], P'] ~\reducesto^\ast~ \EC_1[\Return\;b_1, P'] \\
            & \reducesto^\ast & \dots
            ~\reducesto^\ast~ \EC_{r-1}[P'~Q_{r-1}[P'], P'] ~\reducesto^\ast~ \EC_{r-1}[\Return\;b_{r-1}, P']
            ~\reducesto~ \Return\;b
  \end{eqs}
  \end{mathpar}
  }%
  This establishes the induction step for part~(i).

  (ii) We may apply a similar analysis to the computation of $P~Q[P]$
  to detect the places where $Q[P]$ is applied to an argument. We do
  this by considering the reduction behaviour of $P~q$, where
  $q:\Point$ is the distinguished variable that featured in
  Definition~\ref{def:model-construction}.  In this way we may analyse
  the computation of $P~Q[P]$ as:
  {\small
  \begin{mathpar}
  \begin{eqs}
     P~Q[P] & ~\reducesto^\ast~ & \EC_0[Q[P]~k_0, Q[P]] ~\reducesto^\ast~ \EC_0[\Return\;b_0, Q[P]]
                    ~\reducesto^\ast~ \EC_1[Q[P]~k_1, Q[P]]  ~\reducesto^\ast~ \dots \\
                 & ~\reducesto^\ast~ & \EC_{r-1}[Q[P]~k_{r-1}, Q[P]] ~\reducesto^\ast~ \EC_{r-1}[\Return\;b_{r-1}, Q[P]]
                   ~\reducesto~ \Return\;b
  \end{eqs}
  \end{mathpar}}
where for each $j$, the first hole in $\EC_j[-,-]$ is in evaluation
position, the term $Q[P]~k_j$ is safe, the reduction
$Q[P]~k_j \reducesto^\ast \Return\;b_j$ has length $<m$, and the
remaining portions of computation are instantiations of generic
reductions involving $q$.  By part~(i) of the induction hypothesis we
may conclude that also $Q[P']~k_j \reducesto^\ast \Return\;b_j$ for
each $j$, and for the remaining segments of computation we may
instantiate $q$ to $Q[P']$.  We thus obtain a computation exhibiting
that $P~Q[P'] \reducesto^\ast \Return\;b$.

  It remains to show that the applicand $P$ may be replaced by $P'$
  here without affecting the result. The idea here is that the
  booleans $b_0,\dots,b_{r-1}$ trace out a path through the decision
  tree for $P$; but since $P~Q[P]$ is safe, we have that $\val{Q[P]}
  \neq \pi$, and so this path does \emph{not} lead to the critical
  leaf $l$. We now have everything we need to establish that $P'~Q[P']
  \reducesto^\ast \Return\;b$ as required.

  More formally, in view of the correspondence between small-step reduction
  and abstract machine semantics, we may readily correlate the above computation of $P~Q[P]$
  with an exploration of the path $bs = b_0 \dots b_{r-1}$ in $\tau = \tru(P)$,
  leading to a leaf with label $\ans b$.
  Since $P$ is $n$-standard, this correlation shows that $r=n$, that for each $j$ we have
  $\tau(b_0\ldots b_{j-1}) = \query k_j$, and that $\{ k_0,\ldots,k_{r-1} \} = \{ 0,\dots,n-1 \}$.
  Furthermore, we have already ascertained that the values of $Q[P]$ and $Q[P']$ at $k_j$ are both $b_j$,
  whence $\val{Q[P]} = \val{Q[P']} = \pi'$ where $\pi'(k_j)=b_j$ for all $j$.
  But $P~Q[P]$ is safe, so in particular $\pi' = \val{Q[P]} \neq \pi$.
  We therefore also have $\tau'(b_0 \dots b_{j-1}) = \query k_j$ for each $j \leq r$
  and $\tau'(b_0 \dots b_{r-1}) = b$.
  Since $\tau' = \tru(P')$ and $\val{Q[P']} = \pi'$, we may conclude by Proposition~\ref{prop:pred-tree}
  that $P'~Q[P'] \reducesto^\ast \Return\;b$.
  This completes the proof of Lemma~\ref{lem:replacement}.

  To finish off the proof of Lemma~\ref{lem:no-shortcuts}, we apply the same analysis
  one last time to the reduction of $\Countprog~P$ itself. This will have the form
  {\small
  \begin{mathpar}
  \begin{eqs}
  \Countprog~P & ~\reducesto^\ast~ & \EC_0[P~Q_0[P], P] ~\reducesto^\ast \EC_0[\Return\;b_0,P]
            ~\reducesto^\ast~ \dots \\
         & ~\reducesto^\ast~ & \EC_{r-1}[P~Q_{r-1}[P], P] ~\reducesto^\ast \EC_{r-1}[\Return\;b_{r-1},P]
            ~\reducesto^\ast~ \Return\;c
  \end{eqs}
  \end{mathpar}
  }%
  where, by hypothesis, each $P~Q_j[P]$ is safe. Using Lemma~\ref{lem:replacement} we may
  replace each subcomputation $P~Q_j[P] \reducesto^\ast \Return\;b_j$ with
  $P'~Q_j[P'] \reducesto^\ast \Return\;b_j$, and so construct a computation exhibiting that
  $\Countprog~P' \reducesto^\ast \Return\;c$.

  This gives our contradiction, as the values of $\Countprog~P$ and $\Countprog~P'$
  are supposed to differ by 1.
\end{proof}

\begin{corollary}
  Suppose $K$ and $P$ are as in Lemma~\ref{lem:no-shortcuts}.  For any
  semantic $n$-point $\pi$ and any natural number $k < n$, the
  reduction sequence for $K~P$ contains a term $\FF[Q~k]$, where $\FF$
  is an evaluation context and $\val{Q}=\pi$.
\end{corollary}

\begin{proof}
  Suppose $\pi \in \B^n$. By Lemma~\ref{lem:no-shortcuts}, the
  computation of $\Countprog~P$ contains some $\EC[P~Q]$ where
  $\val{Q} = \pi$, and the above analysis of the computation of $P~Q$
  shows that it contains a term $\EC'[Q~k]$ for each $k < n$. The
  corollary follows, taking $\FF[-] \defas \EC[\EC'[-]]$.
\end{proof}

This gives our desired lower bound. Since our $n$-points $Q$ are
values, it is clearly impossible that $\FF[Q~k] = \FF'[Q'~k']$ (where
$\FF,\FF'$ are evaluation contexts) unless $Q=Q'$ and $k=k'$.  We may
therefore read off $\pi$ from $\FF[Q~k]$ as $\val{Q}$.  There are thus
at least $n2^n$ distinct terms in the reduction sequence for
$\Countprog~P$, so the reduction has length $\geq n 2^n$.  We have
thus proved:

\begin{theorem}
  If $\Countprog$ is a $\BCalc$ program that correctly counts all $n$-standard $\BCalc$ predicates,
  and $P$ is any $n$-standard $\BCalc$ predicate, then the evaluation of $\Countprog~P$ must take time
  $\Omega(n2^n)$.  \qed
\end{theorem}

Although we shall not go into details, it is not too hard to apply our
proof strategy with minor adjustments to certain richer languages: for
instance, an extension of $\BCalc$ with exceptions, or one containing
the memoisation primitive required for $\BergerCount$
(Appendix~\extendedref{sec:berger-count}). A deeper adaptation is required
for languages with state: we will return to this in
Section~\ref{sec:robustness}.

It is worth noting where the above argument breaks down if applied to $\HCalc$.
In $\BCalc$, in the course of computing $\Countprog~P$, every $Q$ to which $P$ is applied
will be a self-contained closed term denoting some specific point $\pi$.
This is intuitively why we may only learn about one point at a time.
In $\HCalc$, this is not the case, because of the presence of operation symbols.
For instance, our $\ECount$ program from Section~\ref{sec:effectful-counting}
will apply $P$ to the `generic point' $\Superpoint$.
Thus, for example, in our treatment of Lemma~\ref{lem:replacement}(i),
it need no longer be the case that the reduction of $Q[p]~k$ either yields a value
or gets stuck at some $\EC_0[p~Q_0[p],p]$: a third possibility is that it gets stuck
at some invocation of $\ell$, so that control will then pass to the effect handler.

\section{Extensions and Variations}
\label{sec:robustness}

Our complexity result is robust in that it continues to hold in more
general settings. We outline here how it generalises: beyond
$n$-standard predicates, from generic count to generic search, and
from pure $\BCalc$ to stateful $\BCalcS$.

\subsection{Beyond $n$-Standard Predicates}
\label{sec:beyond}
The $n$-standard restriction on predicates serves to make the
efficiency phenomenon stand out as clearly as possible. However, we
can relax the restriction by tweaking $\ECount$ to handle repeated
queries and missing queries.
The trade off is that the analysis of $\ECount$ becomes more involved.
The key to relaxing the $n$-standard restriction is the use of state
to keep track of which queries have been computed.
We can give stateful implementations of $\ECount$ without changing its
type signature by using \emph{parameter-passing}~\citep{KammarLO13,
  Pretnar15} to internalise state within a handler.
Parameter-passing abstracts every handler clause such that the current
state is supplied before the evaluation of a clause continues and the
state is threaded through resumptions: a resumption becomes a
two-argument curried function $r : B \to S \to D$, where the first
argument of type $B$ is the return type of the operation and the
second argument is the updated state of type $S$.

\paragraph{Repeated queries} We can generalise $\ECount$ to handle
repeated queries by memoising previous answers. First, we generalise
the type of $\Branch$ such that it carries an index of a query.
{\small
\[
  \Branch : \Nat \to \Bool
\]}
We assume a family of natural number to boolean maps, $\dec{Map}_n$
with the following interface.
{\small
  \begin{equations}
  \dec{empty}_n  &:& \dec{Map}_n \\
  \dec{add}_n    &:& (\Nat_n \times \Bool) \to \dec{Map}_n \to \dec{Map}_n \\
  \dec{lookup}_n &:& \Nat_n \to \dec{Map}_n \to (\One + \Bool) \\
  \end{equations}}%
Invoking $\dec{lookup}~i~map$ returns $\Inl~\Unit$
if $i$ is not present in $map$, and $\Inr~ans$ if $i$ is
associated by $map$ with the value $ans : \Bool$.
Allowing ourselves a few extra constant-time arithmetic operations, we
can realise suitable maps in $\BCalc$ such that the time complexity of
$\dec{add}_n$ and $\dec{lookup}_n$ is
$\BigO(\log n)$~\cite{Okasaki99}.  We can then use parameter-passing
to support repeated queries as follows.
{\small
\[
  \bl
    \ECount'_n : ((\Nat_n \to \Bool) \to \Bool) \to \Nat\\
    \ECount'_n\,pred \defas
      \bl
      \Let\; h \revto \Handle\; pred\,(\lambda i. \Do\; \Branch~i)\; \With\\
      \quad\ba[t]{@{}l@{\hspace{1.5ex}}c@{\hspace{1.5ex}}l@{}}
             \Val\, x          &\mapsto& \lambda s. \If\; x\; \Then\; 1 \;\Else\; 0 \\
             \Branch\,i\,\,r &\mapsto&
               \ba[t]{@{}l}\lambda s.
                 \ba[t]{@{}l}
                 \Case\; \dec{lookup}_n~i~s\; \{\\
                   \quad\ba[t]{@{~}l@{~}c@{~}l}
                         \Inl\, \Unit &\mapsto&
                           \ba[t]{@{}l}
                           \Let\;x_\True \revto  r~\True~(\dec{add}_n~\Record{i, \True}~s)\; \In\\
                           \Let\;x_\False \revto r~\False~(\dec{add}_n~\Record{i, \False}~s)\; \In\\
                             (x_\True + x_\False); \\
                           \ea\\
                         \Inr\,x &\mapsto& r~x~s\; \} \\
                         \ea \\
                 \ea \\
               \ea \\
           \ea\\
      \In\;h~\dec{empty}_n \\
      \el \\
  \el
\]}%
The state parameter $s$ memoises query results, thus avoiding
double-counting and enabling $\ECount'_n$ to work correctly for
predicates performing the same query multiple times.

\paragraph{Missing queries}
Similarly, we can use parameter-passing to support missing queries.
{\small
\[
  \bl
    \ECount''_n : ((\Nat_n \to \Bool) \to \Bool) \to \Nat\\
    \ECount''_n\,pred \defas
      \bl
      \Let\;h \revto \bl
                     \Handle\;pred\,(\lambda i. \Do\;\Branch~\Unit)\;\With\\
                     \quad
                       \ba[t]{@{}l@{\hspace{1.5ex}}c@{\hspace{1.5ex}}l@{}}
                         \Val~x &\mapsto& \lambda d.
                           \ba[t]{@{}l}
                             \Let\; result \revto \If\;x\;\Then\;1\;\Else\;0\;\In\;
                             result \times 2^{n - d}\\
                           \ea\\
                         \Branch~\Unit~r &\mapsto& \lambda d.
                           \ba[t]{@{}l}
                             \Let\;x_\True \revto r~\True~(d+1)\;\In\\
                             \Let\;x_\False \revto r~\False~(d+1)\;\In\\
                             (x_\True + x_\False)
                           \ea
                       \ea\\
                     \el \\
      \In\;h~0 \\
      \el \\
  \el
\]}%
The parameter $d$ tracks the depth and the returned result is scaled
by $2^{n-d}$ accounting for the unexplored part of the current
subtree.
This enables $\ECount''_n$ to operate correctly on predicates that
inspect $n$ points at most once.
We leave it as an exercise for the reader to combine $\ECount'_n$ and
$\ECount''_n$ in order to handle both repeated queries and missing
queries.

\subsection{From Generic Count to Generic Search}
\label{sec:count-vs-search}

We can generalise the problem of generic counting to generic
searching. The main operational difference is that a generic search
procedure must materialise a list of solutions, thus its type is
{\small
\[
  \mathsf{search}_{n} : ((\Nat_n \to \Bool) \to \Bool) \to \List_{\Nat_n \to \Bool}
\]}%
where $\List_A$ is the type of cons-lists whose elements have type
$A$.
We modify $\ECount$ to return a list of solutions rather than the
number of solutions by lifting each result into a singleton list and
using list concatenation instead of addition to combine partial
results $xs_\True$ and $xs_\False$ as follows.
\newcommand{\ESearch}{\mathsf{effsearch}}
\newcommand{\Singleton}{\mathsf{singleton}}
\newcommand{\Concat}{\mathsf{concat}}
\newcommand{\HughesList}{\mathsf{HList}}
\newcommand{\ToConsList}{\mathsf{toConsList}}
{\small
\[
  \bl
    \ESearch_n : ((\Nat_n \to \Bool) \to \Bool) \to \List_{\Nat_n \to \Bool}\\
    \ESearch_n\,pred \defas
         \bl\Let\; f \revto \bl
                             \Handle\; pred\,(\lambda i. \Do\; \Branch~i)\; \With\\
                             \ba[t]{@{}l@{\hspace{1.5ex}}c@{\hspace{1.5ex}}l@{}}
                                           \Val\, x        &\mapsto& \lambda q. \If\, x \;\Then\; \Singleton~q \;\Else\; \dec{nil} \\
                                           \Branch\,i\,\,r &\mapsto&
                                             \ba[t]{@{}l}\lambda q.
                                                \ba[t]{@{}l}
                                                 \Let\;xs_\True \revto  r~\True~(\lambda j.\If\;i=j\;\Then\;\True\;\Else\;q~j) \;\In\\
                                                 \Let\;xs_\False \revto r~\False~(\lambda j. \If\;i=j\;\Then\;\False\;\Else\;q~j) \;\In\\
                                                 \Concat~\Record{xs_\True,xs_\False} \\
                                                 \ea\\
                                             \ea\\
                              \ea \\
                              \el \\
             \In\;\ToConsList~(f~(\lambda j. \bot))
          \el \\
  \el
\]}%
The $\Branch$ operation is now parameterised by an index $i$.
The handler is now parameterised by the current path as a point $q$,
which is output at a leaf iff it is in the predicate.
A little care is required to ensure that $\ESearch_n$ has runtime
$\BigO(2^n)$; \naive use of cons-list concatenation would result in
$\BigO(n2^n)$ runtime, as cons-list concatenation is linear in its
first operand. In place of cons-lists we use Hughes
lists~\citep{Hughes86}, which admit constant time concatenation:
$\HughesList_A \defas \List_A \to \List_A$. The empty Hughes list
$\dec{nil} : \HughesList_A$ is defined as the identity function:
$\dec{nil} \defas \lambda xs. xs$.
{\small
  \[
  \ba{@{}l@{\qquad}l@{\qquad}l}
     \Singleton_A : A \to \HughesList_A
   & \Concat_A : \HughesList_A \times \HughesList_A \to \HughesList_A
   & \ToConsList_A : \HughesList \to \List_A\\
     \Singleton_A~x \defas \lambda xs. x \cons xs
   & \Concat_A~f\,g \defas \lambda xs. g~(f~xs)
   & \ToConsList_A~f \defas f~\nil
   \ea
\]}%
We use the function $\ToConsList$ to convert the final Hughes list to
a standard cons-list at the end; this conversion has linear time
complexity (it just conses all of the elements of the list together).

\subsection{From Pure $\BCalc$ to Stateful $\BCalcS$}
\label{sec:mutable-state}

Mutable state is a staple ingredient of many practical programming
languages.  We now outline how our main lower bound result can be
extended to a language with state.  We will not give full details, but
merely point out the respects in which our earlier treatment needs to
be modified.

We have in mind an extension $\BCalcS$ of $\BCalc$ with ML-style
reference cells: we extend our grammar for types with a reference type
($\Ref~A$), and that for computation terms with forms for creating
references ($\keyw{letref}\; x = V\; \In\; N$), dereferencing ($!x$),
and destructive update ($x := V$), with the familiar typing rules.  We
also add a new kind of value, namely \emph{locations} $l^A$, of type
$\Ref~A$. We adopt a basic Scott-Strachey~\citeyearpar{ScottS71} model
of store: a location is a natural number decorated with a type, and
the execution of a stateful program allocates locations in the order
$0,1,2,\ldots$, assigning types to them as it does so. A \emph{store}
$s$ is a type-respecting mapping from some set of locations $\{
0,\ldots,l-1 \}$ to values.  For the purposes of small-step
operational semantics, a \emph{configuration} will be a triple
$(M,l,s)$, where $M$ is a computation, $l$ is a `location counter',
and $s$ is a store with domain $\{ 0,\ldots,l-1 \}$. A reduction
relation $\reducesto$ on configurations is defined in a familiar way
(again we omit the details).

Certain aspects of our setup require care in the presence of state.
For instance, there is in general no unique way to assign an (untimed)
decision tree to a closed value $P : \Predicate_n$, since the
behaviour of $P$ on a value $q : \Point_n$ may depend both on the
initial state when $P$ is invoked, and on the ways in which the
associated computations $q~V \reducesto^\ast \Return\;W$ modify the
state.  In this situation, there is not even a clear specification for
what an $n$-count program ought to do.

The simplest way to circumvent this difficulty is to restrict
attention to predicates $P$ \emph{within the sublanguage $\BCalc$}.
For such predicates, the notions of decision tree, counting and
$n$-standardness are unproblematic. Our result will establish a
runtime lower bound of $\Omega(n2^n)$ for programs $\Countprog \in
\BCalcS$ that correctly count predicates $P$ of this kind.
On the other hand, since $\Countprog$ itself may be stateful, we
cannot exclude the possibility that $\Countprog~P$ will apply $P$ to a
term $Q$ that is itself stateful. Such a $Q$ will no longer
unambiguously denote a semantic point $\pi$, hence the proof of
Section~\ref{sec:pure-counting} must be adapted.

To adapt our proof to the setting of $\BCalcS$, some more machinery is
needed.  If $\Countprog$ is an $n$-count program and $P$ an
$n$-standard predicate, we expect that the evaluation of
$\Countprog~P$ will feature terms $\EC[P~Q]$ which are then reduced to
some $\EC[\Return\;b]$, via a reduction sequence which, modulo
$\EC[-]$, has the following form:
{\small
\[ P\,Q \reducesto^\ast \EC_0[Q~k_0] \reducesto^\ast \EC_0[\Return\,b_0] \reducesto^\ast \cdots
   \reducesto^\ast \EC_{n-1}[Q~k_{n-1}] \reducesto^\ast \EC_{n-1}[\Return\,b_{n-1}]
   \reducesto^\ast \Return\;b
\]}%
(For notational clarity, we suppress mention of the location and store
components here.)  Informally we think of this as a dialogue in
which control passes back and forth between $P$ and $Q$. We shall
refer to the portions $\EC_j[Q~k_j] \reducesto^\ast
\EC_j[\Return\;b_j]$ of the above reduction as \emph{$Q$-sections},
and to the remaining portions (including the first and the last) as
\emph{$P$-sections}. We refer to the totality of these $P$-sections
and $Q$-sections as the \emph{thread} arising from the given
occurrence of the application $P\,Q$.  An important point to note is
that since $Q$ may contain other occurrences of $P$, it is quite
possible for the $Q$-sections above to contain further threads
corresponding to other applications $P~Q'$.

Since $P$ is $n$-standard, we know that each thread will consist of
$n+1$ $P$-sections separated by $n$ $Q$-sections.
Indeed, it is clear that this computation traces the path
$b_0 \ldots b_{n-1}$ through the decision tree for $P$, with
$k_0,\ldots,k_{n-1}$ the corresponding internal node labels.  We may
now, `with hindsight', construe this as a semantic point
$\pi : \N_n \to \B$ (where $\pi(k_j)=b_j$ for each $j$), and call it
the semantic point \emph{associated with} (the thread arising from)
the application occurrence $P~p$.

The following lemma now serves as a surrogate for
Lemma~\ref{lem:no-shortcuts}:

\begin{lemma}
  Let $P$ be an $n$-standard predicate. For any semantic point
  $\pi \in \B^n$, the evaluation of $\Countprog~P$ involves an
  application occurrence $P~Q$ with which $\pi$ is associated.
\end{lemma}
The proof of this lemma is not too different from that of
Lemma~\ref{lem:no-shortcuts}: if $\pi$ were a point with no associated
thread, there would be an unvisited leaf in the decision tree, and we
could manufacture an $n$-standard predicate $P'$ whose tree differed
from that of $P$ only at this leaf. We can then show, by induction on
length of reductions, that any portion of the evaluation of
$\Countprog~P$ can be suitably mimicked with $P$ replaced by $P'$.
Naturally, this idea now needs to be formulated at the level of
\emph{configurations} rather than plain terms: in the course of
reducing $(\Countprog~P,0,[])$, we may encounter configurations
$(M,l,s)$ in which residual occurrences of $P$ have found their way
into $s$ as well as $M$, so in order to replace $P$ by $P'$ we must
abstract on all these occurrences via an evident notion of
\emph{configuration context}.  With this adjustment, however, the
argument of Lemma~\ref{lem:no-shortcuts} goes through.

A further argument is then needed to show that any two threads are
indeed ‘disjoint’ as regards their $P$-sections, so that there must be
at least $n2^n$ steps in the overall reduction sequence.

\newcommand{\tooslow}{-}

\newcommand{\tableone}
{\begin{table*}
    \footnotesize
  \caption{SML/NJ: Runtime Relative to Effectful Implementation}
  \label{tbl:results}
  \vspace{-2.5ex}
  \begin{tabular}{@{}| l | r@{\,} | r@{\,} | r@{\,} |@{\,}| r@{\,} | r@{\,} | r@{\,} |@{\,}| r@{\,} |@{\,}| r@{\,} | r@{\,} | r@{\,} |@{\,}| r@{\,} | r@{\,} | r@{\,} | r@{\,} | r@{\,} |@{}}
    \cline{2-16}
    \multicolumn{1}{l |}{} &
    \multicolumn{6}{@{}c@{} |@{\,}|}{\textbf{Queens}} &
    \multicolumn{9}{@{}c@{} |}{\textbf{Integration}}
    \\\cline{2-16}
    \multicolumn{1}{c |}{} &
    \multicolumn{3}{| @{}c@{} |@{\,}|}{\textbf{First solution}} &
    \multicolumn{3}{| @{}c@{} |@{\,}|}{\textbf{All solutions}} &
    \multicolumn{1}{@{}c@{} |@{\,}|}{\textbf{Id}} &
    \multicolumn{3}{  @{}c@{} |@{\,}|}{\textbf{Squaring}} &
    \multicolumn{5}{  @{}c@{} |}{\textbf{Logistic}}
    \\\cline{2-16}

    \multicolumn{1}{c |}{\emph{Parameter\!\!}} &
    \multicolumn{1}{@{}c@{} |}{$20$} &
    \multicolumn{1}{@{}c@{} |}{$24$} &
    \multicolumn{1}{@{}c@{} |@{\,}|}{$28$} &
    \multicolumn{1}{@{}c@{} |}{$8$} &
    \multicolumn{1}{@{}c@{} |}{$10$} &
    \multicolumn{1}{@{}c@{} |@{\,}|}{$12$} &
    \multicolumn{1}{@{}c@{} |@{\,}|}{$20$} &
    \multicolumn{1}{@{}c@{} |}{$14$} &
    \multicolumn{1}{@{}c@{} |}{$17$} &
    \multicolumn{1}{@{}c@{} |@{\,}|}{$20$} &
    \multicolumn{1}{@{}c@{} |}{$1$} &
    \multicolumn{1}{@{}c@{} |}{$2$} &
    \multicolumn{1}{@{}c@{} |}{$3$} &
    \multicolumn{1}{@{}c@{} |}{$4$} &
    \multicolumn{1}{@{}c@{} |}{$5$}
    \\\hline

    \Naive   &
    $\tooslow$ &
    $\tooslow$ &
    $\tooslow$ &
    $\!\!217.74$ &
    $\tooslow$ &
    $\tooslow$ &
    $\!\!12.89$ &
    $\!\!45.04$ &
    $\!\!57.80$ &
    $\!\!69.86$ &
    $\tooslow$ &
    $\tooslow$ &
    $\tooslow$ &
    $\tooslow$ &
    $\tooslow$
    \\\hline

    Berger &
    $11.24$ &
    $15.70$ &
    $\tooslow$ &
    $2.06$ &
    $2.86$ &
    $3.64$ &
    $5.18$ &
    $\!\!20.62$ &
    $\!\!22.37$ &
    $\!\!23.46$ &
    $22.51$ &
    $28.97$ &
    $30.14$ &
    $29.30$ &
    $27.94$
    \\\hline

    Pruned &
    $2.13$ &
    $2.54$ &
    $2.91$ &
    $1.04$ &
    $1.24$ &
    $1.39$ &
    $2.07$ &
    $3.78$ &
    $4.05$ &
    $4.24$ &
    $4.10$ &
    $5.44$ &
    $6.42$ &
    $7.26$ &
    $7.94$
    \\\cline{1-16}

    Bespoke &
    $0.12$ &
    $0.12$ &
    $0.12$ &
    $0.13$ &
    $0.13$ &
    $0.12$ &
    \multicolumn{9}{l}{}
    \\\cline{1-7}
  \end{tabular}
\end{table*}}

\newcommand{\tabletwo}
{\begin{table*}
  \footnotesize
  \caption{MLton: Runtime Relative to Effectful Implementation}
  \label{tbl:results-mlton}
  \vspace{-2.5ex}
  \begin{tabular}{@{}| l | r@{\,} | r@{\,} | r@{\,} |@{\,}| r@{\,} | r@{\,} | r@{\,} |@{\,}| r@{\,} |@{\,}| r@{\,} | r@{\,} | r@{\,} |@{\,}| r@{\,} | r@{\,} | r@{\,} | r@{\,} | r@{\,} |@{}}
    \cline{2-16}
    \multicolumn{1}{l |}{} &
    \multicolumn{6}{@{}c@{} |@{\,}|}{\textbf{Queens}} &
    \multicolumn{9}{@{}c@{} |}{\textbf{Integration}}
    \\\cline{2-16}
    \multicolumn{1}{c |}{} &
    \multicolumn{3}{| @{}c@{} |@{\,}|}{\textbf{First solution}} &
    \multicolumn{3}{| @{}c@{} |@{\,}|}{\textbf{All solutions}} &
    \multicolumn{1}{@{}c@{} |@{\,}|}{\textbf{Id}} &
    \multicolumn{3}{  @{}c@{} |@{\,}|}{\textbf{Squaring}} &
    \multicolumn{5}{  @{}c@{} |}{\textbf{Logistic}}
    \\\cline{2-16}

    \multicolumn{1}{c |}{\emph{Parameter\!\!}} &
    \multicolumn{1}{@{}c@{} |}{$20$} &
    \multicolumn{1}{@{}c@{} |}{$24$} &
    \multicolumn{1}{@{}c@{} |@{\,}|}{$28$} &
    \multicolumn{1}{@{}c@{} |}{$8$} &
    \multicolumn{1}{@{}c@{} |}{$10$} &
    \multicolumn{1}{@{}c@{} |@{\,}|}{$12$} &
    \multicolumn{1}{@{}c@{} |@{\,}|}{$20$} &
    \multicolumn{1}{@{}c@{} |}{$14$} &
    \multicolumn{1}{@{}c@{} |}{$17$} &
    \multicolumn{1}{@{}c@{} |@{\,}|}{$20$} &
    \multicolumn{1}{@{}c@{} |}{$1$} &
    \multicolumn{1}{@{}c@{} |}{$2$} &
    \multicolumn{1}{@{}c@{} |}{$3$} &
    \multicolumn{1}{@{}c@{} |}{$4$} &
    \multicolumn{1}{@{}c@{} |}{$5$}
    \\\hline

    \Naive   &
    $\tooslow$ &
    $\tooslow$ &
    $\tooslow$ &
    $17.31$ &
    $\tooslow$ &
    $\tooslow$ &
    $1.45$ &
    $4.51$ &
    $5.13$ &
    $5.82$ &
    $\tooslow$ &
    $\tooslow$ &
    $\tooslow$ &
    $\tooslow$ &
    $\tooslow$
    \\\hline

    Berger &
    $0.52$ &
    $0.66$ &
    $\tooslow$ &
    $0.19$ &
    $0.22$ &
    $0.20$ &
    $0.43$ &
    $2.02$ &
    $1.95$ &
    $1.92$ &
    $2.17$ &
    $3.59$ &
    $4.24$ &
    $4.34$ &
    $4.28$
    \\\hline

    Pruned &
    $0.11$ &
    $0.11$ &
    $0.13$ &
    $0.10$ &
    $0.10$ &
    $0.08$ &
    $0.14$ &
    $0.39$ &
    $0.35$ &
    $0.35$ &
    $0.39$ &
    $0.63$ &
    $0.86$ &
    $1.03$ &
    $1.21$
    \\\cline{1-16}

    Bespoke &
    $0.005$ &
    $0.004$ &
    $0.004$ &
    $0.01$ &
    $0.009$ &
    $0.006$ &
    \multicolumn{9}{l}{}
    \\\cline{1-7}
  \end{tabular}
\end{table*}}

\newcommand{\tablethree}
{\begin{table*}
  \caption{MLton: Runtime Relative to SML/NJ}
  \label{tbl:results-mlton-vs-smlnj}
  \vspace{-2.5ex}
\footnotesize
  \begin{tabular}{@{}| l | r@{\,} | r@{\,} | r@{\,} |@{\,}| r@{\,} | r@{\,} | r@{\,} |@{\,}| r@{\,} |@{\,}| r@{\,} | r@{\,} | r@{\,} |@{\,}| r@{\,} | r@{\,} | r@{\,} | r@{\,} | r@{\,} |@{}}
    \cline{2-16}
    \multicolumn{1}{l |}{} &
    \multicolumn{6}{@{}c@{} |@{\,}|}{\textbf{Queens}} &
    \multicolumn{9}{@{}c@{} |}{\textbf{Integration}}
    \\\cline{2-16}
    \multicolumn{1}{c |}{} &
    \multicolumn{3}{| @{}c@{} |@{\,}|}{\textbf{First solution}} &
    \multicolumn{3}{| @{}c@{} |@{\,}|}{\textbf{All solutions}} &
    \multicolumn{1}{@{}c@{} |@{\,}|}{\textbf{Id}} &
    \multicolumn{3}{  @{}c@{} |@{\,}|}{\textbf{Squaring}} &
    \multicolumn{5}{  @{}c@{} |}{\textbf{Logistic}}
    \\\cline{2-16}

    \multicolumn{1}{c |}{\emph{Parameter\!\!}} &
    \multicolumn{1}{@{}c@{} |}{$20$} &
    \multicolumn{1}{@{}c@{} |}{$24$} &
    \multicolumn{1}{@{}c@{} |@{\,}|}{$28$} &
    \multicolumn{1}{@{}c@{} |}{$8$} &
    \multicolumn{1}{@{}c@{} |}{$10$} &
    \multicolumn{1}{@{}c@{} |@{\,}|}{$12$} &
    \multicolumn{1}{@{}c@{} |@{\,}|}{$20$} &
    \multicolumn{1}{@{}c@{} |}{$14$} &
    \multicolumn{1}{@{}c@{} |}{$17$} &
    \multicolumn{1}{@{}c@{} |@{\,}|}{$20$} &
    \multicolumn{1}{@{}c@{} |}{$1$} &
    \multicolumn{1}{@{}c@{} |}{$2$} &
    \multicolumn{1}{@{}c@{} |}{$3$} &
    \multicolumn{1}{@{}c@{} |}{$4$} &
    \multicolumn{1}{@{}c@{} |}{$5$}
    \\\hline

    \Naive   &
    $\tooslow$ &
    $\tooslow$ &
    $\tooslow$ &
    $0.49$ &
    $\tooslow$ &
    $\tooslow$ &
    $0.55$ &
    $0.35$ &
    $0.35$ &
    $0.35$ &
    $\tooslow$ &
    $\tooslow$ &
    $\tooslow$ &
    $\tooslow$ &
    $\tooslow$
    \\\hline

    Berger &
    $0.62$ &
    $0.64$ &
    $\tooslow$ &
    $0.73$ &
    $0.65$ &
    $0.68$ &
    $0.41$ &
    $0.35$ &
    $0.34$ &
    $0.34$ &
    $0.37$ &
    $0.37$ &
    $0.37$ &
    $0.37$ &
    $0.37$
    \\\hline

    Pruned &
    $0.70$ &
    $0.68$ &
    $0.71$ &
    $0.74$ &
    $0.70$ &
    $0.71$ &
    $0.34$ &
    $0.36$ &
    $0.35$ &
    $0.35$ &
    $0.36$ &
    $0.35$ &
    $0.35$ &
    $0.35$ &
    $0.36$
    \\\hline

    Effectful &
    $12.87$ &
    $13.99$ &
    $14.90$ &
    $8.00$ &
    $8.60$ &
    $12.19$ &
    $4.93$ &
    $3.53$ &
    $3.95$ &
    $4.20$ &
    $3.80$ &
    $3.00$ &
    $2.62$ &
    $2.46$ &
    $2.37$
    \\\cline{1-16}

    Bespoke &
    $0.56$ &
    $0.56$ &
    $0.56$ &
    $0.69$ &
    $0.63$ &
    $0.59$ &
    \multicolumn{9}{l}{}
    \\\cline{1-7}
  \end{tabular}
\end{table*}}

\tableone
\tabletwo

\section{Experiments}
\label{sec:experiments}
The theoretical efficiency gap between realisations of $\BCalc$ and
$\HCalc$ manifests in practice. We observe it empirically on
instantiations of $n$-queens and exact real number integration, which
can be cast as generic search. Table~\ref{tbl:results} shows the
speedup of using an effectful implementation of generic search over
various pure implementations. We discuss the benchmarks and results in
further detail below.

\setlength{\floatsep}{1.0ex}
\setlength{\textfloatsep}{1.0ex}

\paragraph{Methodology}
We evaluated an effectful implementation of generic search against
three ``pure'' implementations which are realisable in $\BCalc$
extended with mutable state:
\begin{itemize}
\item \Naive: a simple, and rather \naive, functional implementation;
\item Pruned: a generic search procedure with space pruning based on
  Longley's technique~\cite{Longley99} (uses local state);
\item Berger: a lazy pure functional generic search procedure based on
  Berger's algorithm.
\end{itemize}
Each benchmark was run 11 times. The reported figure is the median
runtime ratio between the particular implementation and the baseline
effectful implementation. Benchmarks that failed to terminate within a
threshold (1 minute for single solution, 8 minutes for enumerations),
are reported as $\tooslow$. The experiments were conducted in
\citet{smlnj} v110.97 64-bit with factory settings on an Intel Xeon
CPU E5-1620 v2 @ 3.70GHz powered workstation running Ubuntu 16.04. The
effectful implementation uses an encoding of delimited control akin to
effect handlers based on top of SML/NJ's call/cc.
The complete source code for the benchmarks is available at:
\begin{center}
  \url{https://github.com/dhil/effects-for-efficiency-code}
\end{center}

\paragraph{Queens}
We phrase the $n$-queens problem as a generic search problem. As a
control we include a bespoke implementation hand-optimised for the
problem. We perform two experiments: finding the first solution for $n
\in \{20,24,28\}$ and enumerating all solutions for $n \in
\{8,10,12\}$. The speedup over the \naive implementation is dramatic,
but less so over the Berger procedure. The pruned procedure is more
competitive, but still slower than the baseline. Unsurprisingly, the
baseline is slower than the bespoke implementation.

\paragraph{Exact Real Integration}
The integration benchmarks are adapted from \citet{Simpson98}. We
integrate three different functions with varying precision in the
interval $[0,1]$. For the identity function (Id) at precision $20$ the
speedup relative to Berger is $5.18\times$. For the squaring function
the speedups are larger at higher precisions: at precision $14$ the
speedup is $3.78\times$ over the pruned integrator, whilst it is
$4.24\times$ at precision $20$. The speedups are more extreme against
the \naive and Berger integrators. We also integrate the logistic map
$x \mapsto 1 - 2x^2$ at a fixed precision of $15$. We make the
function harder to compute by iterating it up to $5$ times. Between
the pruned and effectful integrator the speedup ratio increases as the
function becomes harder to compute.

\paragraph{MLton}
SML/NJ is compiled into CPS, thus providing a particularly efficient
implementation of call/cc.
\citet{mlton}, a whole program compiler for SML, implements
call/cc by copying the stack.
We repeated our experiments using MLton 20180207.
Table~\ref{tbl:results-mlton} shows the results. The effectful
implementation performs much worse under MLton than SML/NJ, being
surpassed in nearly every case by the pruned search procedure and in
some cases by the Berger search procedure.
Table~\ref{tbl:results-mlton-vs-smlnj} summarises the runtime of MLton
relative to SML/NJ. Berger, Pruned, and Bespoke run between 1 and 3
times as fast with MLton compared to SML/NJ.
However, the effectful implementation runs between 2 and 14 times as
fast with SML/NJ compared with MLton.

\tablethree

\section{Conclusions and Future Work}
\label{sec:conclusions}
We presented a PCF-inspired language $\BCalc$ and its extension with
effect handlers $\HCalc$. We proved that $\HCalc$ supports an
asymptotically more efficient implementation of generic search than
any possible implementation in $\BCalc$. We observed its effect in
practice on several benchmarks.
We also proved that our $\Omega(n2^n)$ lower bound applies to a
language $\BCalcS$ which extends $\BCalc$ with state.

Our positive result for $\HCalc$ extends to other control operators by appeal to existing
results on interdefinability of handlers and other control
operators~\citep{ForsterKLP19,PirogPS19}.
The result no longer applies directly if we add an effect type system
to $\HCalc$, as the implementation of the counting program would
require a change of type for predicates to reflect the ability to
perform effectful operations.
In future we plan to investigate how to account for effect type systems.

We have verified that our $\Omega(n2^n)$ lower bound also applies to
a language $\BCalcE$ with (Benton-Kennedy style~\citep{BentonK01})
\emph{exceptions} and handlers.
The lower bound also applies to the combined language $\BCalcSE$
with both state and exceptions --- this seems to bring us close to
the expressive power of real languages such as Standard ML, Java, and
Python, strongly suggesting that the speedup we have discussed is
unattainable in these languages.

In future work, we hope to establish the more general result that our
$\Omega(n2^n)$ applies to a language with \emph{affine effect
  handlers} (handlers which invoke the resumption $r$ at most once).
This would not only subsume our present results (since state and
exceptions are examples of affine effects), but would also apply
e.g.\ to a richer language with \emph{coroutines}.  However, it
appears that our present methods do not immediately adapt to this more
general situation, as our arguments depend at various points on an
orderly nesting of subcomputations which coroutining would break.

One might object that the efficiency gap we have analysed is of merely
theoretical interest, since an $\Omega(2^n)$ runtime is already
`infeasible'.  We claim, however, that what we have presented is an
example of a much more pervasive phenomenon, and our generic count
example serves merely as a convenient way to bring this phenomenon
into sharp formal focus. Suppose, for example, that our programming
task was not to count all solutions to $P$, but to find just one of
them.  It is informally clear that for many kinds of predicates this
would in practice be a feasible task, and also that we could still
gain our factor $n$ speedup here by working in a language with
first-class control.  However, such an observation appears less
amenable to a clean mathematical formulation, as the runtimes in
question are highly sensitive to both the particular choice of
predicate and the search order employed.

\begin{acks}                            
  We would like to thank James McKinna and Maciej Pir\'og for
  insightful discussions, and Danel Ahman and the anonymous reviewers
  for helpful feedback and suggestions for improvement.
  Daniel Hillerström was supported by
  \href{https://www.epsrc.ac.uk/}{EPSRC} grant
  \href{http://pervasiveparallelism.inf.ed.ac.uk}{EP/L01503X/1}
  and by ERC Consolidator Grant Skye (grant number 682315).
  Sam Lindley was supported by \href{https://www.epsrc.ac.uk}{EPSRC}
  grant \href{http://groups.inf.ed.ac.uk/abcd/}{EP/K034413/1} (From
  Data Types to Session Types---A Basis for Concurrency and
  Distribution).
\end{acks}

\bibliography{generic-search}

\ifdefined\EXTENDED
\clearpage
\fi

\appendix
\extendedsection{Correctness of the Base Machine}
\label{sec:base-machine-correctness}
\ifdefined\EXTENDED

We now show that the base abstract machine is correct with respect to
the operational semantics, that is, the abstract machine faithfully
simulates the operational semantics.
Initial states provide a canonical way to map a computation term onto
the abstract machine.
A more interesting question is how to map an arbitrary configuration
to a computation term.
Figure~\ref{fig:config-to-term} describes such a mapping $\inv{-}$
from configurations to terms via a collection of mutually recursive
functions defined on configurations, continuations, computation terms,
value terms, and machine values. The mapping makes use of two
operations on environments, $\gamma$, which we define now.

\begin{definition}
\begin{sloppypar}
We write $dom(\gamma)$ for the domain of $\gamma$, and $\gamma \res
\{x_1, \dots, x_n\}$ for the restriction of environment $\gamma$ to
$dom(\gamma) \res \{x_1, \dots, x_n\}$.
\end{sloppypar}
\end{definition}

\renewcommand{\contapp}[2]{#1 #2}
\renewcommand{\contappp}[2]{#1(#2)}
\begin{figure*}
\small
\raggedright

\begin{minipage}{0.4\textwidth}
\textbf{Configurations}
\begin{mathpar}
\inv{\cek{M \mid \env \mid \sigma}} = \contappp{\inv{\sigma}}{\inv{M}\env}
\end{mathpar}
\end{minipage}
\begin{minipage}{0.5\textwidth}
\textbf{Pure continuations}
\begin{equations}
\contapp{\inv{[]}}{M}
  &=& M \\
\contapp{\inv{(\env, x, N) \cons \sigma}}{M}
  &=& \contappp{\inv{\sigma}}{\Let\; x \revto M \;\In\; \inv{N}(\env \res \{x\})} \\
\end{equations}
\end{minipage}

\textbf{Computation terms}
\begin{equations}
\inv{V\,W}\env &=& \inv{V}\env\,\inv{W}{\env} \\
\inv{\Let\;\Record{x; y} = V \;\In\;N}\env
  &=& \Let\;\Record{x; y} =\inv{V}\env \;\In\; \inv{N}(\env \res \{x, y\}) \\
\hspace{2cm}\inv{\Case\;V\,\{\Inl\;x \mapsto M; \Inr\;y \mapsto N\}}\env
  &=& \Case\;\inv{V}\env \,\{\bl
                             \Inl\;x \mapsto \inv{M}(\env \res \{x\}); \\
                             \Inr\;y \mapsto \inv{N}(\env \res \{y\})\} \\
                             \el \\
\inv{\Return\;V}\env &=& \Return\;\inv{V}\env \\
\inv{\Let\;x \revto M \;\In\;N}\env
  &=& \Let\;x \revto\inv{M}\env \;\In\; \inv{N}(\env \res \{x\}) \\
\end{equations}

\textbf{Value terms and values}
\begin{mathpar}
\begin{eqs}
\inv{x}\env                  &=& \inv{v}, \quad \text{ if }\env(x) = v \\
\inv{x}\env                  &=& x, \quad \text{ if }x \notin dom(\env) \\
\inv{n}\env                  &=& n \\
\inv{\lambda x^A.M}\env      &=& \lambda x^A.\inv{M}(\env \res \{x\}) \\
\inv{\Rec\,f\, x^A.M}\env    &=& \Rec\,f\,x^A.\inv{M}(\env \res \{f,x\}) \\
\inv{\Unit{}}\env          &=& \Unit{} \\
\inv{\Record{V, W}}\env &=& \Record{\inv{V}\env, \inv{W}\env} \\
\inv{(\Inl\;V)^B}\env        &=& (\Inl\;\inv{V}\env)^B \\
\inv{(\Inr\;W)^A}\env        &=& (\Inr\;\inv{W}\env)^A \\
\end{eqs}

\begin{eqs}
\inv{n}                        &=& n \\
\inv{(\env, \lambda x^A.M)}      &=& \lambda x^A.\inv{M}(\env \res \{x\}) \\
\inv{(\env, \Rec\,f\,x^A.M)}     &=& \Rec\,f\, x^A.\inv{M}(\env \res \{f,x\}) \\
\inv{\Unit}                &=& \Unit \\
\inv{\Record{v; w}}       &=& \Record{\inv{v}; \inv{w}} \\
\inv{(\Inl\;v)^B}         &=& (\Inl\;\inv{v})^B \\
\inv{(\Inr\;w)^A}         &=& (\Inr\;\inv{w})^A \\
\inv{\sigma^A} &=& \lambda x^A.\inv{\sigma}(\Return\;x) \\
\end{eqs}
\end{mathpar}

\caption{Mapping from Base Machine Configurations to Terms}
\label{fig:config-to-term}
\end{figure*}

The $\inv{-}$ function enables us to classify the abstract machine
reduction rules according to how they relate to the operational
semantics.
The rule (\mlab{Let}) is administrative in the sense that $\inv{-}$ is
invariant under this rule.
This leaves the $\beta$-rules (\mlab{App}), (\mlab{Split}),
(\mlab{Case}), and (\mlab{RetCont}). Each of these corresponds
directly with performing a reduction in the operational semantics.

\begin{definition}[Auxiliary reduction relations]
We write $\stepsto_{\textrm{a}}$ for administrative steps
($\mlab{Let}$) and $\simeq_{\textrm{a}}$ for the symmetric closure of
$\stepsto_{\textrm{a}}^*$. We write $\stepsto_\beta$ for $\beta$-steps
(all other rules) and $\Stepsto$ for a sequence of steps of the form
$\stepsto_{\textrm{a}}^* \stepsto_\beta$.
\end{definition}

The following lemma describes how we can simulate each reduction in
the operational semantics by a sequence of administrative steps
followed by one $\beta$-step in the abstract machine.
\begin{lemma}
\label{lem:base-simulation}
Suppose $M$ is a computation and $\conf$ is configuration such that
$\inv{\conf} = M$, then if $M \reducesto N$ there exists $\conf'$ such
that $\conf \Stepsto \conf'$ and $\inv{\conf'} = N$, or if
$M \not\reducesto$ then $\conf \not\Stepsto$.
\end{lemma}
\begin{proof}
By induction on the derivation of $M \reducesto N$.
\end{proof}
The correspondence here is rather strong: there is a one-to-one
mapping between $\reducesto$ and
$\Stepsto\mathbin{/}\simeq_{\textrm{a}}$ (where we write $R/S$ for the
quotient of relation $R$ by relation $S$). The inverse of the lemma is
straightforward as the semantics is deterministic.
Notice that Lemma~\ref{lem:base-simulation} does not require that $M$ be
well-typed. We have chosen here not to perform type-erasure, but the
results can be adapted to semantics in which all type annotations are
erased.

\begin{theorem}[Base simulation]\label{thm:base-simulation}
  If $\typ{}{M : A}$ and $M \reducesto^+ N$ such that
  $N$ is normal, then $\cek{M \mid \emptyset \mid \nil} \stepsto^+
  \conf$ such that $\inv{\conf} = N$, or $M
  \not\reducesto$ then $\cek{M \mid \emptyset \mid \nil}
  \not\stepsto$.
\end{theorem}
\begin{proof}
By repeated application of Lemma~\ref{lem:base-simulation}.
\end{proof}
\fi

\extendedsection{Correctness of the Handler Machine}
\label{sec:handler-machine-correctness}
\ifdefined\EXTENDED

The correctness result for the base machine can mostly be repurposed
for the handler machine as we need only recheck the cases for
$(\mlab{Let})$ and $(\mlab{RetCont})$ and check the cases for
handlers. Figure~\ref{fig:config-to-term-eff} shows the necessary
changes to the $\inv{-}$ function.

\begin{lemma}
\label{lem:handler-simulation}
Suppose $M$ is a computation and $\conf$ is configuration such that
$\inv{\conf} = M$, then if $M \reducesto N$ there exists $\conf'$ such
that $\conf \Stepsto \conf'$ and $\inv{\conf'} = N$, or if
$M \not\reducesto$ then $\conf \not\Stepsto$.
\end{lemma}
\begin{proof}
By induction on the derivation of $M \reducesto N$.
\end{proof}

\begin{theorem}[Handler simulation]\label{thm:handler-simulation}
  If $\typ{}{M : A}$ and $M \reducesto^+ N$ such that
  $N$ is normal, then $\cek{M \mid \emptyset \mid \kappa_0} \stepsto^+
  \conf$ such that $\inv{\conf} = N$, or $M
  \not\reducesto$ then $\cek{M \mid \emptyset \mid \kappa_0}
  \not\stepsto$.
\end{theorem}
\begin{proof}
By repeated application of Lemma~\ref{lem:handler-simulation}.
\end{proof}

\begin{figure*}
\small
\raggedright

\begin{minipage}[t]{0.4\textwidth}
\textbf{Configurations}
\begin{mathpar}
\inv{\cek{M \mid \env \mid \kappa}} = \contappp{\inv{\kappa}}{\inv{M}\env}
\end{mathpar}
\end{minipage}
\begin{minipage}[t]{0.5\textwidth}
\textbf{Continuations}
\begin{equations}
\contapp{\inv{[]}}{M}
  &=& M \\
\contapp{\inv{(\sigma, \chi) \cons \kappa}}{M}
  &=& \contappp{\inv{\kappa}}{\contappp{\inv{\chi}}{\contappp{\inv{\sigma}}{M}}}
\end{equations}
\end{minipage}

\textbf{Handler Closures and Definitions}
\begin{mathpar}
\begin{eqs}
  \inv{(\env, H)}M &=& \Handle\;M\;\With\;\inv{H}\env
\end{eqs}
\begin{eqs}
  \inv{\{\Val~x \mapsto M\}}\env &=& \{\Val~x \mapsto \inv{M}(\env \res \{x\}) \}\\
  \inv{\{\ell~x~r \mapsto M\} \uplus H}\env &=& \{\ell~x~r \mapsto \inv{M}(\env \res \{x, r\}) \} \uplus \inv{H}\env
\end{eqs}
\end{mathpar}

\textbf{Computation Terms and Machine Values}
\begin{mathpar}
\begin{eqs}
  \inv{\Handle\;M\;\With\;H}\env &=& \Handle\;\inv{M}\env \;\With\;\inv{H}\env\\
  \inv{\Do\;\ell\,V}\env         &=& \Do\;\ell\,\inv{V}\env
\end{eqs}

\begin{eqs}
  \inv{(\env, H)^D}\env &=& \lambda x^D.\inv{(\env, H)}(\Return\;x)
\end{eqs}
\end{mathpar}
\caption{Mapping from Handler Machine Configurations to Terms}
\label{fig:config-to-term-eff}
\end{figure*}
\fi

\extendedsection{Proof Details for the Complexity of Effectful Generic Count}
\label{sec:positive-theorem}
\ifdefined\EXTENDED
\newcommand{\HCount}{H_{\Count}}
\newcommand{\Henv}{\env_{\HCount}}
\newcommand{\Pure}{\dec{pure}}
\newcommand{\envt}{\dec{env}}
\newcommand{\hclo}{\chi_{\Count}}
\newcommand{\DTDT}{\dec{DT}}
\newcommand{\CF}{\dec{CF}}
\newcommand{\residual}{\dec{residual}}
\newcommand{\cont}{\dec{cont}}
\newcommand{\comp}{\dec{comp}}
\newcommand{\purecont}{\dec{purecont}}
\newcommand{\descend}[1]{\dec{env}^{\downarrow}_{#1}}
\newcommand{\ascend}[1]{\dec{env}^{\uparrow}_{#1}}
\newcommand{\initial}{\dec{env}^{\bot}}
\newcommand{\final}{\dec{env}^{\top}}
\newcommand{\ctrl}{\dec{control}}
\newcommand{\dt}{\mathcal{D}}
\newcommand{\arrive}{\dec{arrive}}
\newcommand{\depart}{\dec{depart}}
\newcommand{\where}[1]{%
  \multicolumn{2}{l}%
  {\text{where } #1}%
}

In this appendix we give proof details and artefacts for
Theorem~\ref{thm:complexity-effectful-counting}. Throughout this
section we let $\HCount$ denote the handler definition of $\Count$,
that is
\[
  \HCount \defas
  \left\{
    \ba[m]{@{~}l@{~}c@{~}l}
      \Val~x           &\mapsto& \If\;x\;\Then\; \Return~1 \;\Else\; \Return~0\\
      \Branch~\Unit~r  &\mapsto&
      \ba[t]{@{}l}
        \Let\; x_\True \revto r~\True \; \In\\
        \Let\; x_\False \revto r~\False \; \In\\
        x_\True + x_\False
      \ea
    \ea
  \right\}
\]
The timed decision tree model embeds timing information. For the proof
we must also know the abstract machine environment and the pure
continuation. Thus we decorate timed decision trees with this
information.
\begin{definition}[decorated timed decision trees]
  A decorated timed decision tree is a partial function $\tree :
  \Addr \pto (\Lab \times \Nat) \times \Conf_q$ such
  that its first projection $bs \mapsto \tree(bs).1$ is a timed
  decision tree.
\end{definition}
We extend the projections $\labs$ and $\steps$ in the obvious way to
work over decorated timed decision trees.  We define three further
projections. The first $\comp(\tree) \defas bs \mapsto \tree(bs).2.1$
projects the computation component of the configuration, the second
$\envt(\tree) \defas bs \mapsto \tree(bs).2.2$ projects the
environment, and finally the third
$\Pure(\tree) \defas bs \mapsto \mathsf{head}(t(bs).2.3).1$ projects
the pure continuation.

The following definition gives a procedure for constructing a
decorated timed decision tree. The construction is analogous to that
of Definition~\ref{def:model-construction}.
\begin{definition}\label{def:model-construction-extended}
 (i) Define $\dt : \Conf_q \pto \Addr \pto (\Lab \times \Nat) \times \Conf_q$ to be the minimal family of
  partial functions satisfying the following equations:
{\small
\begin{mathpar}
\ba{@{}r@{~}c@{~}l@{\qquad}l@{}}
  \dt(\cek{\Return\;W \mid \env \mid \nil})\, \nil  &~=~& ((!b, 0), \cek{\Return\;W \mid \env \mid \nil}),
                                                    &\text{if }\val{W}\env = b \smallskip\\
%
  \dt(\cek{z\,V \mid \env \mid \kappa})\, \nil  &~=~& ((?\val{V}{\env}, 0), \cek{z\,V \mid \env \mid \kappa}),
                                                    & \text{if } \gamma(z) = q \smallskip\\
  \dt(\cek{z\,V \mid \env \mid \kappa})\, (b \cons bs) &~\simeq~& \dt(\cek{\Return\;b \mid \env \mid \kappa})\,bs,
                                                                & \text{if } \gamma(z) = q \smallskip\\
  \dt(\cek{M \mid \env \mid \kappa})\, bs &~\simeq~& \mathsf{inc}\,(\dt(\cek{M' \mid \env' \mid \kappa'})\, bs),
  &\text{if } \cek{M \mid \env \mid \kappa} \stepsto \cek{M' \mid \env' \mid \kappa'}
\ea
\end{mathpar}}%
Here
$\mathsf{inc}((\ell, s), \mathcal{C}) = ((\ell, s + 1), \mathcal{C})$,
and in all of the above equations $\gamma(q) = \gamma'(q) =
q$. Clearly $\dt(\conf)$ is a decorated timed decision tree for any
$\conf \in \Conf_q$.

(ii) The decorated timed decision tree of a computation term is
obtained by placing it in the initial configuration:
$\dt(M) \defas \dt(\cek{M, \emptyset[q \mapsto q], \kappa_0})$.

(iii) The decorated timed decision tree of a closed value
$P:\Predicate$ is $\dt(P\,q)$. Since $q$ plays the role of a dummy
argument, we will usually omit it and write $\dt(P)$ for $\dt(P\,q)$.
\end{definition}

We define some functions, that given a list of booleans and a
$n$-standard predicate, compute configurations of the effectful
abstract machine at particular points of interest during evaluation of
the given predicate. Let
$\hclo(V) \defas (\emptyset[pred \mapsto \val{V}\emptyset], \HCount)$
denote the handler closure of $\HCount$.

\paragraph{Notation.} For an $n$-standard predicate $P$ we write
$|P| = n$ for the size of the predicate. Furthermore, we define
$\chi_{\text{id}}$ for the identity handler closure
$(\emptyset, \{ \Val~x \mapsto x \})$.

\begin{definition}[computing machine configurations]
  For any $n$-standard predicate $P$ and a list of booleans $bs$, such
  that $|bs| \leq n$, we can compute machine configurations at points
  of interest during evaluation of $\Count~P$.

  To make the notation slightly simpler we use the following
  conventions whenever $n$, $\tree$, and $c$ appear free: $n = |P|$,
  $\tree = \dt(P)$, and $c(bs) = \sharp(bs' \mapsto \val{P}~(bs \concat bs'))$.
  The definitions are presented in a top-down manner.
  \begin{itemize}
  \item The function $\arrive$ either computes the configuration at a
    query node, if $|bs| < n$, or the configuration at an answer node.
    \begin{equations}
      \arrive &:& \Addr \times \ValCat \pto \Conf \\
      \arrive(bs, P) &\defas& \cek{z~V \mid \env \mid (\sigma, \hclo(P)) \cons \residual(bs, P)}, \quad \text{if } |bs| < n\\
      \multicolumn{3}{@{}l@{}}
      {\hfill
        \text{where } \ba[t]{@{~}l}
        z~V = \comp(\tree)(bs), \env = \envt(\tree)(bs), \env(z) = (\initial(P), \Superpoint) \\
        ?k = \labs(\tree)(bs), \val{V}\env = k, \text{ and } \sigma = \Pure(\tau)(bs)
           \ea}\\
      \arrive(bs, P) &\defas& \cek{\Return\;W \mid \env \mid (\nil, \hclo(P)) \cons \residual(bs, P)}, \quad \text{if } |bs| = n\\
      \multicolumn{3}{@{}l@{}}
      {\hfill
        \text{where } \Return\;W = \comp(\tree)(bs), \env = \envt(\tree)(bs), !b = \labs(\tree)(bs),  \text{ and } \val{W}\env = b}
    \end{equations}
  \item Correspondingly, the $\depart$ function computes the
    configuration either after the completion of a query or handling
    of an answer.
    \begin{equations}
      \depart &:& \Addr \times \ValCat \pto \Conf \\
      \depart(bs, P) &\defas& \cek{\Return\; m \mid \env \mid \residual(bs, P)}, \quad \text{if } |bs| < n\\
      \multicolumn{3}{@{}l@{}}
      {\hfill
        \text{where } \env = \ascend{\False}(bs, P) \text{ and } m = c(bs)}\\
      \depart(bs, P) &\defas& \cek{\Return\;m \mid \env \mid \residual(bs, P)}, \quad \text{if } |bs| = n\\
      \multicolumn{3}{@{}l@{}}
      {\hfill
        \text{where }\ba[t]{@{}l}
        m = c(bs),
        b = \begin{cases} \True & \text{if } m = 1\\
                          \False & \text{if } m = 0
            \end{cases}, \text{ and } \env = \initial(P)[x \mapsto b]
       \ea}
    \end{equations}
    The two clauses of $\depart$ yield slightly different
    configurations. The first clause computes a configuration inside
    the operation clause of $\HCount$. The configuration is exactly
    tail-configuration after summing up the two respective values
    returned by the two invocations of resumption. Whilst the second
    clause computes the tail-configuration inside of the success clause
    of $\HCount$ after handling a return value of the predicate.
    \item The $\residual$ function computes the residual continuation
    structure which contains the bits of computations to perform after
    handling a complete path in a decision tree.
    \begin{equations}
      \residual &:& \Addr \times \ValCat \pto \Cont\\
      \residual(bs, P) &\defas& [(\purecont(bs, P), \chi_{id})]
    \end{equations}
    \item The function $\purecont$ computes the pure continuation.
    \begin{equations}
      \purecont &:& \Addr \times \ValCat \pto \PureCont\\
      \purecont (\nil, P) &\defas& \nil\\
      \purecont (\snoc{bs}{\True}, P) &\defas& (\env, x_\True, \Let\;x_\False\revto r~\False\;\In\;x_\True+x_\False) \cons \purecont(bs, P),\\
      \multicolumn{3}{@{}l@{}}
      {\hfill
        \text{where } \env = \descend{\True}(\snoc{bs}{\True}, P)}\\
      \purecont (\snoc{bs}{\False}, P) &\defas& (\env, x_\False, x_\True+x_\False) \cons \purecont (bs, P),\\
      \multicolumn{3}{@{}l@{}}
      {\hfill
        \text{where } \env = \descend{\False}(\snoc{bs}{\False}, P)}\\
    \end{equations}
  \item The function $\initial$ computes the initial environment of
    the handler. The family of functions $\descend{b\in\mathbb{B}}$
    contains two functions, one for each instantiation of $b$, which
    describe how to compute the environment prior \emph{descending}
    down a branch as the result of invoking a resumption with
    $b$. Analogously, the functions in the family
    $\ascend{b \in \mathbb{B}}$ describe how to compute the
    environment after \emph{ascending} from the resumptive exploration
    of a branch.
    \begin{equations}
      \initial &:& \ValCat \to \Env\\
      \initial(P) &\defas& \emptyset[pred \mapsto \val{P}\emptyset]\\[1.5ex]
    \end{equations}
  \begin{minipage}{.5\linewidth}
  \begin{equations}
      \descend{\True} &:& \Addr \times \ValCat \pto \Env\\
      \descend{\True}(bs, P) &\defas& \initial(P)[r \mapsto (\sigma, \hclo(P))],\\
      \multicolumn{3}{@{}l@{}}
      {\qquad
        \text{where } \sigma = \Pure(\tree)(bs)}\\[1.5ex]
      \ascend{\True} &:& \Addr \times \ValCat \pto \Env\\
      \ascend{\True}(bs, P) &\defas& \env[x_\True \mapsto i],\\
      \multicolumn{3}{@{}l@{}}
      { \qquad\bl
          \text{where } \env = \descend{\True}(bs, P)\\
          \text{and } i = c(\snoc{bs}{\True})
        \el}\\[1.5ex]
    \end{equations}
  \end{minipage}
  \begin{minipage}{.5\linewidth}
    \begin{equations}
      \descend{\False} &:& \Addr \times \ValCat \pto \Env\\
      \descend{\False}(bs, P) &\defas& \ascend{\True}\\[1.5ex]
      \ascend{\False} &:& \Addr \times \ValCat \pto \Env\\
      \ascend{\False}(bs, P) &\defas& \env[x_\False \mapsto j],\\
      \multicolumn{3}{@{}l@{}}
      {\qquad\bl
        \text{where } \env = \descend{\False}(bs, P)\\
        \text{ and } j = c(\snoc{bs}{\False})
       \el}\\
    \end{equations}
  \end{minipage}
  \end{itemize}
\end{definition}
The proof of Theorem~\ref{thm:complexity-effectful-counting} works by
alternating between two different modes of reasoning: intensional and
extensional. The former is used to reason directly about the steps
taken by $\ECount$ program and the latter is used to reason about
steps taken by the provided predicate. The number of steps taken by an
$n$-standard predicate is readily available by constructing its
corresponding decorated timed decision tree model. The model is
constructed using a distinguished free variable $q$ to denote a
point. The following lemma lets us reason about the number of steps
taken by a predicate between its initial application and its first
query, between subsequent queries, and between final query and answer
when $q$ is instantiated to $\Superpoint$.
\begin{lemma}\label{lem:inductive-lem-aux}
  Suppose $P$ is an $n$-standard predicate, $bs \in \Addr$ is a list
  of booleans, and for all $\chi \in \HClosure$ and
  $\kappa \in \Cont$. Let $q$ denote the distinguished free variable
  used to construct the decorated timed decision tree $\tree$ of $P$.
  \begin{enumerate}
  \item If $|bs| = 0$ then
    \begin{derivation}
      &\cek{pred~q \mid \initial(P)[q \mapsto q] \mid (\nil, \chi) \cons \kappa}\\
       \stepsto&^{\steps(\tree)(\nil)}\\
      &\cek{z~V \mid \env[q \mapsto q] \mid (\sigma, \chi) \cons \kappa}
    \end{derivation}
    where $z~V = \comp(\tree)(\nil)$, $\env = \envt(\tree)(\nil)$,
    $?k = \labs(\tree)(\nil)$, $\val{V}\env = k$, $\env(z) = q$, and
    $\sigma = \Pure(\tree)(\nil)$; implies
    \begin{derivation}
      &\cek{pred~(\Superpoint) \mid \initial(P) \mid (\nil, \chi) \cons \kappa}\\
       \stepsto&^{\steps(\tree)(\nil)}\\
      &\cek{z~V \mid \env[z \mapsto (\initial(P), \Superpoint)] \mid (\sigma, \chi) \cons \kappa}
    \end{derivation}

  \item If $|bs| < n - 1$ then for all $b \in \mathbb{B}$ and $W \in \ValCat$
    \begin{derivation}
      &\cek{\Return\;W \mid \descend{b}(bs,P) \mid (\sigma, \chi) \cons \kappa}\\
       \stepsto&^{\steps(\tree)(\snoc{bs}{b})}\\
      &\cek{z~V \mid \env[q \mapsto q] \mid (\sigma', \chi) \cons \kappa}
    \end{derivation}
    where $\val{W}(\descend{b}(bs,P)) = b$, $\sigma = \Pure(\tree)(bs)$, $z~V = \comp(\tree)(\snoc{bs}{b})$, $\env = \envt(\tree)(\snoc{bs}{b})$, $\env(z) = q$, $?k = \labs(\tree)(\snoc{bs}{b})$, $\val{V}\env = k$, and $\sigma' = \Pure(\tree)(\snoc{bs}{b})$; implies
    \begin{derivation}
      &\cek{\Return\;W \mid \descend{b}(bs,P) \mid (\sigma, \chi) \cons \kappa}\\
       \stepsto&^{\steps(\tree)(\snoc{bs}{b})}\\
      &\cek{z~V \mid \env[z \mapsto (\initial(P), \Superpoint)] \mid (\sigma', \chi) \cons \kappa}
    \end{derivation}

  \item If $|bs| = n - 1$ then for all $b \in \mathbb{B}$ and $W \in \ValCat$
    \begin{derivation}
      &\cek{\Return\;W \mid \descend{b}(bs,P) \mid (\sigma, \chi) \cons \kappa}\\
       \stepsto&^{\steps(\tree)(\snoc{bs}{b})}\\
      &\cek{\Return\;W' \mid \env[q \mapsto q] \mid (\nil, \chi) \cons \kappa}
    \end{derivation}
    where $\val{W}(\descend{b}(bs,P)) = b$, $\sigma = \Pure(\tree)(bs)$, $\Return\;W' = \comp(\tree)(\snoc{bs}{b})$, $\env = \envt(\tree)(\snoc{bs}{b})$, $\ans b' = \labs(\tree)(\snoc{bs}{b})$, and $\val{W'}\env = b'$; implies
    \begin{derivation}
      &\cek{\Return\;W \mid \descend{b}(bs,P) \mid (\sigma, \chi) \cons \kappa}\\
       \stepsto&^{\steps(\tree)(\snoc{bs}{b})}\\
      &\cek{\Return\;W' \mid \env \mid (\nil, \chi) \cons \kappa}
    \end{derivation}
  \end{enumerate}
\end{lemma}
\begin{proof}
  By unfolding Definition~\ref{def:model-construction-extended}.
\end{proof}

Let $\ctrl : \Conf \pto \ValCat$ denote a partial function that hoists
a value out of a given machine configuration, that is
\[
  \ctrl(\cek{M \mid \env \mid \kappa})
    \defas
    \begin{cases}
      \val{V}{\env}  & \text{if } M = \Return\; V \\
      \bot           & \text{otherwise}
    \end{cases}
\]

\paragraph{Notation}
For a given predicate $P$ we write $\hclo(P)^\Val$ to mean
$\hclo(P)^\Val = (\emptyset[pred \mapsto \val{P}\emptyset],
\HCount)^\Val = \HCount^\Val$, that is the projection of the success
clause of $\HCount$.

The following lemma performs most of the heavy lifting for the proof
of Theorem~\ref{thm:complexity-effectful-counting}.
\begin{lemma}\label{lem:inductive-bit-of-thm1}
  Suppose $P$ is an $n$-standard predicate, then for any list of
  booleans $bs \in \Addr$ such that $|bs| \leq n$
  \[
    \ba{@{}l}
      \arrive(bs, P) \reducesto^{T(bs, n)} \depart(bs, P),
    \ea
  \]
  and $\ctrl(\depart(bs, P)) \leq 2^{n - |bs|}$ with the
  function $T$ defined as
  \[
    T(bs, n) =
     \begin{cases}
       9*(2^{n - |bs|} - 1) + 2^{n - |bs| + 1} + \sum_{bs' \in \Addr}^{1 \leq |bs'| \leq n - |bs|}\steps(\tree)(bs \concat bs') & \text{if } |bs| < n\\
       2            & \text{if } |bs| = n
     \end{cases}
   \]

\end{lemma}
\begin{proof}
  By downward induction on $bs$.
  \begin{description}
  \item[Base step] We have that $|bs| = n$. Since the predicate is
    $n$-standard we further have that $n \geq 1$. We proceed by direct
    calculation.
    \begin{derivation}
      &\arrive(bs, P)\\
      =& \reason{definition of $\arrive$ when $n = |bs|$}\\
      &\cek{\Return\;W \mid \env \mid (\nil, \hclo(P)) \cons \residual(bs, P)}\\
      \where{\ba[t]{@{~}l}
        \Return\;W = \comp(\tree)(bs), \env = \envt(\tree)(bs), !b = \labs(\tree)(bs), \text{ and } \val{W}\env = b
        \ea}\\
      \stepsto& \reason{\mlab{RetHandler}, $\hclo(P)^\Val = \{\Val~x \mapsto \cdots\}$}\\
      &\cek{\If\;x\;\Then\;\Return\;1\;\Else\;\Return\;0 \mid \env'[x \mapsto \val{b}\env'] \mid \residual(bs, P)}\\
      \where{\env' = \hclo(P).1}
    \end{derivation}
    The value $b$ can assume either of two values. We consider first
    the case $b = \True$.
    \begin{derivation}
      =& \reason{assumption $b = \True$, definition of $\val{-}$ (2 value steps)}\\
      &\cek{\If\;x\;\Then\;\Return\;1\;\Else\;\Return\;0 \mid \env'[x \mapsto \True] \mid \residual(bs, P)}\\
      \stepsto& \reason{\mlab{Case-inl} (and $\log|\env'[x \mapsto \True]| = 1$ environment operations)}\\
      &\cek{\Return\;1 \mid \env'[x \mapsto \True] \mid \residual(bs, P)}\\
      =& \reason{definition of $\depart$ when $n = |bs|$}\\
      &\depart(bs, P)
    \end{derivation}
    We have that $\ctrl(\depart(bs, P)) = 1 \leq 2^0 = 2^{n - |bs|}$.
    Next, we consider the case when $b = \False$.
    \begin{derivation}
      =& \reason{assumption $b = \False$, definition of $\val{-}$ (2 value steps)}\\
      &\cek{\If\;x\;\Then\;\Return\;1\;\Else\;\Return\;0 \mid \env'[x \mapsto \False] \mid \residual(bs, P)}\\
      \stepsto& \reason{\mlab{Case-Inl} (and $\log|\env'[x \mapsto \False]| = 1$ environment operations)}\\
      &\cek{\Return\;0 \mid \env'[x \mapsto \False] \mid \residual(bs, P)}\\
      =& \reason{definition of $\depart$ when $n = |bs|$}\\
      &\depart(bs, P)
    \end{derivation}
    Again, we have that $\ctrl(\depart(bs, P)) = 0 \leq 2^0 = 2^{n - |bs|}$.
    \paragraph{Step analysis}
    In either case, the machine uses exactly 2 transitions. Thus we get that
    \[
        2 = T(bs, n), \quad \text{when } |bs| = n
    \]
  \item[Inductive step] The induction hypothesis states that for all
    $b \in \mathbb{B}$ and $|bs| < n$
    \[
      \arrive(\snoc{bs}{b}, P) \reducesto^{T(\snoc{bs}{b}, n)} \depart(\snoc{bs}{b}, P),
    \]
    such that $\ctrl(\depart(\snoc{bs}{b}, P)) \leq 2^{n - |\snoc{bs}{b}|}$.
    We proceed by direct calculation.
    \begin{derivation}
      &\arrive(bs, P)\\
      =& \reason{definition of $\arrive$ when $n < |bs|$}\\
      &\cek{z~V \mid \env \mid (\sigma, \hclo(P)) \cons \residual(bs, P)}\\
      \where{\ba[t]{@{~}l}
           z~V = \comp(\tree)(bs), \env = \dec{env}(\tree)(bs)[z \mapsto (\initial(P), \Superpoint)],\\
          ?k = \labs(\tree)(bs), \val{V}\env = k, \text{ and } \sigma = \Pure(\tree)(bs)
        \ea}\\
      \stepsto& \reason{\mlab{App}}\\
      &\cek{\Do\;\Branch~\Unit \mid \env'[\_ \mapsto k] \mid (\sigma, \hclo(P)) \cons \residual(bs, P)}\\
      \where{\env' = \initial(P)}\\
      \stepsto& \reason{\mlab{Handle-Op}, $\hclo(P)^{\Branch} = \{\Branch~\Unit~r \mapsto \cdots\}$}\\
      &\left\langle
        \ba[m]{@{}l}
          \Let\;x_{\True} \revto r~\True \;\In\\
          \Let\;x_{\False} \revto r~\False \;\In\\
          x_\True + x_\False
        \ea
       \mid \env[r \mapsto \val{(\sigma, \hclo(P))}\env] \mid \residual(bs, P)
     \right\rangle\\
     \where{\env = \initial(P)}\\
     =& \reason{definition of $\val{-}$ (1 value step)}\\
     &\left\langle
        \ba[m]{@{}l}
          \Let\;x_{\True} \revto r~\True \;\In\\
          \Let\;x_{\False} \revto r~\False \;\In\\
          x_\True + x_\False
        \ea
       \mid \env' \mid \residual(bs, P)
     \right\rangle\\
     \where{\env' = \env[r \mapsto (\sigma, \hclo(P))]}\\
    \stepsto& \reason{\mlab{Let}, definition of $\residual$}\\
    &\cek{r~\True \mid \env' \mid \residual(\snoc{bs}{\True} bs, P)}\\
    \stepsto& \reason{\mlab{Resume}, $\val{r}\env' = (\sigma, \hclo(P))$ ($\log|\env'| = 1$ environment operations)}\\
    &\cek{\Return\;\True \mid \env' \mid (\sigma, \hclo(P)) \cons \residual(\snoc{bs}{\True}, P)}
    \end{derivation}
    We now use Lemma~\ref{lem:inductive-lem-aux} to reason about the
    progress of the predicate computation $\sigma$. There are two
    cases consider, either $1 + |bs| < n$ or
    $1 + |bs| = n$.
    \begin{description}
    \item[Case] $1 + |bs| < n$. We obtain the following internal node
      configuration.
    \begin{derivation}
      \stepsto&^{\steps(\tree)(\snoc{bs}{\True})} \reason{by Lemma~\ref{lem:inductive-lem-aux}}\\
      &\cek{z~V \mid \env'' \mid (\sigma', \hclo(P)) \cons \residual(\snoc{bs}{\True}, P)}\\
      \where{\ba[t]{@{~}l}
        z~V = \comp(\tree)(bs), \env'' = \dec{env}(\tree)(\snoc{bs}{\True})[z \mapsto (\initial(P), \Superpoint)],\\
        ?k = \labs(\tree)(\snoc{bs}{\True}), \val{V}\env'' = k, \text{ and } \sigma' = \Pure(\tree)(\snoc{bs}{\True})
       \ea}\\
      =& \reason{definition of $\arrive$ when $1 + |bs| < n$}\\
      &\arrive(\snoc{bs}{\True}, P)\\
      \stepsto&^{T(\snoc{bs}{\True}, n)} \reason{induction hypothesis}\\
      &\depart(\snoc{bs}{\True}, P)\\
      =& \reason{definition of $\depart$ when $1 + |bs| < n$}\\
      &\cek{\Return\;i \mid \env \mid \residual(\snoc{bs}{\True}, P)}\\
      \where{i = c(\snoc{\snoc{bs}{\True}}{\True}) + c(\snoc{\snoc{bs}{\True}}{\False}) \text{ and } \env = \ascend{\False}(\snoc{bs}{\True}, P)}\\
      =& \reason{definition of $\residual$ and $\purecont$}\\
      &\cek{\Return\;i \mid \env \mid [((\env', x_\True, \Let\;x_\False\revto r~\False\;\In\;x_\True+x_\False) \cons \purecont(bs, P), \chi_{id})]}\\
      \where{\env' = \descend{\True}(bs, P)}\\
      \stepsto& \reason{\mlab{RetCont}}\\
      &\cek{\Let\;x_\False\revto r~\False\;\In\;x_\True+x_\False \mid \env'' \mid [(\purecont(bs, P), \chi_{id})]}\\
      \where{\env'' = \env'[x_\True \mapsto \val{i}\env']}\\
      \stepsto& \reason{\mlab{Let}}\\
      &\cek{r~\False \mid \env'' \mid [((\env'', x_\False, x_\True + x_\False) \cons \purecont(bs, P), \chi_{id})]}
    \end{derivation}
    \begin{derivation}
      =& \reason{definition of $\purecont$ and $\residual$}\\
      &\cek{r~\False \mid \env'' \mid \residual(\snoc{bs}{\False}, P)}\\
      \stepsto& \reason{\mlab{Resume}}\\
      &\cek{\Return\;\False \mid \env'' \mid (\sigma, \hclo(P)) \cons \residual(\snoc{bs}{\False}, P)}\\
      \where{\sigma = \Pure(\tree)(bs)}\\
      \stepsto&^{\steps(\tree)(\snoc{bs}{\False})} \reason{by Lemma~\ref{lem:inductive-lem-aux}}\\
      &\cek{z~V \mid \env \mid (\sigma, \hclo(P)) \cons \residual(\snoc{bs}{\False}, P)}\\
      \where{\ba[t]{@{~}l}
        z~V = \comp(\tree)(bs), \env = \envt(\tree)(\snoc{bs}{\False})[q \mapsto (\initial(P), \Superpoint)],\\
        ?k = \labs(\tree)(\snoc{bs}{\False}), \val{V}\env = k, \text{ and } \sigma = \Pure(\tree)(\snoc{bs}{\False})
        \ea}\\
      =& \reason{definition of $\arrive$ when $1 + |bs| < n$}\\
      &\arrive(\snoc{bs}{\False}, P)\\
      \stepsto&^{T(\snoc{bs}{\False}, n)} \reason{induction hypothesis}\\
      &\depart(\snoc{bs}{\False}, P)\\
      =& \reason{definition of $\depart$ when $1 + |bs| < n$}\\
      &\cek{\Return\;j \mid \env \mid \residual(\snoc{bs}{\False}, P)}\\
      \where{j = c(\snoc{\snoc{bs}{\False}}{\True}) + c(\snoc{\snoc{bs}{\False}}{\False}) \text{ and } \env = \ascend{\False}(\snoc{bs}{\False}, P)}\\
      =& \reason{definition of $\residual$ and $\purecont$}\\
      &\cek{\Return\;j \mid \env \mid [((\env'', x_\False, x_\True + x_\False) \cons \purecont(bs, P), \chi_{id})]}\\
      \stepsto& \reason{\mlab{RetCont}}\\
      &\cek{x_\True + x_\False \mid \env''[x_\False \mapsto \val{j}\env''] \mid \residual(bs, P)}\\
      \stepsto& \reason{\mlab{Plus}}\\
      &\cek{\Return\;m \mid \env''[x_\False \mapsto \val{j}\env''] \mid \residual(bs, P)}\\
      &\text{where}
        \ba[t]{@{~}l@{~}l}
      m &= c(\snoc{\snoc{bs}{\True}}{\True}) + c(\snoc{\snoc{bs}{\True}}{\False}) \\
        &+ c(\snoc{\snoc{bs}{\False}}{\True}) + c(\snoc{\snoc{bs}{\False}}{\False})\\
        &= c(\snoc{bs}{\True}) + c(\snoc{bs}{\False}) = c(bs) \leq 2^{n - |bs|}
        \ea
      \\
    =& \reason{definition of $\depart$ when $|bs| < n$}\\
    &\depart(bs, P)
  \end{derivation}
  \paragraph{Step analysis} The total number of machine steps is
  given by
    \begin{derivation}
      &9 + \steps(\tree)(\snoc{bs}{\True}) + T(\snoc{bs}{\True}, n) + \steps(\tree)(\snoc{bs}{\False}) + T(\snoc{bs}{\False}, n)\\
      =& \reason{reorder}\\
      &9 + T(\snoc{bs}{\True}, n) + \steps(\tree)(\snoc{bs}{\False}) + \steps(\tree)(\snoc{bs}{\True}) + \steps(\tree)(\snoc{bs}{\False})\\
      =& \reason{definition of $T$}\\
      & 9 + 9*(2^{n - |\snoc{bs}{\True}|} - 1) + 9*(2^{n - |\snoc{bs}{\False}|} - 1) + 2^{n - |\snoc{bs}{\True}| + 1} + 2^{n - |\snoc{bs}{\False}| + 1}\\
      &+ \displaystyle\sum_{bs' \in \Addr}^{1 \leq |bs'| \leq n - |\snoc{bs}{\True}|}\steps(\tree)(\snoc{bs}{\True} \concat bs')
      + \displaystyle\sum_{bs' \in \Addr}^{1 \leq |bs'| \leq n - |\snoc{bs}{\False}|}\steps(\tree)(\snoc{bs}{\False} \concat bs')\\
      &+\steps(\tree)(\snoc{bs}{\True}) + \steps(\tree)(\snoc{bs}{\False})\\
      =& \reason{simplify}\\
       & 9 + 9*(2^{n - |\snoc{bs}{\True}|} - 1) + 9*(2^{n - |\snoc{bs}{\False}|} - 1) + 2^{n - |bs| + 1}\\
      &+ \displaystyle\sum_{bs' \in \Addr}^{1 \leq |bs'| \leq n - |\snoc{bs}{\True}|}\steps(\tree)(\snoc{bs}{\True} \concat bs')
      + \displaystyle\sum_{bs' \in \Addr}^{1 \leq |bs'| \leq n - |\snoc{bs}{\False}|}\steps(\tree)(\snoc{bs}{\False} \concat bs')\\
      &+\steps(\tree)(\snoc{bs}{\True}) + \steps(\tree)(\snoc{bs}{\False})\\
    \end{derivation}
    \begin{derivation}
      =& \reason{merge sums}\\
       & 9 + 9*(2^{n - |\snoc{bs}{\True}|} - 1) + 9*(2^{n - |\snoc{bs}{\False}|} - 1) + 2^{n - |bs| + 1}\\
       &+ \left(\displaystyle\sum_{bs' \in \Addr}^{2 \leq |bs'| \leq n - |bs|}\steps(\tree)(bs \concat bs')\right)
       + \steps(\tree)(\snoc{bs}{\True}) + \steps(\tree)(\snoc{bs}{\False})\\
      =& \reason{rewrite binary sum}\\
      &9 + 9*(2^{n - |\snoc{bs}{\True}|} - 1) + 9*(2^{n - |\snoc{bs}{\False}|} - 1) + 2^{n - |bs| + 1}\\
      &+ \displaystyle\sum_{bs' \in \Addr}^{2 \leq |bs'| \leq n - |bs|}\steps(\tree)(bs \concat bs')
      + \displaystyle\sum_{bs' \in \Addr}^{1 \leq |bs'| \leq 1}\steps(\tree)(bs \concat bs')\\
      =& \reason{merge sums}\\
      &9 + 9*(2^{n - |\snoc{bs}{\True}|} - 1) + 9*(2^{n - |\snoc{bs}{\True}|} - 1) + 2^{n - |bs| + 1} + \displaystyle\sum_{bs' \in \Addr}^{1 \leq |bs'| \leq n - |bs|}\hspace{-0.5cm}\steps(\tree)(bs \concat bs')\\
      =& \reason{factoring}\\
       &9 + 2*9*(2^{n - |bs| - 1} - 1) + 2^{n - |bs| + 1} + \displaystyle\sum_{bs' \in \Addr}^{1 \leq |bs'| \leq n - |bs|}\steps(\tree)(bs \concat bs')\\
      =& \reason{distribute}\\
       &9 + 9*(2^{n - |bs|} - 2) + 2^{n - |bs| + 1} + \displaystyle\sum_{bs' \in \Addr}^{1 \leq |bs'| \leq n - |bs|}\steps(\tree)(bs \concat bs')\\
      =& \reason{distribute}\\
       &9 + 9*2^{n - |bs|} - 18 + 2^{n - |bs| + 1} + \displaystyle\sum_{bs' \in \Addr}^{1 \leq |bs'| \leq n - |bs|}\steps(\tree)(bs \concat bs')\\
      =& \reason{simplify}\\
       &9*2^{n - |bs|} - 9 + 2^{n - |bs| + 1} + \displaystyle\sum_{bs' \in \Addr}^{1 \leq |bs'| \leq n - |bs|}\steps(\tree)(bs \concat bs')\\
      =& \reason{factoring}\\
       &9*(2^{n - |bs|} - 1) + 2^{n - |bs| + 1} + \displaystyle\sum_{bs' \in \Addr}^{1 \leq |bs'| \leq n - |bs|}\steps(\tree)(bs \concat bs')\\
      =& \reason{definition of $T$}\\
      &T(bs, n)
    \end{derivation}
    \item[Case] $1 + |bs| = n$. We obtain the following
      configuration.
    \begin{derivation}
    \stepsto&^{\steps(\tree)(\snoc{bs}{\True})} \reason{by Lemma~\ref{lem:inductive-lem-aux}}\\
      &\cek{\Return\;W \mid \env'' \mid (\nil, \hclo(P)) \cons \residual(\snoc{bs}{\True}, P)}\\
      \where{\ba[t]{@{~}l}
        \Return\;W = \comp(\tree)(\snoc{s}{\True}), !b = \labs(\tree)(\snoc{bs}{\True}),\\
        \env'' = \envt(\tree)(\snoc{bs}{\True}), \text{ and } \val{W}\env'' = b
        \ea}\\
      =& \reason{definition of $\arrive$ when $1 + |bs| = n$}\\
      &\arrive(\snoc{bs}{\True}, P)\\
      \stepsto&^{T(\snoc{bs}{\True}, n)} \reason{induction hypothesis}\\
      &\depart(\snoc{bs}{\True}, P)\\
      =& \reason{definition of $\depart$ when $1 + |bs| = n$}\\
      &\cek{\Return\;i \mid \env \mid \residual(\snoc{bs}{\True}, P)}\\
      \where{i = c(\snoc{bs}{\True}) \leq 2^{n - |\snoc{bs}{\True}|} = 1 \text{ and } \env = \initial(P)}\\
      =& \reason{definition of $\residual$ and $\purecont$}\\
      &\cek{\Return\;i \mid \env \mid [((\env',x_\True,\Let\;x_\False\revto r~\False\;\In\;x_\True+x_\False) \cons \purecont(bs, P), \chi_{id})]}\\
      \stepsto& \reason{\mlab{RetCont}}\\
      &\cek{\Let\;x_\False \revto r~\False\;\In\;x_\True + x_\False \mid \env'[x_\True \mapsto \val{i}\env'] \mid [(\purecont(bs, P), \chi_{id})]}\\
      =& \reason{definition of $\val{-}$ (1 value step)}\\
      &\cek{\Let\;x_\False \revto r~\False\;\In\;x_\True + x_\False \mid \env'' \mid [(\purecont(bs, P), \chi_{id})]}\\
      \where{\env'' = \env'[x_\True \mapsto i]}\\
      \stepsto& \reason{\mlab{Let}, definition of $\residual$}\\
      &\cek{r~\False \mid \env'' \mid \residual(\snoc{bs}{\False}, P)}\\
      \stepsto& \reason{\mlab{Resume}}\\
      &\cek{\Return\;\False \mid \env'' \mid (\sigma, \hclo(P)) \cons \residual(\snoc{bs}{\False}, P)}\\
      \where{\sigma = \Pure(\tree)(bs)}\\
      \stepsto&^{\steps(\tree)(\snoc{bs}{\False})} \reason{by Lemma~\ref{lem:inductive-lem-aux}}\\
      &\cek{\Return\;W \mid \env \mid (\nil, \hclo(P)) \cons \residual(\snoc{bs}{\False}, P)}\\
      \where{\ba[t]{@{~}l}
        \Return\;W = \comp(\tree)(\snoc{bs}{\False}), !b = \labs(\tree)(\snoc{bs}{\False}),\\
        \env = \envt(\tree)(\snoc{bs}{\False}), \text{ and } \val{W}\env = b
        \ea}\\
      =& \reason{definition of $\arrive$ when $1 + |bs| = n$}\\
      &\arrive(\snoc{bs}{\False}, P)\\
      \stepsto&^{T(\snoc{bs}{\False}, n)} \reason{induction hypothesis}\\
      &\depart(\snoc{bs}{\False}, P)\\
      =& \reason{definition of $\depart$ when $1 + |bs| = n$}\\
      &\cek{\Return\;j \mid \env \mid \residual(\snoc{bs}{\False}, P)}\\
      \where{j = c(\snoc{bs}{\False}) \leq 2^{n - |\snoc{bs}{\False}|} = 1 \text{ and } \env = \initial(P)}\\
     =& \reason{definition of $\residual$ and $\purecont$}\\
      &\cek{\Return\;j \mid \env \mid [((\env',x_\False,x_\True+x_\False) \cons \purecont(bs, P), \chi_{id})]}\\
      \where{\env' = \descend{\False}(bs, P)}\\
      \stepsto& \reason{\mlab{RetCont}}\\
      &\cek{x_\True + x_\False \mid \env'' \mid [(\purecont(bs, P), \chi_{id})]}\\
      \where{\env'' = \env'[x_\False \mapsto \val{j}\env'] = \env'[x_\False \mapsto j]}\\
      \stepsto& \reason{\mlab{Plus}}\\
      &\cek{\Return\;m \mid \env'' \mid [(\purecont(bs, P), \chi_{id})]}\\
      \where{m = c(\snoc{bs}{\True}) + c(\snoc{bs}{\False}) \leq 2^{n - |bs|}}\\
      =& \reason{definition of $\residual$ and $\depart$ when $|bs| < n$}\\
      &\depart(bs, P)
    \end{derivation}
    \paragraph{Step analysis} The total number of machine steps
    is given by
    \begin{derivation}
      &9 + \steps(\tree)(\snoc{bs}{\True}) + T(\snoc{bs}{\True}, n) + \steps(\tree)(\snoc{bs}{\False}) + T(\snoc{bs}{\False}, n)\\
      =& \reason{reorder}\\
      &9 + T(\snoc{bs}{\True}, n) + T(\snoc{bs}{\False}, n) + \steps(\tree)(\snoc{bs}{\True}) + \steps(\tree)(\snoc{bs}{\False})\\
      =& \reason{definition of $T$ when $|bs| + 1 = n$}\\
      &9 + 2 + 2 + \steps(\tree)(\snoc{bs}{\True}) + \steps(\tree)(\snoc{bs}{\False})\\
      =& \reason{simplify}\\
      &9 + 2^2 + \steps(\tree)(\snoc{bs}{\True}) + \steps(\tree)(\snoc{bs}{\False})\\
      =& \reason{rewrite $2 = n - |bs| + 1$}\\
      &9 + 2^{n - |bs| + 1} + \steps(\tree)(\snoc{bs}{\True}) + \steps(\tree)(\snoc{bs}{\False})\\
      =& \reason{multiply by $1$}\\
      &9*(2^{n - |bs|} - 1) + 2^{n - |bs| + 1} + \steps(\tree)(\snoc{bs}{\True}) + \steps(\tree)(\snoc{bs}{\False})\\
      =& \reason{rewrite binary sum}\\
      &9*(2^{n - |bs|} - 1) + 2^{n - |bs|} + \displaystyle\sum_{bs' \in \Addr}^{1 \leq |bs'| \leq n - |bs|} \steps(\tree)(bs \concat bs')\\
      =& \reason{definition of $T$}\\
      &T(bs, n)
    \end{derivation}
    \end{description}
  \end{description}
\end{proof}

The following theorem is a copy of
Theorem~\ref{thm:complexity-effectful-counting}.

\begin{theorem}\label{thm:complexity-effectful-counting-copy}
  For all $n > 0$ and any $n$-standard predicate $P$ it holds that
  \begin{enumerate}
  \item The program $\ECount$ is a generic count program
  \item The runtime complexity of $\ECount~P$ is given by the following formula:
  \[
    \displaystyle\sum_{bs \in \Addr}^{|bs| \leq n} \steps(\tr(P))(bs) + \BigO(2^n)
  \]
\end{enumerate}
\end{theorem}
\begin{proof}
  The proof begins by direct calculation.
\begin{derivation}
  &\cek{\ECount\,P \mid \emptyset \mid [(\nil, \chi_{id})]} \\
  =& \reason{definition of $\residual$}\\
  & \cek{\ECount\,P \mid \emptyset \mid \residual(\nil, P)} \\
  \stepsto& \reason{\mlab{App}, $\val{\ECount}\emptyset = (\emptyset, \lambda pred. \cdots)$}\\
  & \cek{\Handle\;pred~(\Superpoint)\;\With\;\HCount \mid \env \mid \residual(\nil, P)}\\
  \where{\env = \initial(P)}\\
  \stepsto& \reason{\mlab{Handle}}\\
  & \cek{pred~(\Superpoint) \mid \env \mid (\nil, (\env, \HCount)) \cons \residual(\nil, P)}\\
  =& \reason{definition of $\hclo$}\\
  &\cek{pred~(\Superpoint) \mid \env \mid (\nil, \hclo(P)) \cons \residual(\nil, P)}\\
  \stepsto&^{\steps(\tree)(\nil)} \reason{by Lemma~\ref{lem:inductive-lem-aux}}\\
  &\cek{z~V \mid \env' \mid (\sigma, \hclo(P)) \cons \residual(\nil, P)}\\
  \where{\ba[t]{@{~}l}
      z~V = \comp(\tree)(bs), \env' = \envt(\tree)(\nil)[q \mapsto (\initial(P), \Superpoint)],\\
      ?k = \labs(\tree)(\nil), \val{V}\env' = k, \text{ and } \sigma = \Pure(\tree)(\nil)
      \ea}\\
  \end{derivation}
  \begin{derivation}
  =& \reason{definition of $\arrive$}\\
  &\arrive(\nil, P)\\
  \stepsto&^{T(\nil, n)} \reason{by Lemma~\ref{lem:inductive-bit-of-thm1}}\\
  &\depart(\nil, P)\\
  =& \reason{definition of $\depart$}\\
  &\cek{\Return\;m \mid \env \mid \residual(\nil, P)}\\
  \where{\env = \initial(P) \text{ and }  m = c(\nil) \leq 2^{n - |bs|} = 2^n}\\
  =& \reason{definition of $\residual$}\\
  &\cek{\Return\;m \mid \env \mid [(\nil, \chi_{id})]}\\
  \stepsto& \reason{\mlab{Handle-Ret}, $H_{id}^{\dec{val}} = \{\Val~x \mapsto \Return\;x\}$}\\
  &\cek{\Return\;x \mid \emptyset[x \mapsto m] \mid \nil}
\end{derivation}
\paragraph{Analysis}
The machine yields the value $m$.
By Lemma~\ref{lem:inductive-bit-of-thm1} it follows that
$m \leq 2^{n - |bs|} = 2^{n - |\nil|} = 2^n$. Furthermore, the total
number of transitions used were
\begin{derivation}
    &3 + \steps(\tree)(\nil) + T(\nil, n)\\
    =& \reason{definition of $T$}\\
    &3 + \steps(\tree)(\nil) + 9*2^n + 2^{n + 1} + \displaystyle\sum_{bs' \in \mathbb{B}^{\ast}}^{1 \leq |bs'| \leq n}\steps(\tree)(bs')\\
    =& \reason{simplify}\\
    &3 + \steps(\tree)(\nil) + 9*2^n + 2^{n + 1} + \displaystyle\sum_{bs' \in \mathbb{B}^{\ast}}^{1 \leq |bs'| \leq n}\steps(\tree)(bs')\\
    =& \reason{reorder}\\
    &3 + \left(\displaystyle\sum_{bs' \in \mathbb{B}^{\ast}}^{1 \leq |bs'| \leq n}\steps(\tree)(bs')\right) + \steps(\tree)(\nil) + 9*2^n + 2^{n + 1}\\
    =& \reason{rewrite as unary sum}\\
    &3 + \left(\displaystyle\sum_{bs' \in \mathbb{B}^{\ast}}^{1 \leq |bs'| \leq n}\steps(\tree)(bs') + \displaystyle\sum_{bs' \in \Addr}^{0 \leq |bs'| \leq 0}\steps(\tree)(bs')\right) + 9*2^n + 2^{n + 1}\\
    =& \reason{merge sums}\\
    &3 + \left(\displaystyle\sum_{bs' \in \mathbb{B}^{\ast}}^{0 \leq |bs'| \leq n}\steps(\tree)(bs')\right) + 9*2^n + 2^{n + 1}\\
    =& \reason{definition of $\BigO$}\\
    &\left(\displaystyle\sum_{bs' \in \mathbb{B}^{\ast}}^{0 \leq |bs'| \leq n}\steps(\tree)(bs')\right) + \BigO(2^{n})
\end{derivation}
\end{proof}
\fi

\extendedsection{Berger Count}
\label{sec:berger-count}
\ifdefined\EXTENDED
Here we present the $\BergerCount$ program alluded to in
Section~\ref{sec:pure-counting}, in order to fill out our overall
picture of the relationship between language expressivity and
potential program efficiency.

Berger's original program~\citep{Berger90} introduced a remarkable
search operator for predicates on \emph{infinite} streams of booleans,
and has played an important role in higher-order computability
theory~\citep{LongleyN15}.  What we wish to highlight here is that if
one applies the algorithm to predicates on \emph{finite} boolean
vectors, the resulting program, though no longer interesting from a
computability perspective, still holds some interest from a complexity
standpoint: indeed, it yields what seems to be the best available
implementation of generic count within a PCF-style `functional'
language (provided one accepts the use of a primitive for call-by-need
evaluation).

We give the gist of an adaptation of Berger's search algorithm on
finite spaces.
{\small
\[
    \bl
    \bestshot_n: \Predicate_n \to \Point_n\\
    \bestshot_n~pred \defas \bestshot'_n~pred~\nil \medskip\\

    \bestshot'_n : \Predicate_n \to \List_\Bool \to \Point_n\\
    \bestshot'_n~pred~start \defas
     \ba[t]{@{}l}
       \Let\; f \revto \dec{memoise}~(\lambda\Unit. \bestshot''_n~pred~start)\; \In\\
       \Return\;(\lambda i. \If\; i < |start| \;\Then\; start.i \;\Else\; (f~\Unit).i)
     \ea\medskip\\

    \bestshot''_n : \Predicate_n \to \List_\Bool \to \List_\Bool\\
    \bestshot''_n~pred~start \defas
     \ba[t]{@{}l}
       \If\; |start| = n \;\Then\; \Return\; start\\
       \Else\;
       \ba[t]{@{}l}
         \Let\; f \revto \bestshot'_n~pred~(\dec{append}~start~[\True])\;\In\\
         \If\; pred~f \;\Then\; \Return\; [f~0,\dots,f~(n-1)]\\
         \Else\; \bestshot''_n~pred~(\dec{append}~start~[\False])
       \ea
     \ea
   \el
\]}%
Given any $n$-standard predicate $P$ the function $\bestshot_n$
returns a point satisfying $P$ if one exists, or dummy point
$\lambda i.\False$ if not. It is implemented by via two mutually
recursive auxiliary functions whose workings are admittedly hard to
elucidate in a few words. The function $\bestshot'_n$ is a
generalisation of $\bestshot_n$ that makes a best shot at finding a
point $\pi$ satisfying given predicate and matching some specified
list $start$ in some initial segment of its components
$[\pi(0),\dots,\pi(i-1)]$. It works `lazily', drawing its values from
$start$ wherever possible, and performing an actual search only when
required. This actual search is undertaken by $\bestshot''_n$, which
proceeds by first searching for a solution that extends the specified
list with true; but if no such solution is forthcoming, it settles for
false as the next component of the point being constructed. The whole
procedure relies on a subtle combination of laziness, recursion and
implicit nesting of calls to the provided predicate which means that
the search is self-pruning in regions of the binary tree where the
predicate only demands some initial segment $q~0$,\dots,$q~(i-1)$ of
its argument $q$.

The above program makes use of an operation
{\small
\[
    \dec{memoise} : (\One \to \dec{List}~\Bool) \to (\One \to \dec{List}~\Bool)
\]}%
which transforms a given thunk into an equivalent `memoised' version,
i.e. one that caches its value after its first invocation and
immediately returns this value on all subsequent invocations. Such an
operation may readily be implemented in $\BCalcS$, or alternatively
may simply be added as a primitive in its own right.
The latter has the advantage that it preserves the purely `functional'
character of the language, in the sense that every program is
observationally equivalent to a $\BCalc$ program, namely the one
obtained by replacing $\dec{memoise}$ by the identity.

We now show how the above idea may be exploited to yield a generic
count program (this development appears to be new).
{\small\[
  \bl
    \BergerCount_n : \Predicate_n \to \Nat\\
    \BergerCount_n~pred \defas \Count'_n~pred~[]~0 \medskip\\

    \Count'_n : \Predicate_n \to \List_\Bool \to \Nat \to \Nat\\
    \Count'_n~pred~start~acc \defas
      \ba[t]{@{}l}
         \If\; |start| = n\; \Then\; acc +
                   (\If\; pred\,(\lambda i. start.i) \;\Then\;\Return\;1 \;\Else\;\Return\;0)\\
         \Else\;
           \ba[t]{@{}l}
             \Let\; f \revto \bestshot'_n~pred~start\; \In\\
             \If\; pred~f \;\Then\; \Count''_n~start~ [f~0,\dots,f~(n-1)]~acc \;\Else\; \Return\;acc
           \ea
     \ea \medskip\\

    \Count''_n : \Predicate_n \to \List_\Bool \to \List_\Bool \to \Nat \to \Nat\\
    \Count''_n~pred~start~leftmost~acc \defas
      \ba[t]{@{}l}
         \If\; |start| = n \;\Then\; acc+1\\
         \Else\;
           \ba[t]{@{}l}
             \Let\; b \revto leftmost.|start|\; \In\\
             \Let\; acc' \revto \Count''_n~pred~(\dec{append}~start~[b])~leftmost~acc\; \In\\
             \If\; b \; \Then\; \Count'_n~pred~(\dec{append}~start~[\False])~acc' \;\Else~\Return\;acc'
           \ea
      \ea
  \el
\]}%
Again, $\BergerCount_n$ is implemented by means of two mutually
recursive auxiliary functions.  The function $\Count'_n$ counts the
solutions to the provided predicate $pred$ that start with the
specified list of booleans, adding their number to a previously
accumulated total given by $acc$.  The function $\Count''_n$ does the
same thing, but exploiting the knowledge that a best shot at the
`leftmost' solution to $P$ within this subtree has already been
computed.  (We are visualising $n$-points as forming a binary tree
with $\True$ to the left of $\False$ at each fork.) Thus, $\Count''_n$
will not re-examine the portion of the subtree to the left of this
candidate solution, but rather will start at this solution and work
rightward.

This gives rise to an $n$-count program that can work efficiently on
predicates that tend to `fail fast': more specifically, predicates $P$
that inspect the components of their argument $q$ in order $q~0$,
$q~1$, $q~2$, \dots, and which are frequently able to return $\False$
after inspecting just a small number of these components. Generalising
our program from binary to $k$-ary branching trees, we see that the
$n$-queens problem provides a typical example: most points in the
space can be seen \emph{not} to be solutions by inspecting just the
first few components. Our experimental results in
Section~\ref{sec:experiments} attest to the viability of this approach
and its overwhelming superiority over the \naive functional method.

By contrast, the above program is \emph{not} able to exploit parts of
the tree where our predicate `succeeds fast', i.e.\ returns $\True$
after seeing just a few components.  Unlike the effectful count
program of Section~\ref{sec:effectful-counting}, which may sometimes
add $2^{n-d}$ to the count in a single step, the Berger approach can
only count solutions one at a time. Thus, supposing $P$ is an
$n$-standard predicate the evaluation of $\Count_n~P$ that returns a
natural number $c$ must take time $\Omega(c)$. These observations
informally indicate the likely extent of the efficiency gap between
effectful and purely functional computation when it comes to
non-$n$-standard predicates.
\fi

\end{document}